\documentclass[12pt,journal,draftcls,onecolumn]{IEEEtran}
\usepackage{amsmath,amssymb,mathrsfs}
\usepackage{epsfig,epsf,subfigure,graphicx,graphics}
\usepackage{url}

\newtheorem{lemma}{Lemma}[section]
\newtheorem{theorem}[lemma]{Theorem}

\newtheorem{corollary}[lemma]{Corollary}
\newtheorem{definition}[lemma]{Definition}
\newtheorem{remark}[lemma]{Remark}
\newtheorem{claim}[lemma]{Claim}

\newcommand{\SNR}{\mathsf{SNR}}
\newcommand{\INR}{\mathsf{INR}}
\newcommand{\C}{\mathsf{C}^\mathsf{B}}
\newcommand{\E}{\mathrm{E}}

\newcommand{\sym}{\mathrm{sym}}

\newcommand{\STG}{\mathrm{STG}}

\newcommand{\OR}{\mathrm{OneRound}}
\newcommand{\conv}{\mathrm{conv}}

\newcommand{\aaaa}{\mathrm{(a)}}
\newcommand{\bbbb}{\mathrm{(b)}}
\newcommand{\cccc}{\mathrm{(c)}}
\newcommand{\dddd}{\mathrm{(d)}}
\newcommand{\eeee}{\mathrm{(e)}}
\newcommand{\ffff}{\mathrm{(f)}}

\newcommand{\lp}{\left(}
\newcommand{\rp}{\right)}

\newcommand{\lbp}{\left\{}
\newcommand{\rbp}{\right\}}
\newcommand{\ul}{\underline}
\newcommand{\mcal}{\mathcal}
\newcommand{\mscr}{\mathscr}
\newcommand{\what}{\widehat}
\newcommand{\wtild}{\widetilde}

\newcommand{\ra}{\rightarrow}

\title{Interference Mitigation Through Limited Receiver Cooperation}

\author{
\authorblockN{I-Hsiang Wang and David N. C. Tse}\\
\authorblockA{Wireless Foundations\\
University of California at Berkeley,\\
Berkeley, California 94720, USA\\
\textsf{\{ihsiang, dtse\}@eecs.berkeley.edu}}
\thanks{This work was supported by National Science Foundation under grant \# CCF-0830796.}
}

\begin{document}
\maketitle

\begin{abstract}

Interference is a major issue limiting the performance in wireless networks. Cooperation among receivers can help mitigate interference by forming distributed MIMO systems. The rate at which receivers cooperate, however, is limited in most scenarios. How much interference can one bit of receiver cooperation mitigate? In this paper, we study the two-user Gaussian interference channel with conferencing decoders to answer this question in a simple setting. We identify two regions regarding the gain from receiver cooperation: linear and saturation regions. In the linear region receiver cooperation is efficient and provides a \emph{degrees-of-freedom} gain, which is either \emph{one cooperation bit buys one more bit} or \emph{two cooperation bits buy one more bit} until saturation. In the saturation region receiver cooperation is inefficient and provides a \emph{power} gain, which is at most a constant regardless of the rate at which receivers cooperate.
The conclusion is drawn from the characterization of capacity region to within two bits. The proposed strategy consists of two parts: (1) the transmission scheme, where superposition encoding with a simple power split is employed, and (2) the cooperative protocol, where one receiver quantize-bin-and-forwards its received signal, and the other after receiving the side information decode-bin-and-forwards its received signal. 
\end{abstract}


\section{Introduction}
In modern communication systems, interference is one of the fundamental factors that limit performance.
Characterizing its capacity region is a long-standing open problem, except for several special cases (eg., the strong interference regime \cite{Sato_81}). The largest achievable region to date is reported by Han and Kobayashi \cite{HanKobayashi_81}, and the core of the scheme is a superposition coding strategy. 
Recent progress has been made on both inner bounds and outer bounds: 
Recently, Etkin, Tse, and Wang characterize the capacity region of the Gaussian interference channel to within one bit \cite{EtkinTse_07} by using a superposition coding scheme with a simple power-split configuration and by providing new upper bounds. The constant-gap-to-optimality result \cite{EtkinTse_07} provides a strong guarantee on the performance of the proposed scheme. Later, Motahari and Khandani \cite{MotahariKhandani_09}, Shang, Kramer, and Chen \cite{ShangKramer_09}, and Annapureddy and Veeravalli \cite{AnnapureddyVeeravalli_08} independently improve the bound and characterize the sum capacity in a very weak interference regime and a mixed interference regime.

In the above interference channel set-up, transmitters or receivers are not allowed to communicate with one another, and each user has to combat interference on its own. In various scenarios, however, nodes are not isolated, and transmitters/receivers can exchange certain amount of information. Cooperation among transmitters/receivers can help mitigate interference by forming distributed MIMO systems which provides two kinds of gains: \emph{degrees-of-freedom} gain and \emph{power} gain. The rate at which they cooperate, however, is limited, due to physical constraints. Therefore, one of the fundamental questions is, how much {\it interference} can limited {\it transmitter/receiver cooperation} mitigate? How much gain can it provide?

In this paper, we consider a two-user Gaussian interference channel with {\it conferencing decoders} to answer this question regarding receiver cooperation. Conferencing among encoders/decoders has been studied in \cite{Willems_83}, \cite{BrossLapidoth_08}, \cite{MaricYates_07}, \cite{DaboraServetto_06}, \cite{Simeone_08}, and \cite{YuZhou_08}. Our model is similar to those in \cite{Simeone_08} and \cite{YuZhou_08} but in an interference channel set-up. The work in \cite{Simeone_08} characterizes the capacity region of the compund MAC with unidirectional conferencing between decoders. For general set-up, it provides achievable rates but is not able to establish a constant-gap-to-optimality result. The work in \cite{YuZhou_08} considers one-sided Gaussian interference channels with unidirectional conferencing between decoders and characterizes the capacity region in strong interference regimes and the asymptotic sum capacity at high $\SNR$. For general receiver cooperation, works including \cite{Host-Madsen_06} and \cite{PrabhakaranViswanath_09}, investigate cooperation in interference channels with a set-up where the cooperative links are of the same band as the links in the interference channel. In particular, \cite{PrabhakaranViswanath_09} characterizes the sum capacity of Gaussian interference channels with in-band receiver cooperation to within a constant number of bits. Our work, on the other hand, is focused on the Gaussian interference channel with out-of-band (orthogonal) receiver cooperation, and studies its entire capacity region.

We propose a strategy achieving the capacity region universally to within 2 bits per user, regardless of channel parameters. The two-bit gap is the worst-case gap which can be loose in some regimes, and it is vanishingly small at high $\SNR$ when compared to the capacity. The strategy consists of two parts: (1) the transmission scheme, describing how transmitters encode their messages, and (2) the cooperative protocol, describing how receivers exchange information and decode messages. For transmission, both transmitters use superposition coding \cite{HanKobayashi_81} with the same power-split as in the case without cooperation \cite{EtkinTse_07}. For the cooperative protocol, it is appealing to apply the decode-forward or compress-forward schemes, originally proposed in \cite{CoverElGamal_79} for the relay channel,  like most works dealing with more complicated networks, including \cite{DaboraServetto_06}, \cite{Simeone_08}, \cite{YuZhou_08}, \cite{Host-Madsen_06}, \cite{KramerGastpar_05}, etc. It turns out neither conventional compress-forward nor decode-forward achieves capacity to within a constant number of bits universally for the problem at hand. On the other hand, \cite{CoverKim_07}, \cite{Kim_07}, and \cite{AvestimehrDiggavi_09} observe that the conventional compress-forward scheme \cite{CoverElGamal_79} may be improved by the destination directly decoding the sender's message instead of requiring to first decode the quantized signal of the relay. We use such an improved compress-forward scheme as part of our cooperative protocol. 
One of the receivers quantizes its received signal at an appropriate distortion, bins the quantization codeword and sends the bin index to the other receiver. The other receiver then decodes its own information based on its own received signal and the received bin index. After decoding, it bin-and-forwards the decoded common messages back to the former receiver and helps it decode. 

We identify two regions regarding the gain from receiver cooperation: linear and saturation regions, as illustrated through a numerical example in Fig. \ref{fig_CoopGainRegime}. In the linear region, receiver cooperation is \emph{efficient}, in the sense that the growth of user data rate is roughly linear with respect to the capacity of receiver-cooperative links. The gain in this region is the \emph{degrees-of-freedom} gain that distributed MIMO systems provide. On the other hand, in the saturation region, receiver cooperation is \emph{inefficient} in the sense that the growth of user data rate becomes saturated as one increases the rate in receiver-cooperative links. The gain is the \emph{power} gain of at most a constant number of bits, independent of the channel strength. We will focus on the system performance in the linear region, because not only that in most scenarios the rate at which receivers can cooperate is limited, but also that the gain from cooperation is more significant. 
 
\begin{figure}[htbp]
{\center
\includegraphics[width=5in]{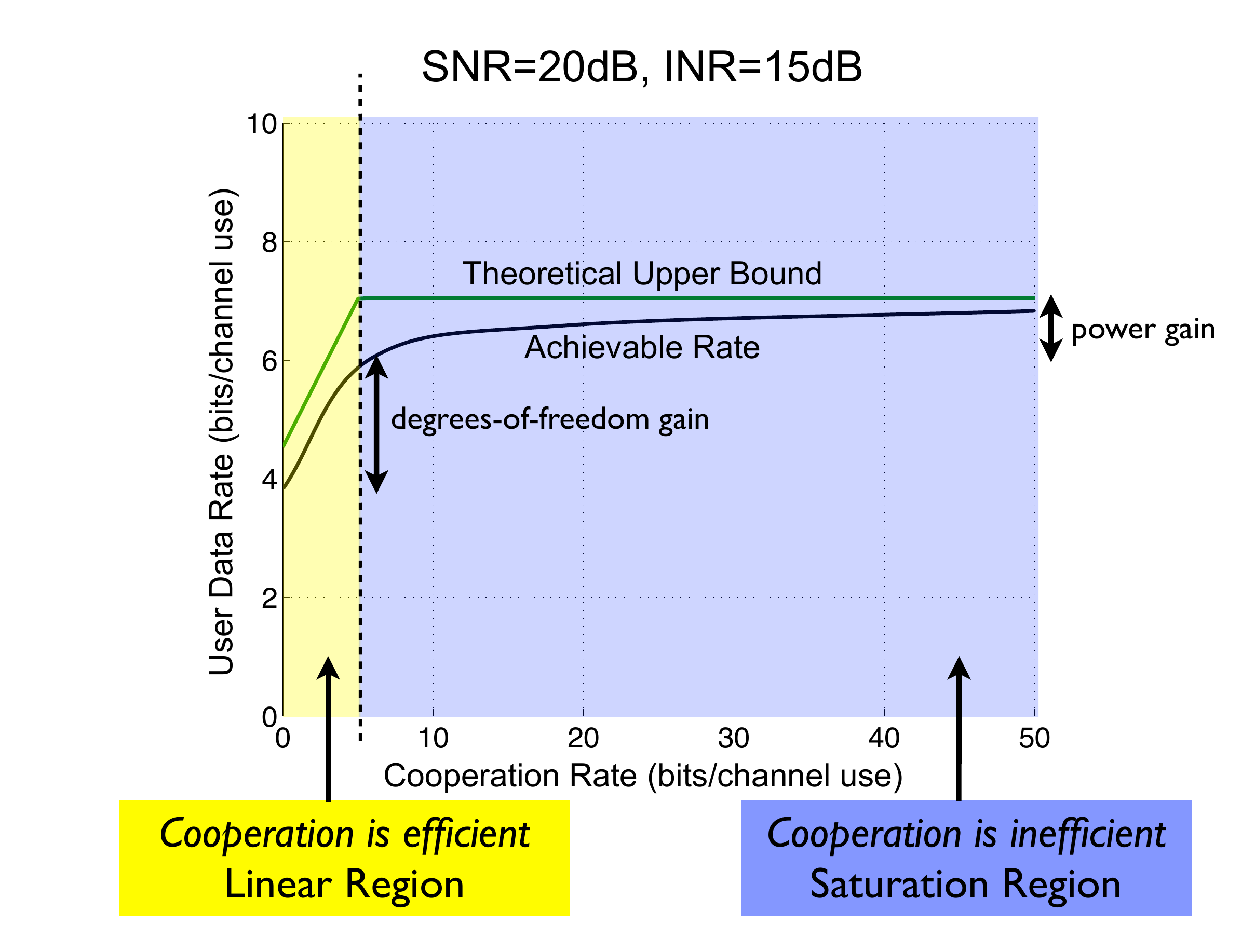}
\caption{The Gain from Limited Receiver Cooperation}
\label{fig_CoopGainRegime}
}
\end{figure}
 
With the constant-gap-to-optimality result, we find that the fundamental gain from cooperation in the linear region as follows: either \emph{one cooperation bit buys one more bit} or \emph{two cooperation bits buy one more bit} until saturation, depending on channel parameters. In the symmetric set-up, at high $\SNR$, when $\INR$ is below 50\% of $\SNR$ in dB scale, one-bit cooperation per direction buys roughly one-bit gain per user until full receiver cooperation performance is reached, while when $\INR$ is between 67\% and 200\% of $\SNR$ in dB scale, one-bit cooperation per direction buys roughly half-bit gain per user. In the weak interference regime, for a given pair of $(\SNR,\INR)$, when the receiver-cooperative link capacity $\C > \log\INR$, cooperation among receivers can get a close-to-interference-free (that is, within a constant number of bits) performance. In the strong interference regime, in contrast to that without cooperation, system performance can be boost {\it beyond} interference-free performance, by utilizing receiver-cooperative links not only for interference mitigation but also for forwarding desired information, since the cross link is stronger than the direct link.


The rest of this paper is organized as follows. In Section \ref{sec_Formulation}, we introduce the channel model and formulate the problem. In Section \ref{sec_Example}, we provide intuitive discussions about achievability and motivate our two-round strategy. Then we give examples to illustrate why it is not a good idea to use cooperative protocols based on conventional compress-forward or decode-forward. In Section \ref{sec_Achievable}, we describe the strategy concretely and derive its achievable rates, and in Section \ref{sec_AsymmConstantGap} we show that the achievable rate region is within 2 bits per user to the outer bounds we provide. In addition, we characterize the capacity region of compound MAC with conferencing decoders to within 1 bit, as a by-product. In Section \ref{sec_Dof}, focusing on the symmetric set-up, we illustrate the fundamental gain from receiver cooperation by deriving the optimal number of {\it generalized degrees of freedom} (g.d.o.f.) and compare it with the achievable ones of suboptimal schemes. 


\section{Problem Formulation}\label{sec_Formulation}

\subsection{Channel Model}
The two-user Gaussian interference channel with conferencing decoders is depicted in Fig. \ref{fig_ChModel}.

\begin{figure}[htbp]
{\center
\includegraphics[width=4in]{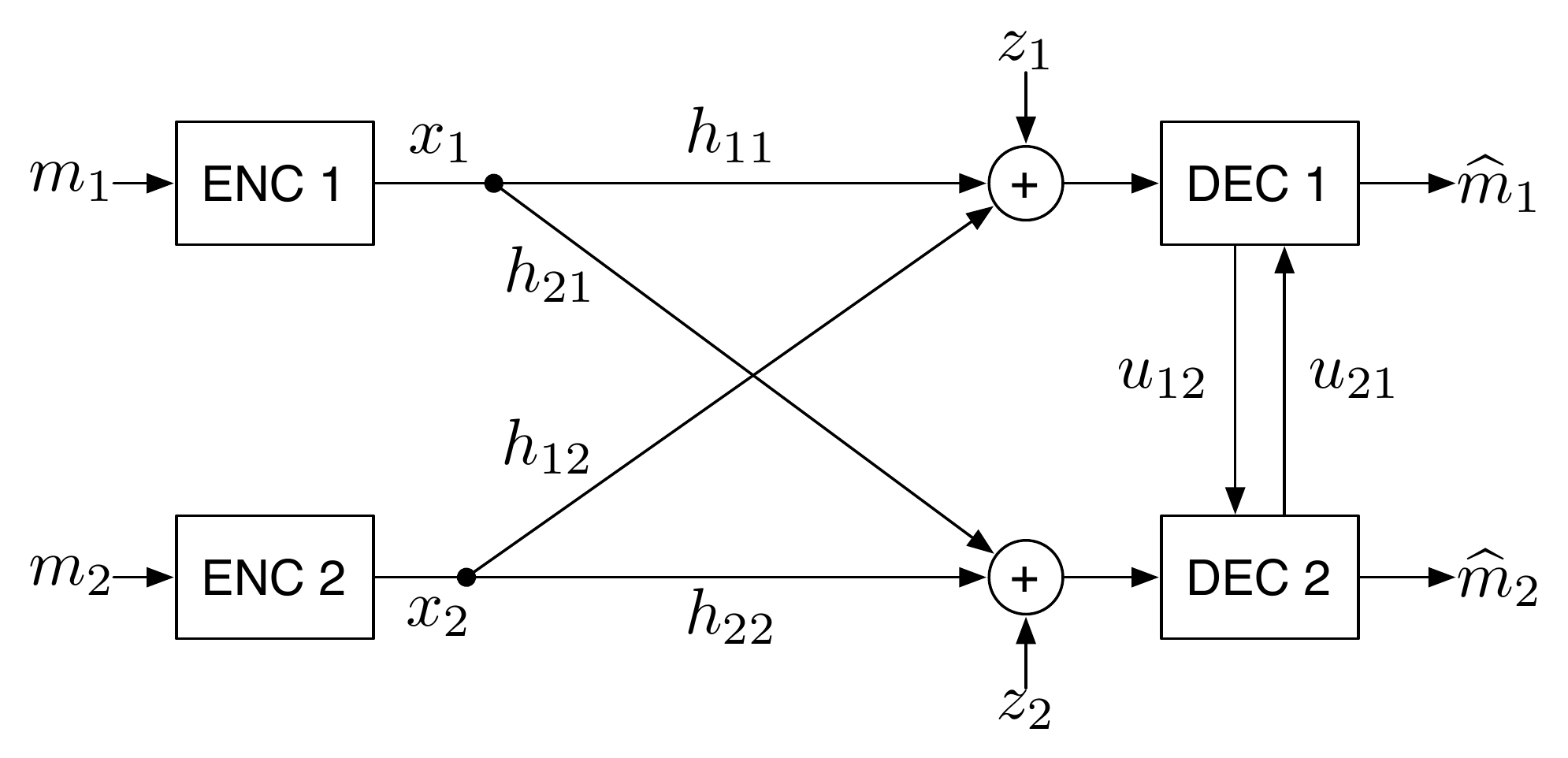}
\caption{Channel Model}
\label{fig_ChModel}
}
\end{figure}

\subsubsection{Transmitter-Receiver Links}
The transmitter-receiver links are modeled as the \emph{normalized} Gaussian interference channel:
\begin{align}
y_1 &= h_{11}x_1 + h_{12}x_2 + z_1\\
y_2 &= h_{21}x_1 + h_{22}x_2 + z_2,
\end{align}
where the additive noise processes $\{z_i[n]\}$, ($i=1,2$), are independent $\mathcal{CN}(0,1)$, i.i.d. over time. In this paper, we use $[.]$ to denote time indices. Transmitter $i$ intends to convey message $m_i$ to receiver $i$ by encoding it into a block codeword $\{x_i[n]\}_{n=1}^N$, with transmit power constraints
\begin{align} 
\frac{1}{N}\sum_{n=1}^N\E\Big[ \big\lvert x_i[n] \big\lvert^2 \Big] \le 1,\ i=1,2,
\end{align}
for arbitrary block length $N$. Note that outcome of the encoder depends solely on its own message. Messages $m_1,m_2$ are independent.
Define channel parameters
\begin{align}
\SNR_i := |h_{ii}|^2,\ \INR_i := |h_{ij}|^2,\ i,j=1,2,\ i\ne j.
\end{align}
%

\subsubsection{Receiver-Cooperative Links}
The receiver-cooperative links are noiseless with finite capacity $\C_{ij}$ from receiver $i$ to $j$. Encoding must satisfy causality constraints: for any time index $n=1,2,\ldots, N$, $u_{21}[n]$ is only a function of $\{y_2[1],\ldots, y_2[n-1], u_{12}[1], \ldots, u_{12}[n-1]\}$, and $u_{12}[n]$ is only a function of $\{y_1[1],\ldots, y_1[n-1], u_{21}[1], \ldots, u_{21}[n-1]\}$.

In the rest of this paper, we use $v^n$ to denote the sequence $\{v[1],\ldots,v[n]\}$.

\subsection{Strategies, Rates, and Capacity Region}
We give the basic definitions for the coding strategies, achievable rates of the strategy, and the capacity region of the channel.
\begin{definition}[Strategy and Average Probability of Error]
An $(M_1, M_2, N)$-strategy consists of the following: for $i,j=1,2$, $i\ne j$,
\begin{itemize}
\item[(1)] message set $\mcal{M}_i:= \{1,2,\ldots, M_i\}$ for user $i$;
\item[(2)] encoding function $e^{(N)}_i: \mcal{M}_i \rightarrow \mathbb{C}^N$, $m_i \mapsto x_i^N$ at transmitter $i$;
\item[(3)] set of relay functions $\{r^{(n)}_{i}\}_{n=1}^N$ such that $u_{ij}[n] =r^{(n)}_{i}(y_i^{n-1},u_{ji}^{n-1}) \in \{1,2,\ldots, 2^{\C_{ij}}\}$, $\forall n=1,2,\ldots,N$ at receiver $i$;
\item[(4)] decoding function $d^{(N)}_i:  \mathbb{C}^N\times \{1,2,\ldots, 2^{N\C_{ji}}\} \rightarrow \mcal{M}_i$, $(y_i^N , u_{ji}^N) \mapsto \what{m}_i$ at receiver $i$.
\end{itemize}

The average probability of error 
\begin{align}
P^{(N)}_e := \frac{1}{M_1M_2}\sum_{m_1\in \mcal{M}_1, m_2\in \mcal{M}_2} 
\Pr\lbp \begin{array}{l}d^{(N)}_1(y_1^N , u_{21}^N) \ne m_1\ \mathrm{or}\\ d^{(N)}_2(y_2^N , u_{12}^N)\ne m_2 \end{array}\Bigg\lvert m_1,m_2\ \textrm{are sent}\rbp
\end{align}
\end{definition}

\begin{definition}[Achievable Rates and Capacity Region]
A rate tuple $(R_1,R_2)$ is achievable if for any $\epsilon > 0$ and for all sufficiently large $N$, there exists an $(M_1, M_2, N)$ strategy with $M_i \ge 2^{NR_i}$, for $i=1,2$, such that $P^{(N)}_e < \epsilon$. The capacity region $\mscr{C}$ is the collections of all achievable $(R_1,R_2)$.
\end{definition}

\subsection{Notations}
We summarize below the notations used in the rest of this paper.
\begin{itemize}
\item
For a real number $a$, $(a)^+ := \max(a,0)$ denotes its positive part.
\item
For sets $A,B \subseteq \mathbb{R}^k$ in $k$-dimensional space, $A \oplus B := \{a+b: a\in A, b\in B\}$ denotes the direct sum of $A$ and $B$. $\conv\{A\}$ denotes the convex hull of the set $A$.
\item
With a little abuse of notations, for $x,y\in\mathbb{F}_q$, $x\oplus y$ denotes the modulo-$q$ sum of $x$ and $y$.
\item
Unless specified, all the logarithms $\log(.)$ is of base 2.
\end{itemize}


\section{Motivation of Strategies}\label{sec_Example}
Before introducing our main result, we first provide intuitive discussions about achievability and motivate our two-round strategy (to be described in detail in Section \ref{sec_Achievable}) from a high-level perspective. Then we give examples to illustrate why cooperative protocols based on conventional compress-forward or decode-forward may not be good for cooperation between receivers to mitigate interference. Throughout the discussion in this section, we will make use of the \emph{linear deterministic channel} proposed in \cite{AvestimehrDiggavi_09}.

\subsection{Optimal Strategy in the Linear Deterministic Channel}
First, consider the following symmetric channel: $\SNR_1=\SNR_2=\SNR$, $\INR_1=\INR_2=\INR$, and $\C_{12}=\C_{21}=\C$. Set $\INR$ to be 2/3 of $\SNR$ in dB scale, that is, $\log\INR=\frac{2}{3}\log\SNR$. Set $\C=\frac{1}{3}\log\SNR$. The corresponding linear deterministic channel (LDC) is depicted in Fig. \ref{fig_Examples}. The bits at the levels of transmitters/receivers can be thought of as chunks of binary expansions of the transmitted/received signals. Note that in this example, one bit in the LDC corresponds to $\frac{1}{3}\log\SNR$ in the Gaussian channel.

We begin with the baseline where two receivers are not allowed to cooperate. The transmitted signals are naturally broken down into two parts: (1) the common levels, which appear at both receivers, and (2) the private levels, which only appear at its own receiver. Each transmitter splits its message into common and private parts, which are linearly modulated onto the common and private levels of the signal respectively. Each receiver then decodes both user's common messages and its own private message by solving the linear equations it received. This is shown to be optimal in the two-user interference channel \cite{BreslerTse_08}. In this example (Fig. \ref{fig_Examples}.(a)), bits $a_1$ and $b_1$ are common, while $a_3$ and $b_3$ are private. The sum capacity without cooperation is 4 bits. Since all levels at both receivers are occupied, one cannot turn on bits $a_2$ or $b_2$ without causing collisions.

\begin{figure}[htbp]
{\center
\subfigure[Without Cooperation]{\includegraphics[width=3in]{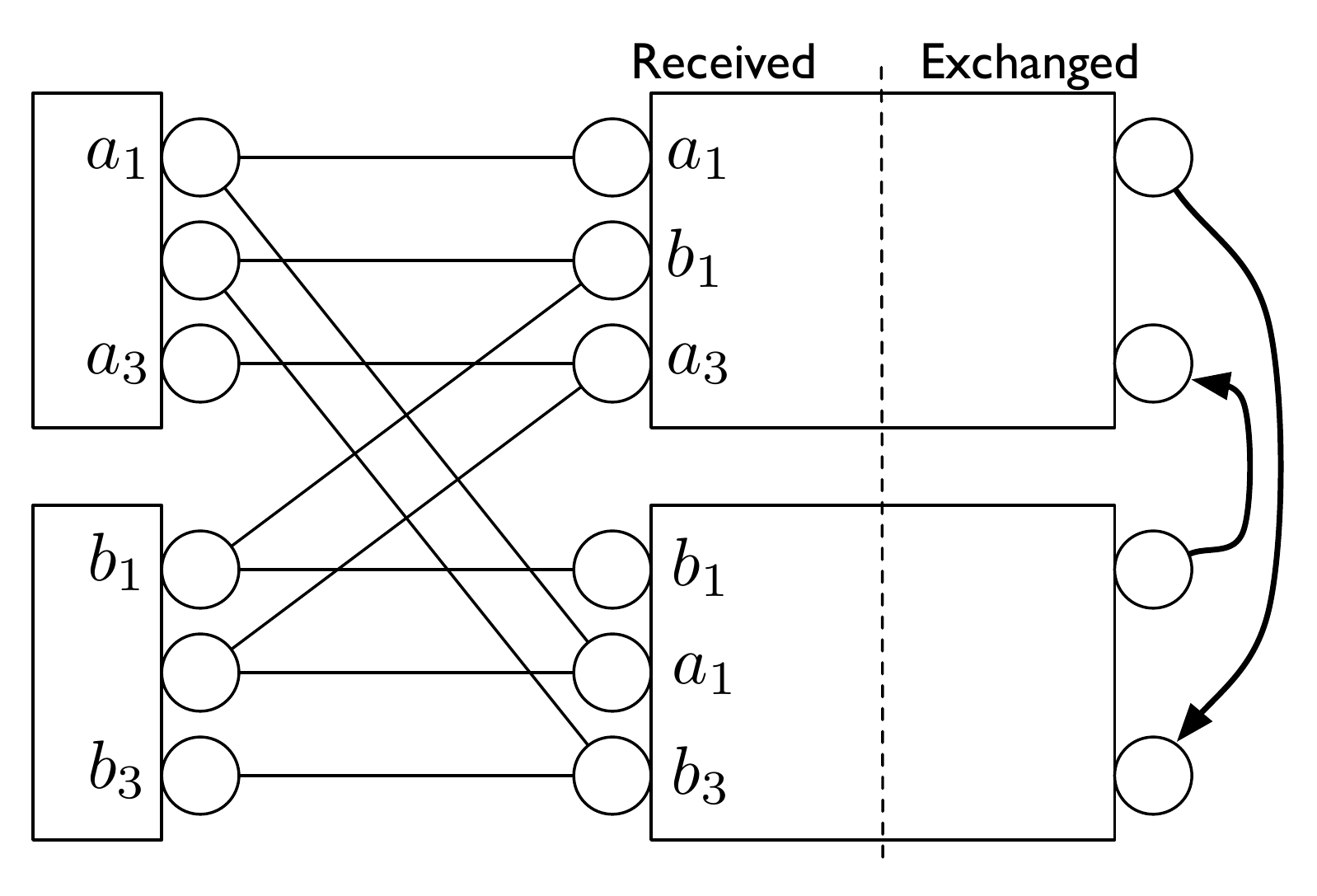}}
\subfigure[With Cooperation]{\includegraphics[width=3in]{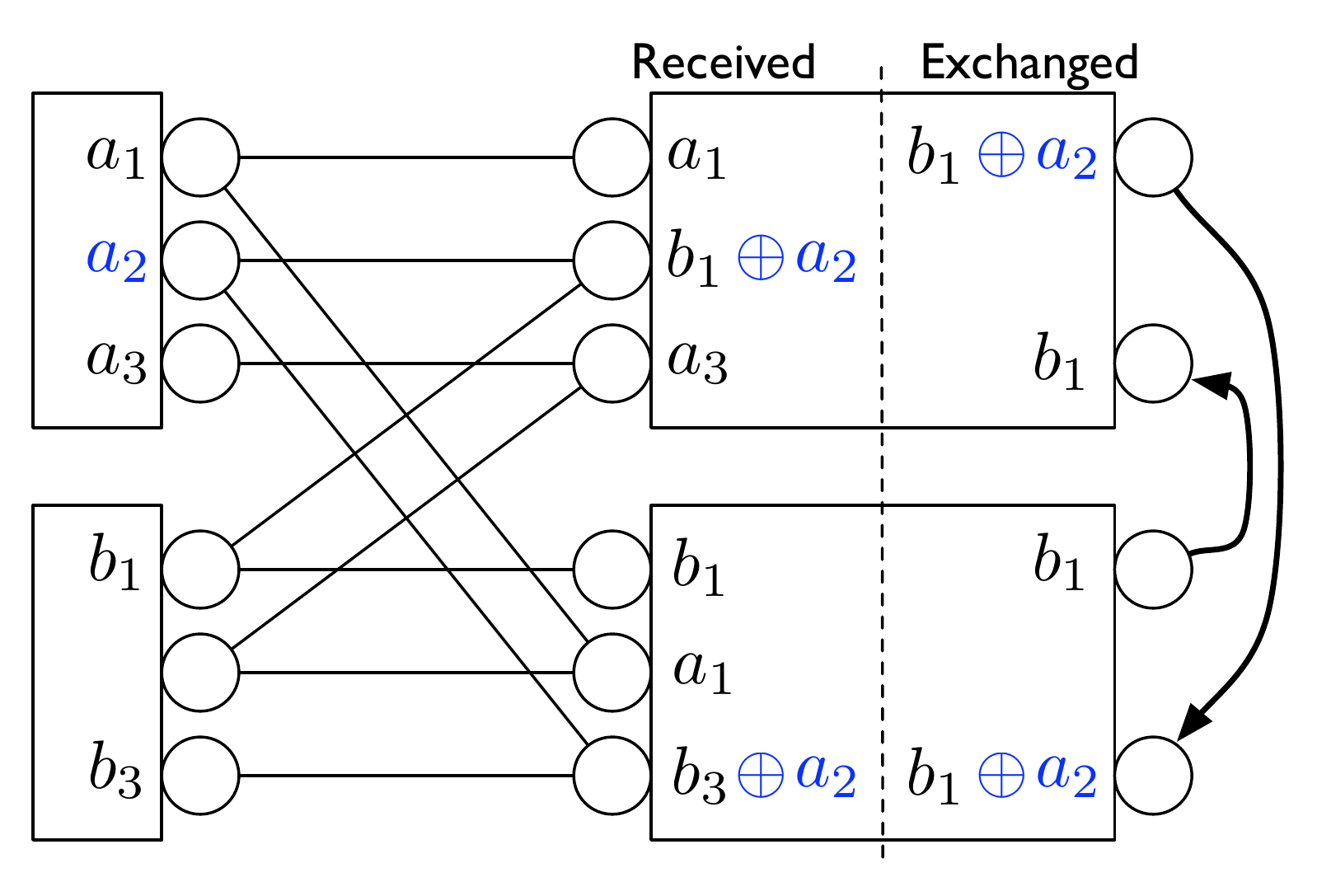}}
\subfigure[Conventional Compress-Forward]{\includegraphics[width=3in]{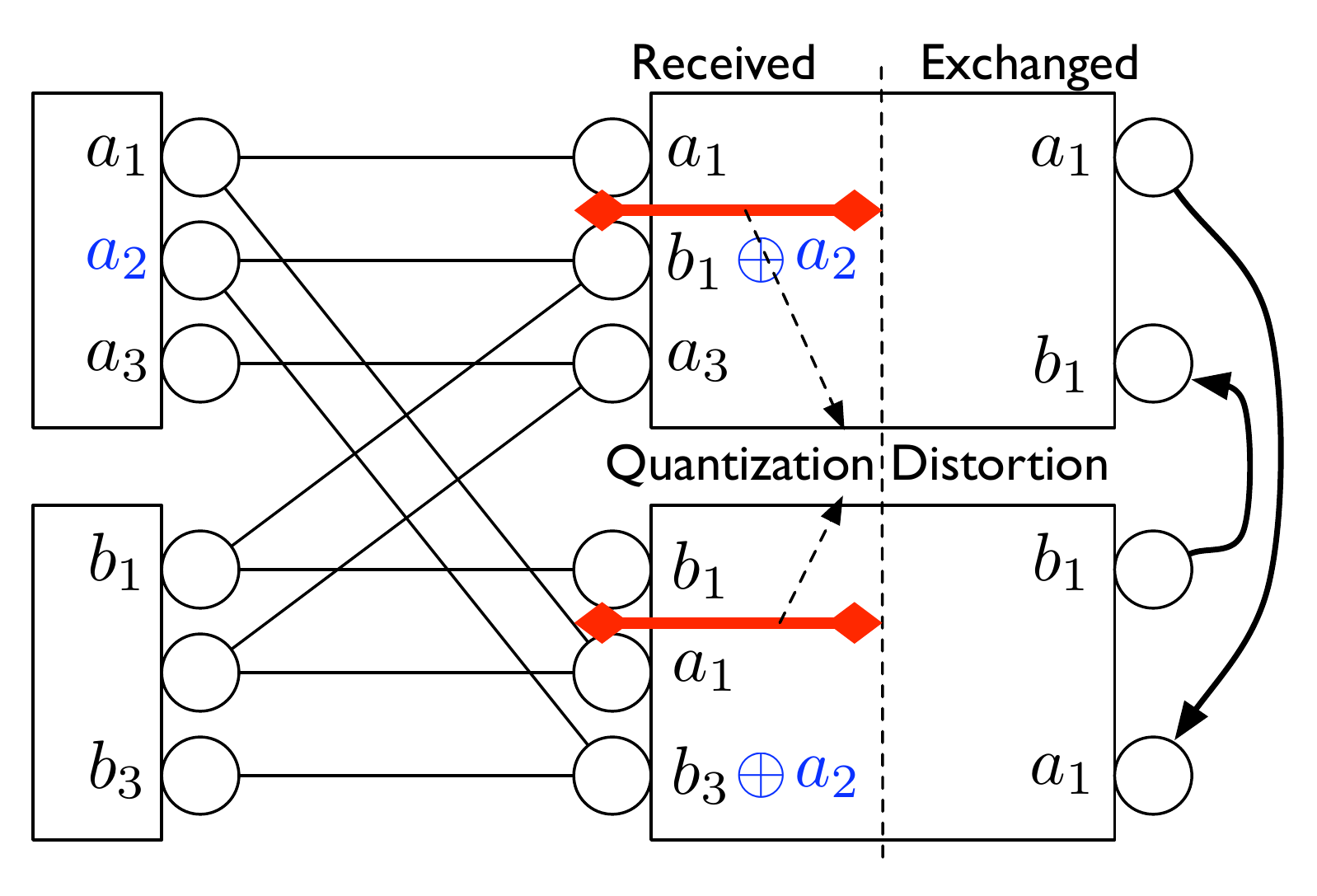}}
\caption{An Example Channel}
\label{fig_Examples}
}
\end{figure}

With receiver cooperation, the natural split of transmitted signals does not change. This suggests that the encoding procedure and the aim of each decoder remain the same. Each receiver with the help from the other receiver, however, is able to decode more information. Since each user's private message is not of interest to the other receiver, a natural scheme for receiver cooperation is to exchange linear combinations formed by the signals \emph{above} the private signal level so that the undesired signal does not pollute the cooperative information. In this example, as illustrated in Fig. \ref{fig_Examples}.(b), with one-bit cooperation in each direction in the LDC, the optimal sum rate is 5 bits, achieved by turning on one more bit $a_2$. This causes collisions at the second level at receiver 1 and at the third level at receiver 2, while they can be resolved with cooperation: receiver 1 sends $b_1\oplus a_2$ to receiver 2, and receiver 2 sends $b_1$ to receiver 1. Now receiver 1 can solve $\lp a_1,a_2,a_3,b_1\rp$, and receiver 2 can solve $\lp b_1,b_3,a_1,a_2\rp$. In fact, the exchanged linear combinations are not unique. For example, receiver 1 can send $\lp b_1\oplus a_2 \rp\oplus a_1$ and receiver 2 can send $b_1\oplus a_1$, and this again achieves the same rates. As long as receiver 1 does not send a linear combination containing the private bit $a_3$ and the sent linear combination is linearly independent of the signals at receiver 2 (and vice versa for the linear combination sent from receiver 2 to receiver 1), the scheme is optimal for this example channel. The above discussion regarding the scheme in the LDC naturally leads to an implementable one-round scheme in the Gaussian channel, where both receivers quantize-and-bin their received signals at their own private signal level.


In the above example, it is optimal that each receiver sends to the other linear combinations formed by its received signal above its private signal level. Is this optimal in general? The answer is no. Consider the following asymmetric example: $\SNR_2=\INR_2$, $\SNR_1$ is $2/3$ of $\SNR_2$ in dB, and $\INR_1$ is $1/3$ of $\SNR_2$ in dB. $\C_{12}= \frac{2}{3}\log\SNR_2$ and $\C_{21}= \frac{1}{3}\log\SNR_2$. The corresponding LDC is depicted in Fig. \ref{fig_Asymm}, where one bit in the LDC corresponds to $\frac{1}{3}\log\SNR_2$ in the Gaussian channel. First consider the same scheme as that in the previous exmaple. Note that if receiver 2 just forwards signals above its private signal level, it can only forward $a_1$ to receiver 1 and achieves $R_1$ up to 2 bits. On the other hand, if receiver 2 forwards $a_3$ to receiver 1, which is below user 2's private signal level, it achieves $R_1=3$ bits. From this example, we see that once there are ``useful" information (which should not be polluted by the receiver's own private bits) lies \emph{at or below} the private signal level (in this example, the bit $a_3$), the one-round scheme described in the previous example is suboptimal. To extract the useful information at or below the private signal level, one of the receivers (in this example, receiver 2) can first decode and then form linear combinations using (decoded) common messages \emph{only}. 

\begin{figure}[hbtp]
{\center
\subfigure[Suboptimal Scheme]{\includegraphics[width=3in]{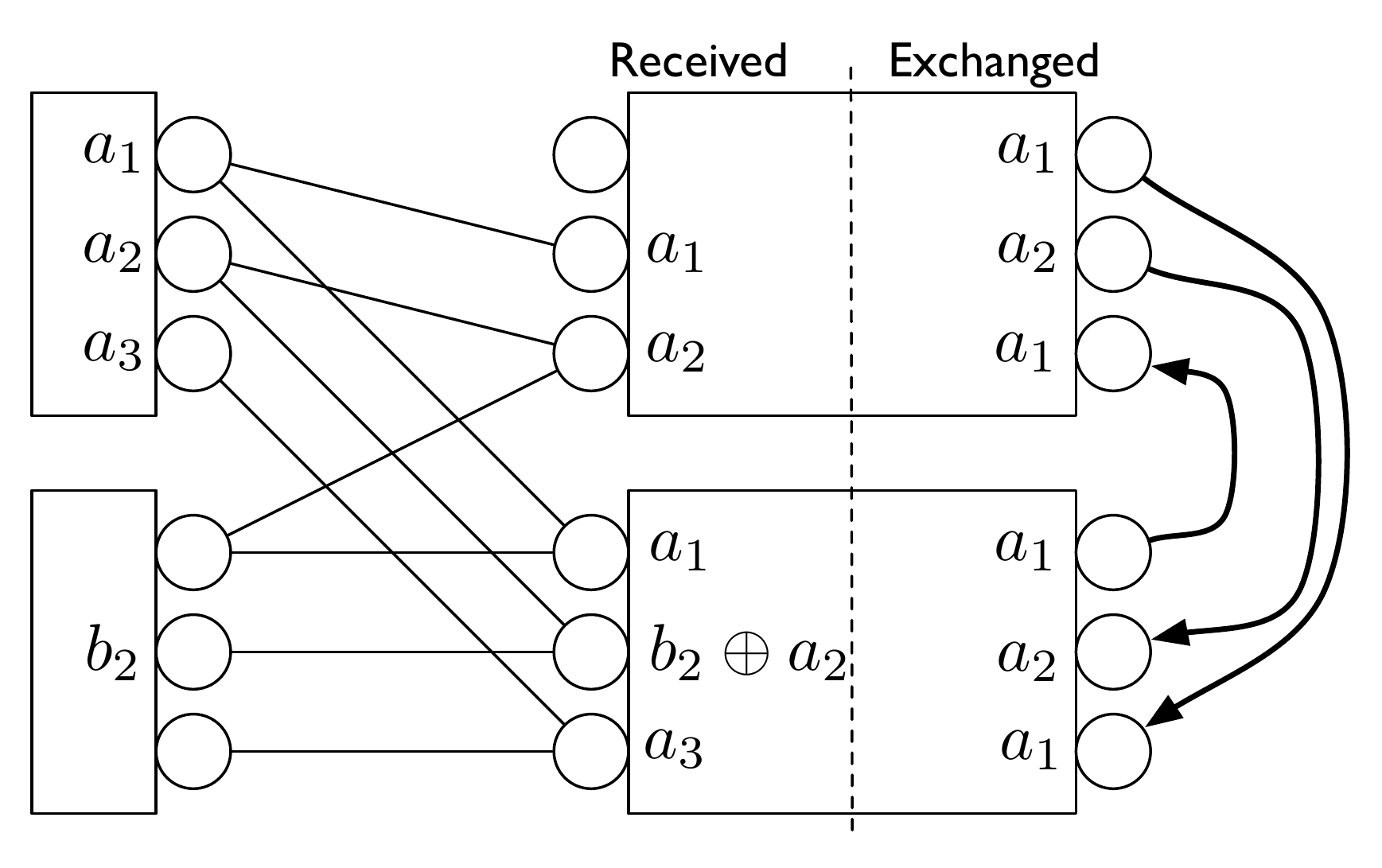}}
\subfigure[Optimal Scheme]{\includegraphics[width=3in]{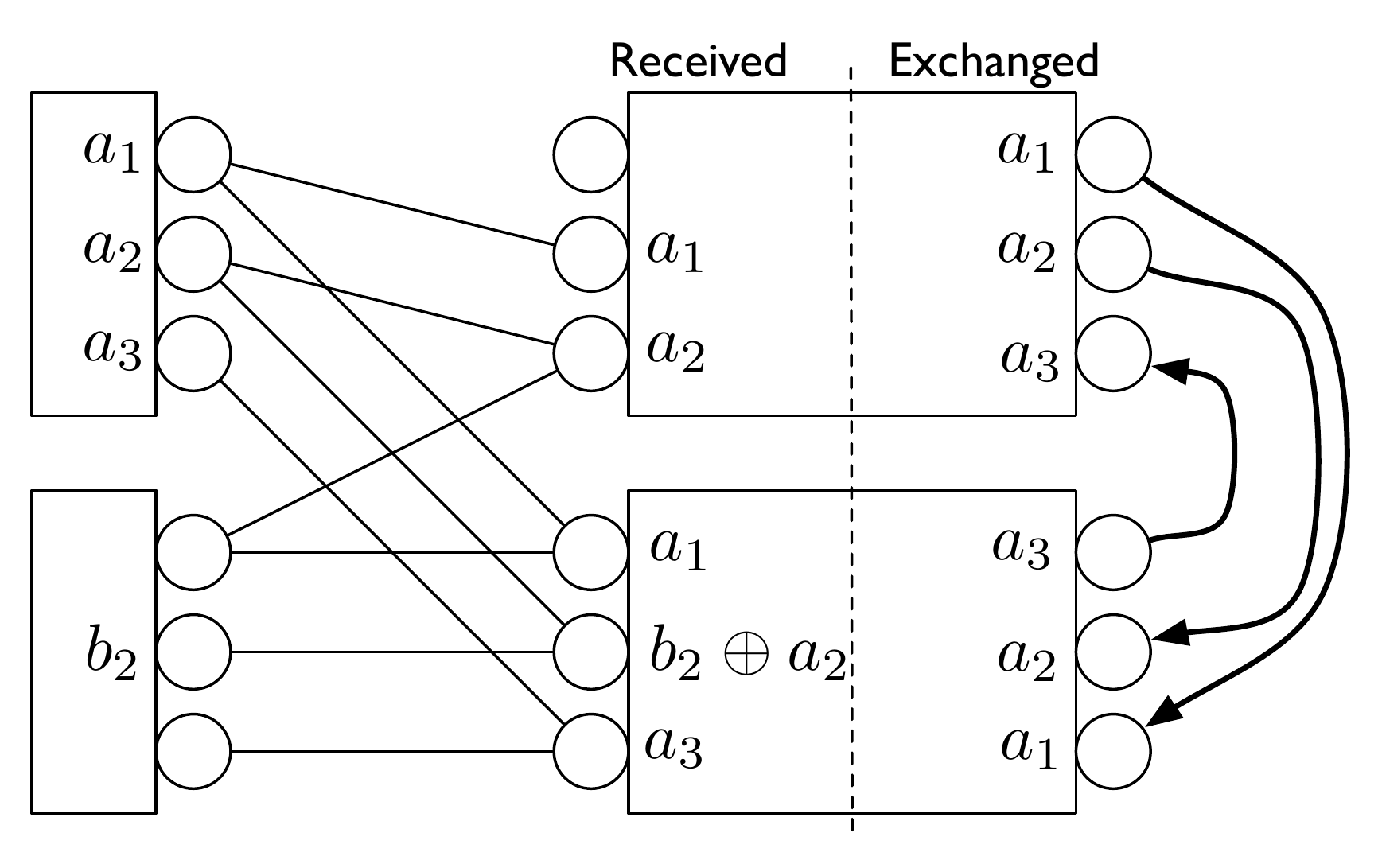}}
\caption{An Asymmetric Example}
\label{fig_Asymm}
}
\end{figure}

Without loss of generality, however, the above situation that there is useful information for the other receiver lies at or below the private signal level, only happens at one of the two receivers. In other words, there exists a receiver where no useful information (for the other receiver) lies at or below the private signal level. The reason is the following:
\begin{enumerate}
\item It is straightforward to see that the capacity region is convex, and hence if a scheme can achieve $\max_{(R_1,R_2) \in \mscr{C} }\lbp \mu_1 R_1 + \mu_2 R_2\rbp$ for all $\mu_1,\mu_2 \ge 0$, it is optimal.
\item If $\mu_1 \ge \mu_2$, we weigh user 1's rate more. Since the private bits are cheaper to support in the sense that they do not cause interference at receiver 2, user 1 should be transmitting at its full private rate, which is equal to the number of levels at or below the private signal level at receiver 1. Therefore, all levels at or below the private signal level are occupied by user 1's private bits and there is no useful information at receiver 1 for receiver 2.
\item Similarly if $\mu_1 \le \mu_2$, there is no useful information at receiver 2 for receiver 1 at or below the private signal level.
\end{enumerate}
Hence, the following two-round strategy is optimal in the LDC: if $\mu_1 \ge \mu_2$, receiver 1 forms a certain number (no more than the cooperative link capacity) of linear combinations composed of the signals above its private signal level and sends them to receiver 2. After receiver 2 decodes, it forms a certain number of linear combinations composed of the decoded common bits and sends them to receiver 1. If $\mu_1 \le \mu_2$, the roles of receiver 1 and 2 are exchanged. Note that depending on the operating point in the capacity region, we use different configurations, implying that time-sharing is needed to achieve the full capacity region.

From the above discussion, a natural and implementable two-round strategy for Gaussian channels emerges. For the transmission, we use a superposition Gaussian random coding scheme with a simple power-split configuration, as described in \cite{EtkinTse_07}. For the cooperative protocol, one of the receivers quantize-and-bins its received signal at its private signal level and forwards the bin index; after the other receiver decodes with the helping side information, it bin-and-forwards the decoded common messages back to the first receiver and helps it decode. In Section \ref{sec_AsymmConstantGap}, we shall prove that this strategy achieves the capacity region universally to within 2 bits per user.



\subsection{Conventional Compress-Forward and Decode-Forward}
Now, suppose the conventional compress-forward assuming joint Gaussianity of the received signals is used for receivers to cooperate. It is a standard approach to evaluate achievable rates of Gaussian channels using compress-forward in the literature, including \cite{DaboraServetto_06}, \cite{Simeone_08}, \cite{YuZhou_08}, \cite{Host-Madsen_06}, \cite{KramerGastpar_05}, etc. From its corresponding LDC, however, one can see that the two received signals of the Gaussian channel, $\lp y_1,y_2\rp$, are not jointly Gaussian. The reason is that, suppose they are jointly Gaussian, the conditional distribution of $y_2$ given $y_1$ should be marginally Gaussian. As Fig. \ref{fig_Examples} suggests, however, conditioning on receiver 1's signal results in a hole at the second level of receiver 2's signal, which was occupied by $a_1$. Therefore, transmitter 2's common codebook is not dense enough to make the conditional distribution of $y_2$ given $y_1$ marginally Gaussian. The incorrect assumption results in larger quantization distortions, as depicted in Fig. \ref{fig_Examples}.(c). The information sent from receiver 1 to receiver 2, $a_1$, is \emph{redundant}, and cannot help mitigate interference $a_2$. Hence, the achievable sum rate is 4 bits (3 bits for user 1 and 1 bit for user 2), which is the same as that without cooperation and is one bit smaller than the optimal one. Recall that 1 bit in the LDC corresponds to $\frac{1}{3}\log\SNR$ in the Gaussian channel, therefore the performance loss is unbounded as $\SNR\rightarrow\infty$. The main reason why conventional compress-forward does not work well is that the scheme does not well utilize the dependency between the two received signals.

On the other hand, suppose two receivers cooperate based on the decode-forward scheme. Note that there is no gain if we require both common messages to be decoded at one of the receivers at the first stage without cooperation. By symmetry we can assume that, without loss of generality, each receiver first decodes its own common message and then bin-and-forwards the decoded information to the other receiver. At the second stage, it then decodes the other user's common message with the help from cooperation and decodes its own private message. In the corresponding LDC, the common bit $a_2$ cannot be decoded at the first stage, and hence the total throughput using this strategy is at most 4 bits, which is again the same as that without cooperation. The reason why decode-forward is not good for the two receivers to cooperate is that, it is too costly to decode users' own common message at the first stage without the help from cooperation.

\section{A Two-Round Strategy}\label{sec_Achievable}
In this section we describe the two-round strategy and derive its achievable rate region.
The strategy consists of two parts: (1) the transmission scheme and (2) the cooperative protocol.
\subsection{Transmission Scheme}\label{subsec_Transmission}
We use a simple superposition coding scheme with Gaussian random codebooks. For each transmitter, it splits its own message into common and private (sub-)messages. Each common message is aimed at both receivers, while each private one is aimed at its own receiver. Each message is encoded into a Gaussian random codeword with certain power. For transmitter $i$, the power for its private and common codewords are $Q_{ip}$ and $Q_{ic}=1-Q_{ip}$ respectively, for $i=1,2$. As \cite{EtkinTse_07} points out, since the private signal is undesired at the unintended receiver, a reasonable configuration is to make the private interference at or below the noise level so that it does not cause much damage and can still convey additional information in the direct link if it is stronger than the cross link. When the interference is stronger than the desired signal, simply set the whole message to be common. In a word, for $(i,j) = (1,2)$ or $(2,1)$, $Q_{ip} = \max\lbp \frac{1}{\INR_j}, 1\rbp$ if $\SNR_i > \INR_j$, and $Q_{ip}=0$ otherwise.

\subsection{Cooperative Protocol}\label{subsec_CoopProtocol}
The cooperative protocol is two-round. We briefly describe it as follows: for $(i,j)=(1,2)$ or $(2,1)$, at the first round, receiver $j$ quantizes its received signal and sends out the bin index (described in detail below). At the second round, receiver $i$ receives this side information, decodes its desired messages (both users' common message and its own private message) with the decoder described in detail below, randomly bins the decoded common messages, and sends the bin indices to receiver $j$. Finally receiver $j$ decodes with the help from the receiver-cooperative link. We call this a two-round strategy $\STG_{j\rightarrow i\rightarrow j}$, meaning that the processing order is: receiver $j$ quantize-and-bins, receiver $i$ decode-and-bins, and receiver $j$ decodes. Its achievable rate region is denoted by $\mscr{R}_{j\ra i \ra j}$. By time-sharing, we can obtain achievable rate region $\mscr{R} := \mathrm{conv}\lbp \mscr{R}_{2\ra1\ra2} \cup \mscr{R}_{1\ra2\ra1}\rbp$, convex hull of the union of two rate regions.

\begin{remark}[Engineering interpretation]
There is a simple way to understand the strategy from an engineering perspective. To achieve $\max_{(R_1,R_2) \in \mscr{R} }\lbp \mu_1 R_1 + \mu_2 R_2\rbp$ for some non-negative $\lp \mu_1,\mu_2\rp$, the processing configuration can be easily determined: strategy $\STG_{j\ra i\ra j}$ should be used, where $i = {\arg\min}_{l=1,2}\{ \mu_l\}$ and $j = {\arg\max}_{l=1,2}\{ \mu_l\}$. In a word, the receiver which decodes last is the one we favor most. This is the high-level intuition we obtained from the discussion in the LDC in Section \ref{sec_Example}. 
\end{remark}

In the following, we describe each component in detail, including quantize-binning, decode-binning, and their corresponding decoders. For simplicity, we consider strategy $\STG_{2\ra1\ra2}$.

{\flushleft \it Quantize-binning}:\par
Upon receiving its signal from the transmitter-receiver link, receiver 2 does not decode messages immediately. Instead, serving as a relay, it first quantizes its signal by a pre-generated Gaussian quantization codebook with certain distortion, and then sends out a bin index determined by a pre-generated binning function. How should we set the distortion? As discussed in the previous section, note that both its own private signal and the noise it encounters are not of interest to receiver 1. Therefore, a natural configuration is to set the distortion level equal to {\it the maximum of noise power and private signal power level}.

{\flushleft \it Decoder at receiver 1}:\par
After retrieving the receiver-cooperative side information, that is, the bin index, receiver 1 decodes two common messages and its own private message, by searching in transmitters' codebooks for a codeword triple (indexed by the two common messages and the user's own private message) that is jointly typical with its received signal and some quantization point (codeword) in the given bin. If there is no such unique codeword triple, it declares an error.

{\flushleft \it Decode-binning}:\par
After receiver 1 decodes, it uses two pre-generated binning functions to bin the two common messages and sends out these two bin indices to receiver 2.

{\flushleft \it Decoder at receiver 2}:\par
After receiving these two bin indices, receiver 2 decodes two common messages and its own private message, by searching in transmitters' codebooks for a codeword triple that is jointly typical with its received signal and the common messages both lie in the given bins.

\begin{remark}[Difference from the conventional compress-forward]\label{rm_DiffCF}
The action of receiver 2 as a relay is very similar to that of the relay in the conventional compress-forward with Gaussian vector quantization. Note that the main difference from the conventional compress-forward with Gaussian vector quantization lies in the \emph{decoding} procedure (at receiver 1) and the chosen distortion. In the conventional Gaussian compress-forward, the decoder first searches in the bin for one quantization codeword that is jointly typical with its received signal from its own transmitter \emph{only}, assuming that the two received signals are jointly Gaussian. This may not be true since a single user may not transmit at the capacity in its own link, which results in ``holes" in signal space. As a consequence, this scheme may not utilize the dependency of two received signals well and cause larger distortions. Our scheme, on the other hand, utilizes the dependency in a better way by \emph{jointly} deciding the quantization codeword and the message triple, consequently allows smaller distortions, and is able to reveal the beneficial side information to the other receiver.

Quantize-binning and its corresponding decoding part of our scheme is very similar to \emph{extended hash-and-forward} proposed in \cite{Kim_07}, in which it is pointed out that the scheme has no advantage over conventional compress-forward in a single-source single-relay setting. Due to the above mentioned issues, however, we recognize in our problem where the channel consists of two source-destination pairs and two relays, the scheme has an unbounded advantage over the conventional compress-forward in certain regimes. This will be made clear in later sections.
\end{remark}

\subsection{Achievable Rates}\label{subsec_AchievableRate}
The following theorem establish the achievable rates of strategy $\STG_{2\ra1\ra2}$. Let $R_{ic}$ and $R_{ip}$ denote the rates for user $i$'s common message and private message respectively, for $i=1,2$.

\begin{theorem}[Achievable Rate Region for $\STG_{2\ra1\ra2}$]\label{thm_CodingThm}
The rate tuple $\lp R_{1c},R_{2c},R_{1p},R_{2p}\rp$ satisfying the following constraints are achievable:
{\flushleft \emph{\underline{Constraints at receiver 1}}:}
\begin{align}
R_{1p} &\le \min \lbp I\lp x_{1};y_1|x_{1c},x_{2c}\rp + (\C_{21} - \xi_1)^+, I\lp x_{1};y_1,\what{y}_2 | x_{1c},x_{2c}\rp \rbp \label{eq_PrivRate}\\
R_{2c} &\le \min \lbp I\lp x_{2c};y_1|x_{1}\rp + (\C_{21} - \xi_1)^+, I\lp x_{2c};y_1,\what{y}_2 | x_{1}\rp \rbp \label{eq_CommRate}\\
R_{2c}+R_{1p} &\le \min \lbp I\lp x_{2c},x_{1};y_1|x_{1c}\rp + (\C_{21} - \xi_1)^+, I\lp x_{2c},x_{1};y_1,\what{y}_2 | x_{1c}\rp \rbp\\
R_{1c}+R_{1p} &\le \min \lbp I\lp x_{1};y_1|x_{2c}\rp + (\C_{21} - \xi_1)^+, I\lp x_{1};y_1,\what{y}_2 | x_{2c}\rp \rbp \label{eq_User1Rate}\\
R_{1c}+R_{2c}+R_{1p} &\le \min \lbp I\lp x_{1},x_{2c};y_1 \rp + (\C_{21} - \xi_1)^+, I\lp x_{1},x_{2c};y_1,\what{y}_2 \rp \rbp
\end{align}
where $\xi_1 = I\lp \what{y}_2; y_2| x_{1c},x_1,x_{2c},y_1\rp$. For $i=1,2$, $x_{ic}\sim\mcal{CN}\lp0,Q_{ic}\rp$ is the common codebook generating random variable. $x_{1} = x_{1p}+x_{1c}$ is the superposition codebook generating variable, where $x_{1p}\sim\mcal{CN}\lp0,Q_{1p}\rp$ is independent of $x_{1c}$. $\what{y}_2\overset{d}{=}y_2+\what{z}_2$ is the quantization codebook generating random variable, and $\what{z}_2\sim\mcal{CN}\lp0,\Delta_2\rp$, independent of everything else. $\Delta_2$ is the quantization distortion at receiver 2.

{\flushleft \emph{\underline{Constraints at receiver 2}}:}
\begin{align}
R_{2p} &\le I\lp x_{2};y_2|x_{2c},x_{1c}\rp\\
R_{1c}+R_{2p} &\le I\lp x_{1c},x_{2};y_2|x_{2c}\rp + \C_{12}\\
R_{2c}+R_{2p} &\le I\lp x_{2};y_2|x_{1c}\rp + \C_{12}\\
R_{2c}+R_{1c}+R_{2p} &\le I\lp x_{2},x_{1c};y_2 \rp + \C_{12},
\end{align}
where $x_{2} = x_{2p}+x_{2c}$ is the superposition codebook generating variable, and $x_{2p}\sim\mcal{CN}\lp0,Q_{2p}\rp$ is independent of $x_{2c}$.
\end{theorem}

\begin{proof}
For details, see Appendix \ref{app_ErrorAnalysis}. Here we give some high-level comments on these rate constraints. First, unlike interference channels without cooperation, here receiver 1 is required to decode $m_{2c}$ correctly so that it can help receiver 2. This additional requirement gives the rate constraint \eqref{eq_CommRate} on $R_{2c}$. 

Second, in the set of constraints at receiver 1, on the right-hand side they are all minimum of two terms. The second term corresponds to the case when the receiver-cooperative link is strong enough to convey the quantized $\what{y}_2^N$ correctly. The first term corresponds to the case when receiver 1 can only figure out a set of candidates of quantized $\what{y}_2^N$. In Section \ref{sec_Example} we see that, in the LDC as long as the quantization level is chosen such that no private signals pollute the cooperative information, the cooperation from receiver 2 to 1 is able to gain $\C_{21}$ comparing with the case without cooperation. In the Gaussian channel, however, due to the carry-over of real additions and the Gaussian vector quantization, there is a rate loss of $\xi_1$ bits (which is at most a constant number of bits if we choose the quantization distortion properly). In fact, $\xi_1 = I\lp \what{y}_2; y_2| x_{1c},x_1,x_{2c},y_1\rp$ corresponds to the number of private bits polluting the cooperative linear combinations in the LDC if one does not choose the quantization distortion properly.

Finally, in the set of constraints at receiver 2, since receiver 1 only help receiver 2 decoding $m_{1c}$ and $m_{2c}$, there is no enhancement in $R_{2p}$.
\end{proof}


We shall use the following shorthand notations throughout the rest of the paper: for $(i,j)=(1,2),(2,1)$,
\begin{align}
\SNR_{ip} &:= |h_{ii}|^2Q_{ip} = \SNR_i\cdot Q_{ip},\\
\INR_{ip} &:= |h_{ij}|^2Q_{jp} = \INR_i\cdot Q_{jp}.
\end{align}

Next, we quantify the ``rate loss" term $\xi_1$ in the set of rate constraints at receiver 1, in terms of distortions $\Delta_2$:
\begin{align}
\xi_1 &= I(\what{y}_2;y_2|x_{1c},x_1,x_{2c},y_1) = h(\what{y}_2|x_{1c},x_1,x_{2c},y_1) - h(\what{y}_2|x_{1c},x_1,x_{2c},y_1,y_2)\\
&= h\lp h_{22}x_{2p}+z_2+\what{z}_2|h_{12}x_{2p}+z_1\rp - h\lp \what{z}_2\rp\\
&= \left\{\begin{array}{ll}\log\lp \frac{1+\Delta_2}{\Delta_2} + \frac{\SNR_{2p}}{(1+\INR_{1p})\Delta_2}\rp, &\SNR_2>\INR_1\\ \log\lp \frac{1+\Delta_2}{\Delta_2}\rp, &\SNR_2\le\INR_1\end{array}\right.\\ 
\label{eq_RateLoss}
&\overset{}{\le} \log\lp \frac{1+\Delta_2+\SNR_{2p}}{\Delta_2}\rp,
\end{align} 

Below we shall see why the intuition of quantizing at the private signal level works. By choosing $\Delta_2 = \max\{\SNR_{2p},1\}$, the ``rate loss" $\xi_1$ is upper bounded by a constant. In particular, when $\SNR_2 > \INR_1$, $\xi_1 \le \log3$; when $\SNR_2 \le \INR_1$, $\xi_1 = 1$. On the other hand, note that for receiver 1 the unwanted signal power level in $y_2$ is roughly $\max\{\SNR_{2p},1\}$, and hence replacing $\what{y}_2$ by $y_2$ gains at most a constant number of bits.

\begin{remark}
The above configuration of the distortion may not be optimal. The achievable rates can be further improved if we optimize over all possible distortions. For example, if the cooperative link capacity is huge, one could lower the distortion level to yield a finer description of received signals. With the above simple configuration, however, we are able to show that it achieves the capacity region to within a constant number of bits universally.
\end{remark}

\section{Characterization of the Capacity Region to within 2 Bits}\label{sec_AsymmConstantGap}

The main result in this section is the characterization of the capacity region to within 2 bits per user universally, regardless of channel parameters. To prove it, first we provide outer bounds of the capacity region. Ideas about how to prove them are outlined, and details are left in appendices. Then we make use of Theorem \ref{thm_CodingThm} to evaluate the achievable rate region, and show that it is within 2 bits per user to the proposed outer bounds.

\subsection{Outer Bounds}
To prove the outer bounds, the main idea is the following: first, upper bound the rates by mutual informations via Fano's inequality and data processing inequality; second, decompose them into two parts: (1) terms which are similar to those in Gaussian interference channels without cooperation, and (2) terms which correspond to the enhancement from cooperation. We use the genie-aided techniques in \cite{EtkinTse_07} to upper bound the first part and obtain namely the Z-channel bound (where the genie gives interfering symbols $x^N_j$ to receiver $i$, $i\ne j$) and ETW-bound (where the genie gives the interference term caused by user $i$ at receiver $j$, ${s}^N_i:=h_{ji}x^N_i+{z}^N_{j}$ to receiver $i$). For the second part, we make use of the fact that $u_{12}^N$ and $u_{21}^N$ are both functions of $(y_1^N,y_2^N)$, and other straightforward bounding techniques. The results are summarized in the following lemma.

%
\begin{lemma}\label{lem_OutBd}
$\mscr{C} \subseteq \overline{\mscr{C}}$, where $\overline{\mscr{C}}$ consists of nonnegative rate tuples $(R_1,R_2)$ satisfying
\begin{align}
\begin{split}
R_1 &\le \log(1+\SNR_1)+ \min\left\{ \C_{21},\log\left(1+\frac{\INR_2}{1+\SNR_1}\right) \right\}\\
R_2 &\le \log(1+\SNR_2)+ \min\left\{ \C_{12},\log\left(1+\frac{\INR_1}{1+\SNR_2}\right) \right\}
\end{split}\label{eq_CutSetBd} \\
R_1+R_2 &\le \log\left(1+\INR_1+\frac{\SNR_1}{1+\INR_2}\right) + \log\left(1+\INR_2+\frac{\SNR_2}{1+\INR_1}\right) + \C_{21}+\C_{12} \label{eq_ETWBound}\\
\begin{split}
R_1+R_2 &\le \log\left(1+\SNR_2+\INR_2\right) + \log\left(1+\frac{\SNR_1}{1+\INR_2}\right) + \C_{12}\\
R_1+R_2 &\le \log\left(1+\SNR_1+\INR_1\right) + \log\left(1+\frac{\SNR_2}{1+\INR_1}\right) + \C_{21}
\end{split} \label{eq_ZBound}\\
R_1+R_2 &\le \log\lp 1+\SNR_1+\SNR_2+\INR_1+\INR_2 + |h_{11}h_{22} - h_{12}h_{21}|^2 \rp \label{eq_SIMOBd}\\
\begin{split}
2R_1 + R_2 & \le \log\left(1+\INR_2+\frac{\SNR_2}{1+\INR_1}\right) + \log\left(1+\frac{\SNR_1}{1+\INR_2}\right)\\&\quad + \log\lp 1+ \SNR_1+\INR_1\rp + \C_{21}+\C_{12}\\
R_1 + 2R_2 & \le \log\left(1+\INR_1+\frac{\SNR_1}{1+\INR_2}\right) + \log\left(1+\frac{\SNR_2}{1+\INR_1}\right)\\&\quad + \log\lp 1+ \SNR_2+\INR_2\rp + \C_{12}+\C_{21}
\end{split} \label{eq_Slope2Bd}\\
\begin{split}
2R_1+R_2 &\le \log\lp 1+ \frac{\SNR_2}{1+\INR_1} + \INR_2 + \SNR_1 + \frac{\INR_1}{1+\INR_1} + \frac{|h_{11}h_{22} - h_{12}h_{21}|^2}{1+\INR_1}\rp\\
&\quad + \log\lp1+\SNR_1+\INR_1\rp + \C_{21}\\
R_1+2R_2 &\le \log\lp 1+ \frac{\SNR_1}{1+\INR_2} + \INR_1 + \SNR_2 + \frac{\INR_2}{1+\INR_2} + \frac{|h_{11}h_{22} - h_{12}h_{21}|^2}{1+\INR_2}\rp\\
&\quad + \log\lp1+\SNR_2+\INR_2\rp + \C_{12}
\end{split} \label{eq_Slope2CutSetBd}
\end{align}
\end{lemma}
\begin{proof}
Details are left in Appendix \ref{app_PfOutBd}. Below we give a short outline and intuitions. First of all, bounds (\ref{eq_CutSetBd}) and (\ref{eq_SIMOBd}) are straightforward cut-set upper bounds of individual rates and sum rate respectively. 

Bound (\ref{eq_ETWBound}) corresponds to the ETW-bound in Gaussian interference channels without cooperation. In the genie-aided channel, we upper bound the gain from receiver cooperation by $\C_{12}+\C_{21}$, that is, in both directions each bit is useful. 

Bounds (\ref{eq_ZBound}) correspond to the Z-channel bounds. In the genie-aided channel, since the genie gives interfering symbols $x^N_j$ to receiver $i$, $i\ne j$, there is no interference at receiver $i$. Intuitively, the cooperation from receiver $j$ to $i$ is now providing only the power gain, and the genie can provide $y_j^N$ to receiver $i$ to upper bound this power gain. The gain from the cooperation from receiver $i$ to $j$ is upper bounded by $\C_{ij}$. 

Bounds (\ref{eq_Slope2Bd}) on $R_i+2R_j$ are derived by giving side information ${s}_i^N$ to receiver $i$ and side information $x_i^N$ and $y_i^N$ to one of the receiver $j$'s. In the genie-aided channel there is an underlying Z-channel structure, and hence the gain from one direction of the cooperation is absorbed into a power gain. The rest is upper bounded by $\C_{12}+\C_{21}$.

Bounds (\ref{eq_Slope2CutSetBd}) on $R_i+2R_j$ are derived by giving side information $y_j^N$ and $\wtild{s}_i^N:=h_{ji}x^N_i+\wtild{z}^N_{j}$, where $\wtild{z}_j\sim\mcal{CN}(0,1)$ and independent of everything else, to receiver $i$ and side information $y_i^N$ to one of the receiver $j$'s. In the genie-aided channel, there is an underlying point-to-point MIMO channel, and hence the gain from both directions of cooperation is absorbed into the MIMO system. The rest is upper bounded by $\C_{ij}$.

Note that the derivation of all bounds are irrelevant to the relations among $\INR$'s and $\SNR$'s.
\end{proof}

We make the following observations:

\begin{remark}[Dependence on phases]
The sum-rate cut-set bound (\ref{eq_SIMOBd}) not only depends on $\SNR$'s and $\INR$'s but also on the phases of channel coefficients, due to the term $|h_{11}h_{22} - h_{12}h_{21}|^2$.
In particular, when the receiver-cooperative link capacities $\C$'s are large, the two receivers become near-fully cooperated, and the system performance is constrained by that of the SIMO MAC, that is, it enters the saturation region. Therefore this bound becomes active and the outer bound depends on phases.
\end{remark}

\begin{remark}[Strong interference regime]
When $\SNR_1 \le \INR_2$ and  $\SNR_2 \le \INR_1$, unlike the Gaussian interference channel of which the capacity region is equal to that of a compound MAC in the strong interference regime \cite{Sato_81}, here we cannot apply Sato's argument. Recall that when there is no cooperation, once user $i$'s own message is decoded successfully at receiver $i$, it can produce $\wtild{y}_j^N$ which has the same distribution as $y_j^N$. Since the error probability for decoding user $j$'s message at receiver $j$ only depends on the {\it marginal} distribution of $y_j^N$, it can be concluded that at receiver $i$ one can achieve the same performance for decoding user $j$'s message by using the same decoder as that in receiver $j$, and hence receiver $i$ can decode user $j$'s message successfully as well. When there is cooperation, however, the error probability for decoding user $j$'s message at receiver $j$ depends on the {\it joint} distribution of $(y_j^N, u_{ij}^N)$. Note that the additive noise terms in $\wtild{y}_j^N$ and $y_j^N$ have different correlations with the noise term $z_i^N$, and $u_{ij}^N$ can be highly correlated with $z_i^N$. As a consequence, the joint distributions of $(y_j^N, u_{ij}^N)$ and $(\wtild{y}_j^N, u_{ij}^N)$ are not guaranteed to be the same, and receiver $i$ may not be able to achieve the same performance for decoding user $j$'s message by using the same decoder as that in receiver $j$. Therefore, we cannot claim that the capacity region under strong interference condition is the same as that of compound MAC with conferencing decoders (CMAC-CD). Instead, we take the Z-channel bound (\ref{eq_ZBound}), which is within 1 bit to the sum rate cut-set bound of CMAC-CD in strong interference regimes. This will be discussed in the last part of this section.
\end{remark}

\subsection{Capacity Region to within 2 bits}

Subsequently we investigate three qualitatively different cases, namely, weak interference, mixed interference, and strong interference\footnote{
We distinguish the general set-up into three qualitatively different cases: (1) weak interference, where $\SNR_1 > \INR_2$ and $\SNR_2 > \INR_1$; (2) mixed interference, where $\SNR_1 > \INR_2$ and $\SNR_2 \le \INR_1$; (3) strong interference, where $\SNR_1 \le \INR_2$ and $\SNR_2 \le \INR_1$.}, in the rest of this section. We summarize the main achievability result in the following theorem: (recall that $\overline{\mscr{C}}$ is the outer bound region defined in Lemma \ref{lem_OutBd})
\begin{theorem}[Within constant gap to capacity region]\label{thm_ApproxCapacity}
\begin{align}
\mscr{R} \subseteq \mscr{C} \subseteq \overline{\mscr{C}} \subseteq \mscr{R} \oplus \big( [0,2]\times[0,2] \big),
\end{align}
\end{theorem}
\begin{proof}
Proved by Lemma \ref{lem_WeakGap}, \ref{lem_MixedGap}, and \ref{lem_StrongGap} in the rest of this section.
\end{proof}

\subsection{Weak interference}
In the case $\SNR_1 > \INR_2$ and $\SNR_2 > \INR_1$, the superposition coding configuration is to split message $m_i$ into $m_{ic}$ and $m_{ip}$, for both users $i=1,2$. We first consider $\STG_{2\ra1\ra2}$: referring to Theorem \ref{thm_CodingThm}, we obtain the set of achievable rates $\lp R_{1c},R_{2c},R_{1p},R_{2p}\rp$. The term $\xi_1 \le \log 3 \approx 1.59$ bits, due to \eqref{eq_RateLoss} in Section \ref{subsec_AchievableRate} and the chosen distortion $\Delta_2 = \max\{\SNR_{2p},1\}$.


To simplify calculations, note that the right-hand-side of \eqref{eq_PrivRate} is at most a constant number of bits greater that its lower bound $I\lp x_{1};y_1|x_{1c},x_{2c}\rp$, the right-hand-side of \eqref{eq_CommRate} is at most a constant number of bits greater that its lower bound $I\lp x_{2c};y_1|x_{1}\rp$, and the right-hand-side of \eqref{eq_User1Rate} is at most a constant number of bits greater that its lower bound $I\lp x_{1};y_1|x_{2c}\rp$. Therefore, we replace these three constraints by
\begin{align}
R_{1p} &\le I\lp x_{1};y_1|x_{1c},x_{2c}\rp\\
R_{2c} &\le I\lp x_{2c};y_1|x_{1}\rp\\
R_{1c}+R_{1p} &\le I\lp x_{1};y_1|x_{2c}\rp
\end{align}
in the following calculations. Next, rewriting $R_{ip} = R_i - R_{ic}$ for $i=1,2$, applying Fourier-Motzkin algorithm to eliminate $R_{1c}$ and $R_{2c}$, and removing redundant terms (details omitted here), we obtain an achievable $\mscr{R}_{2\ra1\ra2}$, which consists of nonnegative $(R_1,R_2)$ satisfying:
\begin{align}
R_1 &\le \min \big\{ I\lp x_1; y_1| x_{2c}\rp, I\lp x_1;y_1| x_{1c},x_{2c}\rp + I\lp x_{1c},x_2; y_2| x_{2c}\rp +\C_{12} \big\}\\
R_2 &\le \min \big\{ I\lp x_2; y_2| x_{1c}\rp + \C_{12}, I\lp x_{2c}; y_1| x_1\rp + I\lp x_2; y_2| x_{1c}, x_{2c}\rp \big\}\\
R_1+R_2 &\le I\lp x_1,x_{2c}; y_1\rp + I\lp x_2; y_2| x_{1c},x_{2c}\rp + \lp\C_{21}-\xi_1\rp^+  \label{eq_SumBd1}\\
R_1+R_2 &\le I\lp x_1,x_{2c}; y_1,\what{y}_2\rp + I\lp x_2; y_2| x_{1c},x_{2c}\rp  \label{eq_SumBd2}\\
R_1+R_2 &\le I\lp x_1,x_{2c}; y_1| x_{1c}\rp + I\lp x_{1c},x_2; y_2| x_{2c}\rp + \C_{12} + \lp\C_{21}-\xi_1\rp^+ \label{eq_SumBd3}\\
R_1+R_2 &\le I\lp x_1,x_{2c}; y_1,\what{y}_2| x_{1c}\rp + I\lp x_{1c},x_2; y_2| x_{2c}\rp + \C_{12} \label{eq_SumBd4}\\
R_1+R_2 &\le I\lp x_1; y_1| x_{1c},x_{2c}\rp + I\lp x_{1c},x_2; y_2\rp + \C_{12} \label{eq_SumBd5}\\
R_1+R_2 &\le I\lp x_1; y_1| x_{1c},x_{2c}\rp + I\lp x_{2c};y_1|x_1\rp + I\lp x_{1c},x_2; y_2| x_{2c}\rp + \C_{12} \label{eq_SumBd6}\\
2R_1+R_2 &\le I\lp x_1,x_{2c}; y_1\rp + I\lp x_1; y_1| x_{1c},x_{2c}\rp + I\lp x_{1c},x_2; y_2| x_{2c}\rp + \C_{12} + \lp\C_{21}-\xi_1\rp^+\\ 
2R_1+R_2 &\le I\lp x_1,x_{2c}; y_1,\what{y}_2\rp + I\lp x_1; y_1| x_{1c},x_{2c}\rp + I\lp x_{1c},x_2; y_2| x_{2c}\rp + \C_{12}  \label{eq_Trouble}
\\ 
R_1+2R_2 &\le I\lp x_1,x_{2c}; y_1| x_{1c}\rp + I\lp x_{1c},x_2; y_2\rp + I\lp x_2; y_2| x_{1c},x_{2c}\rp + \C_{12} + \lp\C_{21}-\xi_1\rp^+\\
R_1+2R_2 &\le I\lp x_1,x_{2c}; y_1| x_{1c}\rp + I\lp x_{2c};y_1|x_1\rp + I\lp x_{1c},x_2; y_2| x_{2c}\rp + I\lp x_2; y_2| x_{1c},x_{2c}\rp\\&\quad + \C_{12} + \lp\C_{21}-\xi_1\rp^+\\
R_1+2R_2 &\le I\lp x_1,x_{2c}; y_1,\what{y}_2| x_{1c}\rp + I\lp x_{1c},x_2; y_2\rp + I\lp x_2; y_2| x_{1c},x_{2c}\rp + \C_{12}\\
R_1+2R_2 &\le I\lp x_1,x_{2c}; y_1,\what{y}_2| x_{1c}\rp + I\lp x_{2c};y_1|x_1\rp + I\lp x_{1c},x_2; y_2| x_{2c}\rp\\&\quad + I\lp x_2; y_2| x_{1c},x_{2c}\rp + \C_{12}
\end{align}

We will show that except \eqref{eq_Trouble}, all bounds are within a constant number of bits to the corresponding outer bounds in Lemma \ref{lem_OutBd}. By symmetry, however, one can write down $\mscr{R}_{1\ra2\ra1}$ and see that it can be compensated by time-sharing with rate points in $\mscr{R}_{1\ra2\ra1}$. Therefore the resulting $\mscr{R} := \mathrm{conv}\lbp \mscr{R}_{2\ra1\ra2} \cup \mscr{R}_{1\ra2\ra1}\rbp$ is within a constant number of bits to the outer bounds in Lemma \ref{lem_OutBd}. An illustration is provided in Fig. \ref{fig_Timeshare}.
\begin{figure}[hbtp]
{\center
\subfigure[Taking union is required, while time-sharing is not]{\includegraphics[width=3in]{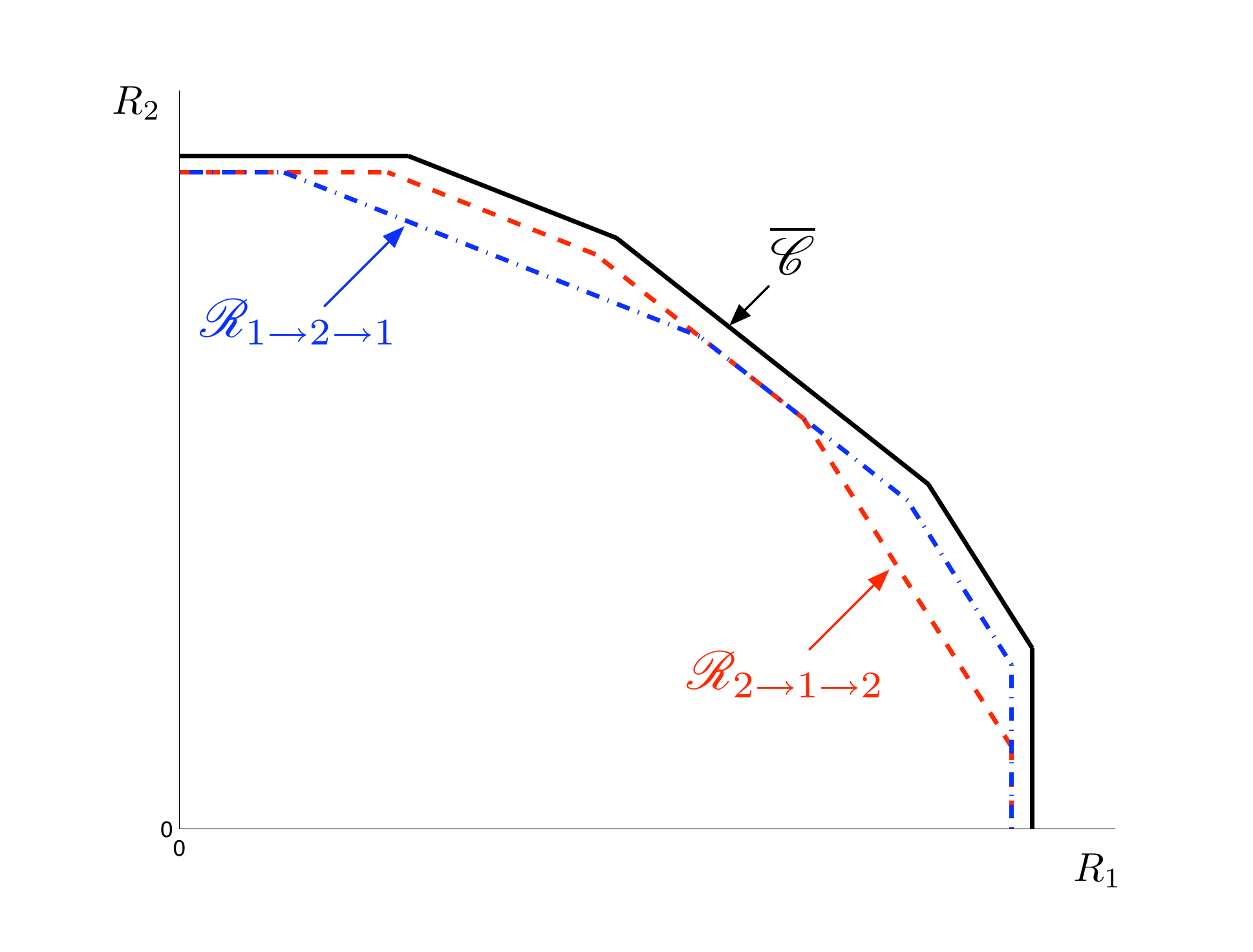}}
\subfigure[Time-sharing is required]{\includegraphics[width=3in]{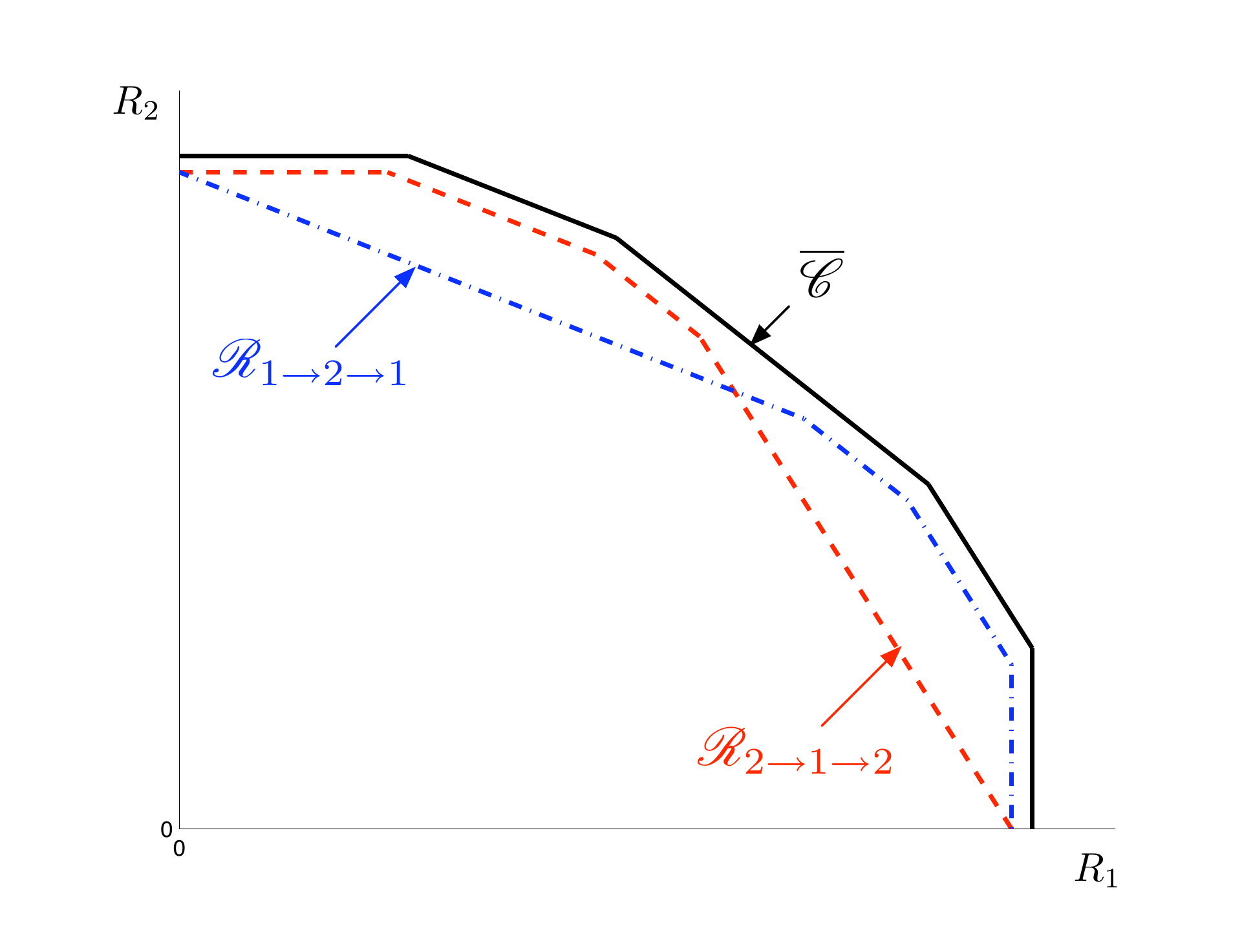}}
\caption{Time-sharing to achieve approximate capacity region}
\label{fig_Timeshare}
}
\end{figure}

We give the following lemma.
\begin{lemma}[Rate region in the weak interference regime]\label{lem_WeakGap}
\begin{align}
\mscr{R} \subseteq \mscr{C} \subseteq \overline{\mscr{C}} \subseteq \mscr{R} \oplus \big( [0,2]\times[0,2] \big),
\end{align}
in the weak interference regime.
\end{lemma}
\begin{proof}
We need the following claims:
\begin{claim}\label{claim_Corner}
In $\mscr{R}_{2\ra1\ra2}$, whenever the $2R_1+R_2$ bound \eqref{eq_Trouble} is active,
\begin{itemize}
\item [(a)] if $R_1+2R_2$ bounds are active, the corner point where $R_1+R_2$ bound and $R_1+2R_2$ bound intersect can be achieved;
\item [(b)] if $R_1+2R_2$ bounds are not active, the corner point where $R_1+R_2$ bound and $R_2$ bound intersect can be achieved.
\end{itemize}

Above two situations are illustrated in Fig. \ref{fig_Region}.
\end{claim}
\begin{proof}
See Appendix \ref{app_PfClaims}.
\end{proof}

\begin{figure}[hbtp]
{\center
\subfigure[$R_1+2R_2$ bound is active]{\includegraphics[width=3in]{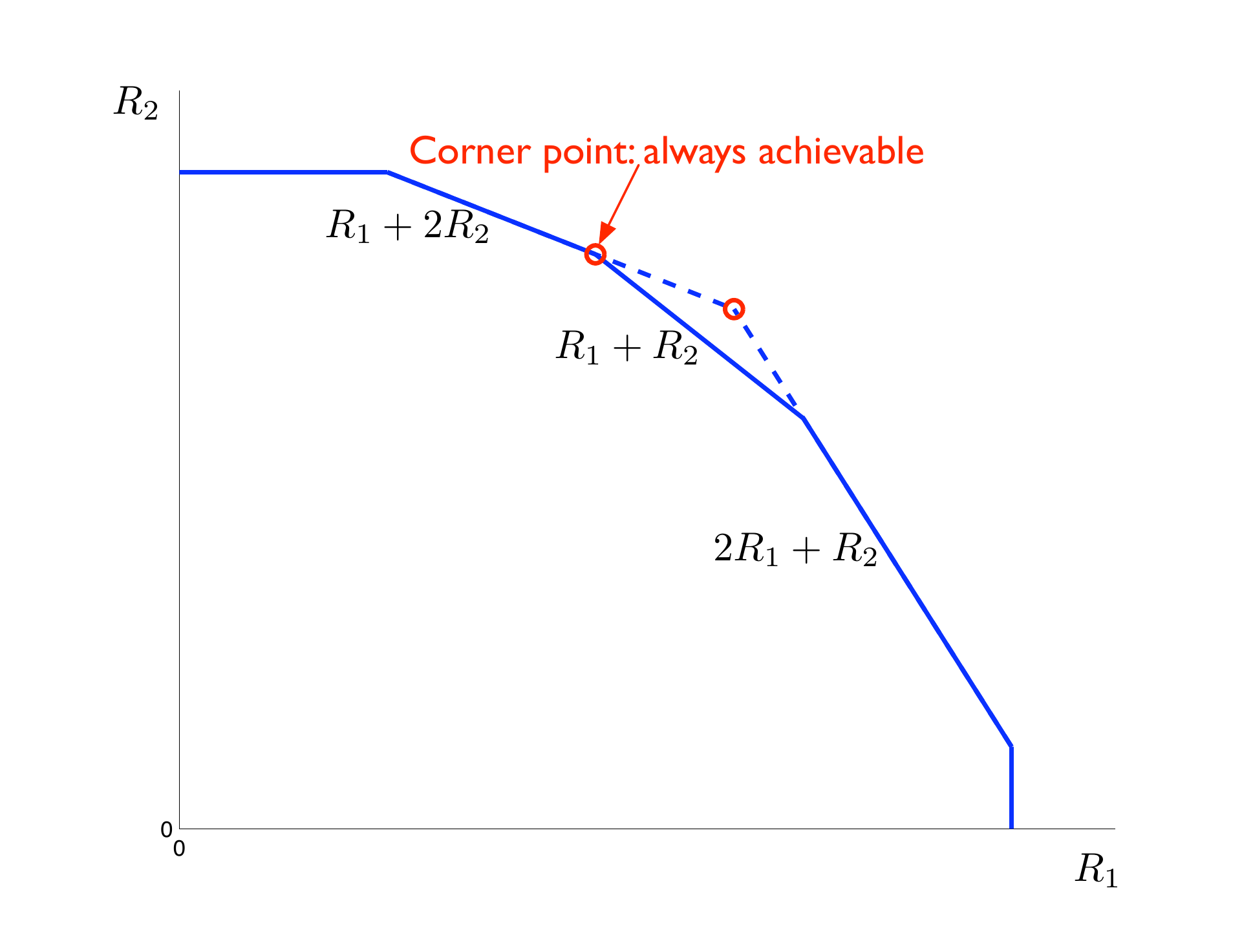}}
\subfigure[$R_1+2R_2$ bound is not active]{\includegraphics[width=3in]{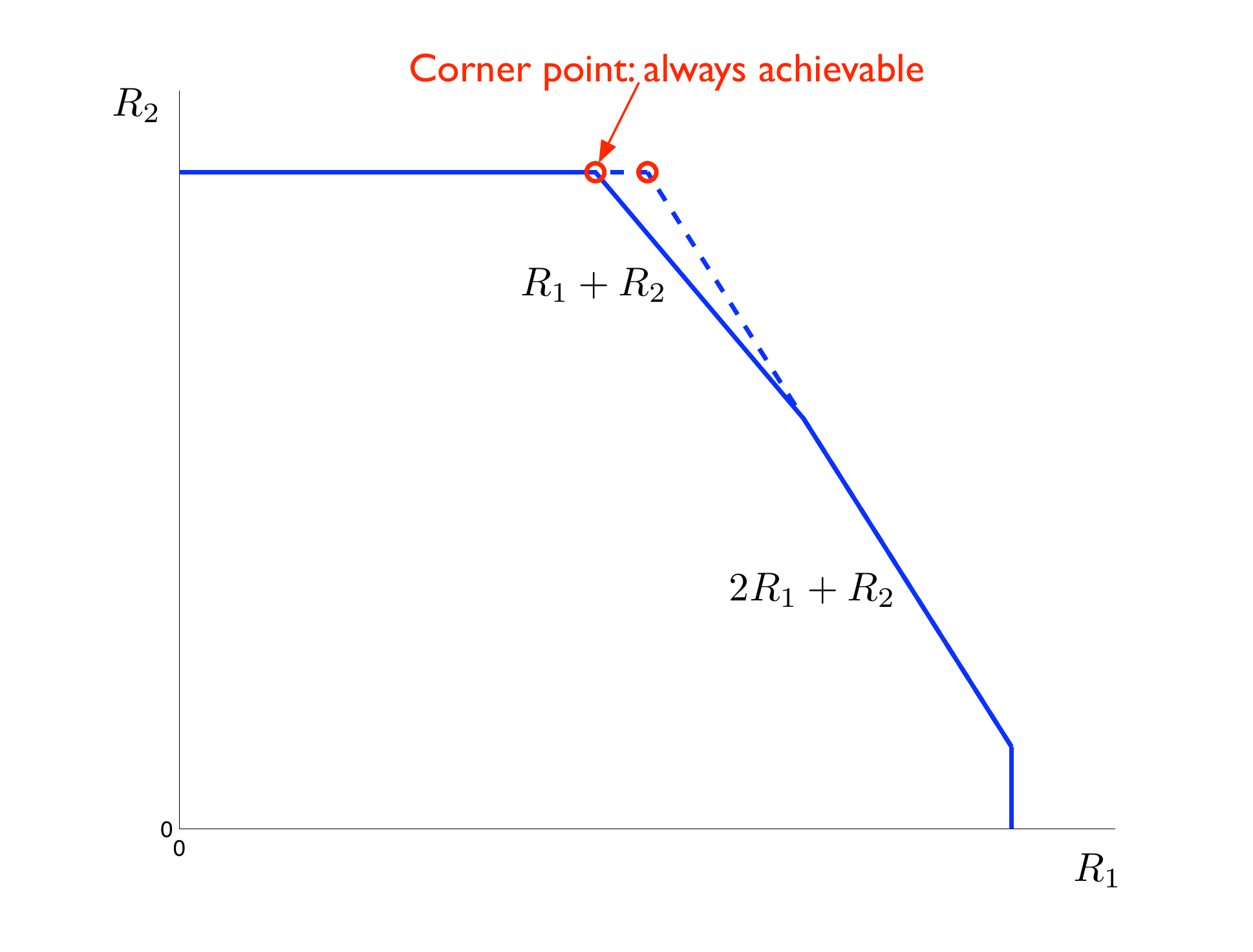}}
\caption{Situations in $\mscr{R}_{2\ra1\ra2}$}
\label{fig_Region}
}
\end{figure}

Therefore, the $2R_1+R_2$ bound \eqref{eq_Trouble} and, by symmetry, its corresponding $R_1+2R_2$ bound in $\mscr{R}_{1\ra2\ra1}$ do not show up in $\mscr{R} = \mathrm{conv}\big\{ \mscr{R}_{2\ra1\ra2} \cup \mscr{R}_{1\ra2\ra1}\big\}$, and $\mscr{R}$ is within 2 bits per user to the outer bounds in Lemma \ref{lem_OutBd}. To show this, we first look at the bounds in $\mscr{R}_{2\ra1\ra2}$ except \eqref{eq_Trouble}. We claim that

\begin{claim}\label{claim_WeakGap}
The bounds in $\mscr{R}_{2\ra1\ra2}$ except \eqref{eq_Trouble} satisfies:
\begin{itemize}
\item $R_1$ bound is within $2$ bits to outer bounds;
\item $R_2$ bound is within $2$ bits to outer bounds;
\item $R_1+R_2$ bound is within $\log 12$ bits to outer bounds;
\item $2R_1+R_2$ bound is within $\log 24$ bits to outer bounds;
\item $R_1+2R_2$ bound is within $\log 48$ bits to outer bounds.
\end{itemize}
\end{claim}
\begin{proof}
See Appendix \ref{app_PfClaims}.
\end{proof}

By symmetry, we obtain similar results for $\mscr{R}_{1\ra2\ra1}$, and hence conclude that the bounds in $\mscr{R}$ satisfies:
\begin{itemize}
\item $R_1$ bound is within $2$ bits to outer bounds;
\item $R_2$ bound is within $2$ bits to outer bounds;
\item $R_1+R_2$ bound is within $\log 12$ bits to outer bounds;
\item $2R_1+R_2$ bound is within $\log 48$ bits to outer bounds;
\item $R_1+2R_2$ bound is within $\log 48$ bits to outer bounds.
\end{itemize}

This completes the proof.
\end{proof}

\subsection{Mixed interference}
In the case $\SNR_1 > \INR_2$ and $\SNR_2 \le \INR_1$, the superposition coding configuration is to split message $m_1$ into $m_{1c}$ and $m_{1p}$, while making the whole $m_2$ to be common. We first consider $\STG_{2\ra1\ra2}$: by Theorem \ref{thm_CodingThm}, rates satisfying the following are achievable:
\begin{align}
R_{1p} &\le \min \lbp I\lp x_{1};y_1|x_{1c},x_{2}\rp + (\C_{21} - \xi_1)^+, I\lp x_{1};y_1,\what{y}_2 | x_{1c},x_{2}\rp \rbp \label{eq_r1p}\\
R_{2} &\le \min \lbp I\lp x_{2};y_1|x_{1}\rp + (\C_{21} - \xi_1)^+, I\lp x_{2};y_1,\what{y}_2 | x_{1}\rp \rbp \label{eq_r2} \\
R_{2}+R_{1p} &\le \min \lbp I\lp x_{2},x_{1};y_1|x_{1c}\rp + (\C_{21} - \xi_1)^+, I\lp x_{2},x_{1};y_1,\what{y}_2 | x_{1c}\rp \rbp\\
R_{1c}+R_{1p} &\le \min \lbp I\lp x_{1};y_1|x_{2}\rp + (\C_{21} - \xi_1)^+, I\lp x_{1};y_1,\what{y}_2 | x_{2}\rp \rbp \label{eq_r1} \\
R_{1c}+R_{2}+R_{1p} &\le \min \lbp I\lp x_{1},x_{2};y_1 \rp + (\C_{21} - \xi_1)^+, I\lp x_{1},x_{2};y_1,\what{y}_2 \rp \rbp\\
R_{1c} &\le I\lp x_{1c};y_2|x_{2}\rp + \C_{12}\\
R_{2} &\le I\lp x_{2};y_2|x_{1c}\rp + \C_{12}\\
R_{2}+R_{1c} &\le I\lp x_{2},x_{1c};y_2 \rp + \C_{12},
\end{align}
where $\xi_1 = 1$ since $\SNR_2 \le \INR_1$.

Again to simplify calculations, note that the right-hand-side of \eqref{eq_r1p} is at most a constant number of bits greater that its lower bound $I\lp x_{1};y_1|x_{1c},x_{2}\rp$, the right-hand-side of \eqref{eq_r2} is at most a constant number of bits greater that its lower bound $I\lp x_{2};y_1|x_{1}\rp$, and the right-hand-side of \eqref{eq_r1} is at most a constant number of bits greater that its lower bound $I\lp x_{1};y_1|x_{2}\rp$. Therefore, we replace these three constraints by
\begin{align}
R_{1p} &\le I\lp x_{1};y_1|x_{1c},x_{2}\rp\\
R_{2} &\le I\lp x_{2};y_1|x_{1}\rp\\
R_{1c}+R_{1p} &\le I\lp x_{1};y_1|x_{2}\rp
\end{align}
in the following calculations. Next, rewriting $R_{1p} = R_1 - R_{1c}$, applying Fourier-Motzkin algorithm to eliminate $R_{1c}$, and removing redundant terms (details omitted here), we obtain an achievable $\mscr{R}_{2\ra1\ra2}$, consists of nonnegative $(R_1,R_2)$ satisfying:
\begin{align}
R_1 &\le \min\big\{ I\lp x_1; y_1| x_2\rp , I\lp x_1; y_1| x_{1c},x_2\rp + I\lp x_{1c}; y_2| x_2\rp + \C_{12} \big\}\\
R_2 &\le \min\big\{ I\lp x_2; y_1| x_1\rp , I\lp x_2; y_2| x_{1c}\rp + \C_{12} \big\}\\
R_1+R_2 &\le I\lp x_1,x_2; y_1\rp + \lp\C_{21}-\xi_1\rp^+\\
R_1+R_2 &\le I\lp x_1,x_2; y_1,\what{y}_2\rp\\
R_1+R_2 &\le I\lp x_1; y_1|x_{1c},x_2\rp + I\lp x_{1c},x_2; y_2\rp + \C_{12}\\
R_1+R_2 &\le I\lp x_1,x_2; y_1| x_{1c}\rp + I\lp x_{1c}; y_2| x_2\rp + \C_{12} + \lp\C_{21}-\xi_1\rp^+\\
R_1+R_2 &\le I\lp x_1,x_2; y_1,\what{y}_2| x_{1c}\rp + I\lp x_{1c}; y_2| x_2\rp + \C_{12}\\
R_1+2R_2 &\le I\lp x_1,x_2; y_1| x_{1c}\rp + I\lp x_{1c},x_2; y_2\rp + \C_{12} + \lp\C_{21}-\xi_1\rp^+\\
R_1+2R_2 &\le I\lp x_1,x_2; y_1,\what{y}_2| x_{1c}\rp + I\lp x_{1c},x_2; y_2\rp + \C_{12}
\end{align}


Comparing $\mscr{R}_{2\ra1\ra2}$ with the outer bounds in Lemma \ref{lem_OutBd}, one can easily conclude that
\begin{lemma}[Mixed interference rate region]\label{lem_MixedGap}
\begin{align}
\mscr{R}_{2\ra1\ra2} \subseteq \mscr{C} \subseteq \overline{\mscr{C}} \subseteq \mscr{R}_{2\ra1\ra2} \oplus \big( [0,1.5]\times[0,1.5] \big), 
\end{align}
in the mixed interference regime. Besides, $\mscr{R}_{2\ra1\ra2} \subseteq \mscr{R}$.
\end{lemma}
\begin{proof}

We investigate the bounds in $\mscr{R}_{2\ra1\ra2}$ and claim that
\begin{claim}\label{claim_MixedGap}
The bounds in $\mscr{R}_{2\ra1\ra2}$ satisfies:
\begin{itemize}
\item $R_1$ bound is within $1$ bit to outer bounds;
\item $R_2$ bound is within $1$ bit to outer bounds;
\item $R_1+R_2$ bound is within $3$ bits to outer bounds;
\item $R_1+2R_2$ bound is within $3$ bits to outer bounds.
\end{itemize}
\end{claim}
\begin{proof}
See Appendix \ref{app_PfClaims}
\end{proof}

This completes the proof.
\end{proof}

\subsection{Strong interference}\label{subsec_StrongAchieve}
In the case $\SNR_1 \le \INR_2$ and $\SNR_2 \le \INR_1$, it turns out that a one-round strategy $\STG_{\OR}$ described below suffices to achieve capacity to within a constant number of bits. The transmission scheme is the same as that described in Section \ref{subsec_Transmission}. The difference is that, both receivers quantize-and-bins their received signals and decode with the help from the side information, as described in Section \ref{subsec_CoopProtocol}. It is called one-round since both receivers decode after one-round exchange of informaion. Below is the coding theorem for this strategy:


\begin{theorem}\label{thm_ErrProb}
The rate tuple $\lp R_{1c},R_{2c},R_{1p},R_{2p}\rp$ satisfying the following constraints are achievable for $\STG_{\OR}$:
{\flushleft \emph{\underline{Constraints at receiver 1}}:}
\begin{align}
R_{1p} &\le \min \lbp I\lp x_{1};y_1|x_{1c},x_{2c}\rp + (\C_{21} - \xi_1)^+, I\lp x_{1};y_1,\what{y}_2 | x_{1c},x_{2c}\rp \rbp \\
R_{2c}+R_{1p} &\le \min \lbp I\lp x_{2c},x_{1};y_1|x_{1c}\rp + (\C_{21} - \xi_1)^+, I\lp x_{2c},x_{1};y_1,\what{y}_2 | x_{1c}\rp \rbp\\
R_{1c}+R_{1p} &\le \min \lbp I\lp x_{1};y_1|x_{2c}\rp + (\C_{21} - \xi_1)^+, I\lp x_{1};y_1,\what{y}_2 | x_{2c}\rp \rbp\\
R_{1c}+R_{2c}+R_{1p} &\le \min \lbp I\lp x_{1},x_{2c};y_1 \rp + (\C_{21} - \xi_1)^+, I\lp x_{1},x_{2c};y_1,\what{y}_2 \rp \rbp
\end{align}

{\flushleft \emph{\underline{Constraints at receiver 2}}: the above constraints with index ``1" and ``2" exchanged.}
\end{theorem}

\begin{proof}
Follows the same line as the proof of Theorem \ref{thm_CodingThm}. There is no rate constraint for $R_{jc}$ at receiver $i$ for $(i,j)=(1,2)$ or $(2,1)$, since decoding $m_{jc}$ incorrectly at receiver $i$ does not account for an error.
\end{proof}

Now, in the strong interference regime, the superposition coding configuration is to make whole message $m_i$ be common for both users $i=1,2$; in words, there is no superposition eventually. One-round strategy $\STG_{\OR}$ yields achievable rate region $\mscr{R}_{\OR}$, which consists of nonnegative $(R_1,R_2)$ satisfying
\begin{align}
\begin{split}
R_{2} &\le \min \lbp I\lp x_{2};y_1|x_{1}\rp + (\C_{21} - \xi_1)^+, I\lp x_{2};y_1,\what{y}_2 | x_{1}\rp \rbp\\
R_{1} &\le \min \lbp I\lp x_{1};y_1|x_{2}\rp + (\C_{21} - \xi_1)^+, I\lp x_{1};y_1,\what{y}_2 | x_{2}\rp \rbp\\
R_{1}+R_{2} &\le \min \lbp I\lp x_{1},x_{2};y_1 \rp + (\C_{21} - \xi_1)^+, I\lp x_{1},x_{2};y_1,\what{y}_2 \rp \rbp\\
R_{1} &\le \min \lbp I\lp x_{1};y_2|x_{2}\rp + (\C_{12} - \xi_2)^+, I\lp x_{1};y_2,\what{y}_1 | x_{2}\rp \rbp\\
R_{2} &\le \min \lbp I\lp x_{2};y_2|x_{1}\rp + (\C_{12} - \xi_2)^+, I\lp x_{2};y_2,\what{y}_1 | x_{1}\rp \rbp\\
R_{2}+R_{1} &\le \min \lbp I\lp x_{2},x_{1};y_2 \rp + (\C_{12} - \xi_2)^+, I\lp x_{2},x_{1};y_2,\what{y}_1 \rp \rbp,
\end{split}\label{eq_StrongRegion}
\end{align}
where $\xi_i = 1$, for both $i=1,2$.



Comparing $\mscr{R}_{\OR}$ with the outer bounds in Lemma \ref{lem_OutBd}, one can easily conclude that
\begin{lemma}[Strong interference rate region]\label{lem_StrongGap}
\begin{align}
\mscr{R}_{\OR} \subseteq \mscr{C} \subseteq \overline{\mscr{C}} \subseteq \mscr{R}_{\OR} \oplus \big( [0,1]\times[0,1] \big),
\end{align}
in the strong interference regime. Besides, $\mscr{R}_{\OR} \subseteq \mscr{R}$.
\end{lemma}
\begin{proof}
We investigate the bounds in $\mscr{R}_{\OR}$ and claim that:
\begin{claim}\label{claim_StrongGap}
The bounds in $\mscr{R}_{\OR}$ satisfies:
\begin{itemize}
\item $R_1$ bound is within $1$ bit to outer bounds;
\item $R_2$ bound is within $1$ bit to outer bounds;
\item $R_1+R_2$ bound is within $2$ bits to outer bounds.
\end{itemize}
\end{claim}
\begin{proof}
See Appendix \ref{app_PfClaims}.
\end{proof}

This completes the proof.
\end{proof}

\subsection{Approximate Capacity of Compound MAC with Conferencing Decoders}\label{subsec_CapCMAC}
One of the contribution in this work is characterizing the capacity region of {\it compound multiple access channel with conferencing decoders} (CMAC-CD) to within 1 bit. The channel is defined as follows.
\begin{definition} 
A compound multiple access channel with conferencing decoders (CMAC-CD), is a channel with the same set-up as depicted in Fig. \ref{fig_ChModel}., while both receivers aim to decode both $m_1$ and $m_2$.
\end{definition}

We give straightforward cut-set upper bounds as follows:
\begin{lemma}\label{lem_CMAC_OutBd}
If $(R_1,R_2)$ is achievable, it must satisfy the following constraints:
\begin{align}
R_1 &\le \min\left\{ \log(1+\SNR_1)+ \C_{21},\log(1+\INR_2)+ \C_{12},\log\left(1+\SNR_1+\INR_2\right) \right\}\\
R_2 &\le \min\left\{ \log(1+\SNR_2)+ \C_{12},\log(1+\INR_1)+ \C_{21},\log\left(1+\SNR_2+\INR_1\right) \right\}\\
R_1+R_2 &\le \log\left(1+\SNR_1+\INR_1\right) + \C_{21}\\
R_1+R_2 &\le \log\left(1+\SNR_2+\INR_2\right) + \C_{12}\\
R_1+R_2 &\le \log\lp 1+\SNR_1+\INR_1+\SNR_2+\INR_2 + |h_{11}h_{22} - h_{12}h_{21}|^2 \rp.
\end{align}
\end{lemma}
\begin{proof} These are straightforward cut-set bounds. We omit the details here. \end{proof}

For achievability, we adapt the scheme proposed above with no superposition coding at transmitters. Therefore, the rate region is exactly the same as \eqref{eq_StrongRegion}. Hence, we conclude that 
\begin{theorem}[Within 1 bit to CMAC-CD Capacity Region]
The scheme achieves the capacity of compound MAC with conferencing decoders to within 1 bit.
\end{theorem}
\begin{proof}
Following the same line in the proof of Lemma \ref{lem_StrongGap}, we can conclude that the bounds in $\mscr{R}_{\OR}$ satisfies:
\begin{itemize}
\item $R_1$ bound is within $1$ bit to outer bounds;
\item $R_2$ bound is within $1$ bit to outer bounds;
\item $R_1+R_2$ bound is within $1$ bit to outer bounds.
\end{itemize}

This completes the proof.
\end{proof}

\section{One-Round Strategy versus Two-Round Strategy}
In Section \ref{sec_AsymmConstantGap} we show that the two-round strategy proposed in Section \ref{sec_Achievable} along with time-sharing achieves the capacity region to within 2 bits universally. One of the drawbacks of the two-round strategy, however, is the round-trip delay. The quantize-binning receiver cannot proceed to decoding until the other receiver decodes and forwards the bin indices back. The round-trip delay is two times the block length, which can be huge. To avoid such huge delay, fortunately in some cases, the one-round strategy $\STG_{\OR}$ described in Section \ref{subsec_StrongAchieve} suffices. One of such cases is the strong interference regime. This can be easily justified in the corresponding linear deterministic channel (LDC). At strong interference, all transmitted signals in the LDC are common. There is no useful information lies below the noise level since the signal is corrupted by the noise. Hence, quantize-binning at the noise level is sufficient to convey the useful information.

Another such cases is the symmetric set-up, where
\begin{align}
&\SNR = \SNR_1 = \SNR_2,\ \INR = \INR_1 = \INR_2;\ \C = \C_{12} = \C_{21}.
\end{align}

For the symmetric set-up, a natural performance measure is the symmetric capacity, defined as follows:
\begin{definition}[Symmetric Capacity]
\begin{align}
C_{\sym} &:= \sup \left\{ R: (R,R) \in \mscr{C}\right\}.
\end{align}
\end{definition}

It turns out that the one-round strategy suffices to achieve $C_{\sym}$ to within a constant number of bits.
\begin{theorem}[Constant Gap to the Symmetric Capacity]\label{thm_ORSymm}
$ $\par
The one-round strategy $\STG_{\OR}$ can achieve the symmtric capacity to within 3 bits.
\end{theorem}
\begin{proof}
See Appendix \ref{app_PfORSymm}.
\end{proof}

The justification in the corresponding LDC is again simple. Since the performance measure in which we are interested is the symmetric capacity, we can without loss of generality assume that both transmitters are transmitting at full private rate, that is, the entropy of each user's private signals is equal to the number of levels below the private signal level. Therefore at each receiver, there is no useful information below the private signal level, and quantize-binning at the private signal level suffices to convey the useful information.

\section{Generalized Degrees of Freedom Characterization}\label{sec_Dof}
With the characterization of the capacity region to within a constant number of bits, we attempt to answer the original fundamental question: how much interference can one bit of receiver cooperation mitigate? For simplicity, we consider the symmetric set-up.


By Lemma \ref{lem_OutBd} and Theorem \ref{thm_ApproxCapacity}, we have the characterization of the symmetric capacity to within 2 bits:
\begin{corollary}[Approximate Symmetric Capacity]\label{cor_SymBd}
Let $\overline{C}_{\sym}$ be the minimum of the below four terms:
\begin{align}
&\log(1+\SNR)+ \min\left\{ \C,\log\left(1+\frac{\INR}{1+\SNR}\right) \right\},\\
&\log\left(1+\INR+\frac{\SNR}{1+\INR}\right) + \C,\\
&\frac{1}{2}\log\left(1+\SNR+\INR\right) + \frac{1}{2}\log\left(1+\frac{\SNR}{1+\INR}\right) + \frac{1}{2}\C,\\
&\frac{1}{2}\log\lp 1+2\SNR+2\INR + |h_{11}h_{22} - h_{12}h_{21}|^2 \rp.
\end{align}

Then, $\overline{C}_{\sym}-2 \le C_{\sym} \le \overline{C}_{\sym}$.
\end{corollary}



\subsection{Generalized Degrees of Freedom}
To study the behavior of the system performance in the linear region, we use the notion of \emph{generalized degrees of freedom} (g.d.o.f.), which is originally proposed in \cite{EtkinTse_07}. A natural extension from the definition in \cite{EtkinTse_07} would be the following: let
\begin{align}
&\lim_{\SNR\rightarrow\infty}\frac{\log\INR}{\log\SNR} = \alpha;\ \lim_{\SNR\rightarrow\infty}\frac{\C}{\log\SNR} = \kappa,
\end{align}
and define the number of generalized degrees of freedom per user as
\begin{align}\label{eq_Dof}
d := \lim_{\begin{subarray}{c} \mathrm{fix}\ \alpha,\kappa\\ \SNR\rightarrow\infty\end{subarray}}\frac{C_{\sym}}{\log\SNR},
\end{align}
if the limit exists. With fixed $\alpha$ and $\kappa$, however, there are certain channel realizations under which \eqref{eq_Dof} has different values and hence the limit does not exist. This happens when $\alpha = 1$, where the phases of the channel gains matter both in inner and outer bounds. In particular, its value can depend on whether the system MIMO matrix is well-conditioned or not.



From the above discussion we see that the limit does not exist, since for different channel phases and different $\INR$ settings the value of \eqref{eq_Dof} may be different. The reason is that, the original notion proposed in \cite{EtkinTse_07} cannot capture the impact of {\it phases} in MIMO situations, while from Corollary \ref{cor_SymBd} we see that our results depend on phases heavily, if the receiver-cooperative link capacity $\C$ is so large that MIMO sum-rate cut-set bound becomes active. Therefore, instead of claiming that the limit (\ref{eq_Dof}) exists for {\it all} channel realizations, we pose a reasonable distribution, namely, i.i.d. uniform distribution, on the phases, show that the limit exists {\it almost surely}, and define the limit to be the number of {\it generalized degrees of freedom} per user.

\begin{lemma}\label{prop_ConvDof}
Let 
\begin{align}
|h_{ij}| = g_{ij},\ \angle h_{ij}=\Theta_{ij},\ \forall i,j\in\{1,2\},
\end{align}
where $\Theta_{ij}$ are i.i.d. uniformly distributed over $[0,2\pi]$. Then the limit (\ref{eq_Dof}) exists almost surely, and is defined as the number of generalized degrees of freedom (per user) in the system.
\end{lemma}
\begin{proof}
We leave the proof in Appendix \ref{app_PfConvDof}.
\end{proof}

Now that the number of g.d.o.f. is well-defined, we can give the following theorem:
\begin{theorem}[Number of Generalized Degrees of Freedom Per User]\label{thm_Dof}
We have a direct consequence from Corollary \ref{cor_SymBd}:\par
For $0 \le \alpha < 1$,
\begin{align}
d =  \min\left\{ 1, \max\lp \alpha,1-\alpha\rp + \kappa, 1-\frac{\alpha-\kappa}{2} 
\right\}.
\end{align}

For $\alpha \ge 1$,
\begin{align}
d = \min\left\{
\alpha, 1 + \kappa, \frac{\alpha+\kappa}{2} 
\right\}.
\end{align}
\end{theorem}

Numerical plots for g.d.o.f. are given in Fig. \ref{fig_Dof}. We observe that at different values of $\alpha$, the gain from cooperation varies. By investigating the g.d.o.f., we conclude that at high $\SNR$, when $\INR$ is below 50\% of $\SNR$ in dB scale, one-bit cooperation per direction buys roughly one-bit gain per user until full receiver cooperation performance is reached, while when $\INR$ is between 67\% and 200\% of $\SNR$ in dB scale, one-bit cooperation per direction buys roughly half-bit gain per user until saturation.

\begin{figure}[htbp]
{\center
\includegraphics[width=4in]{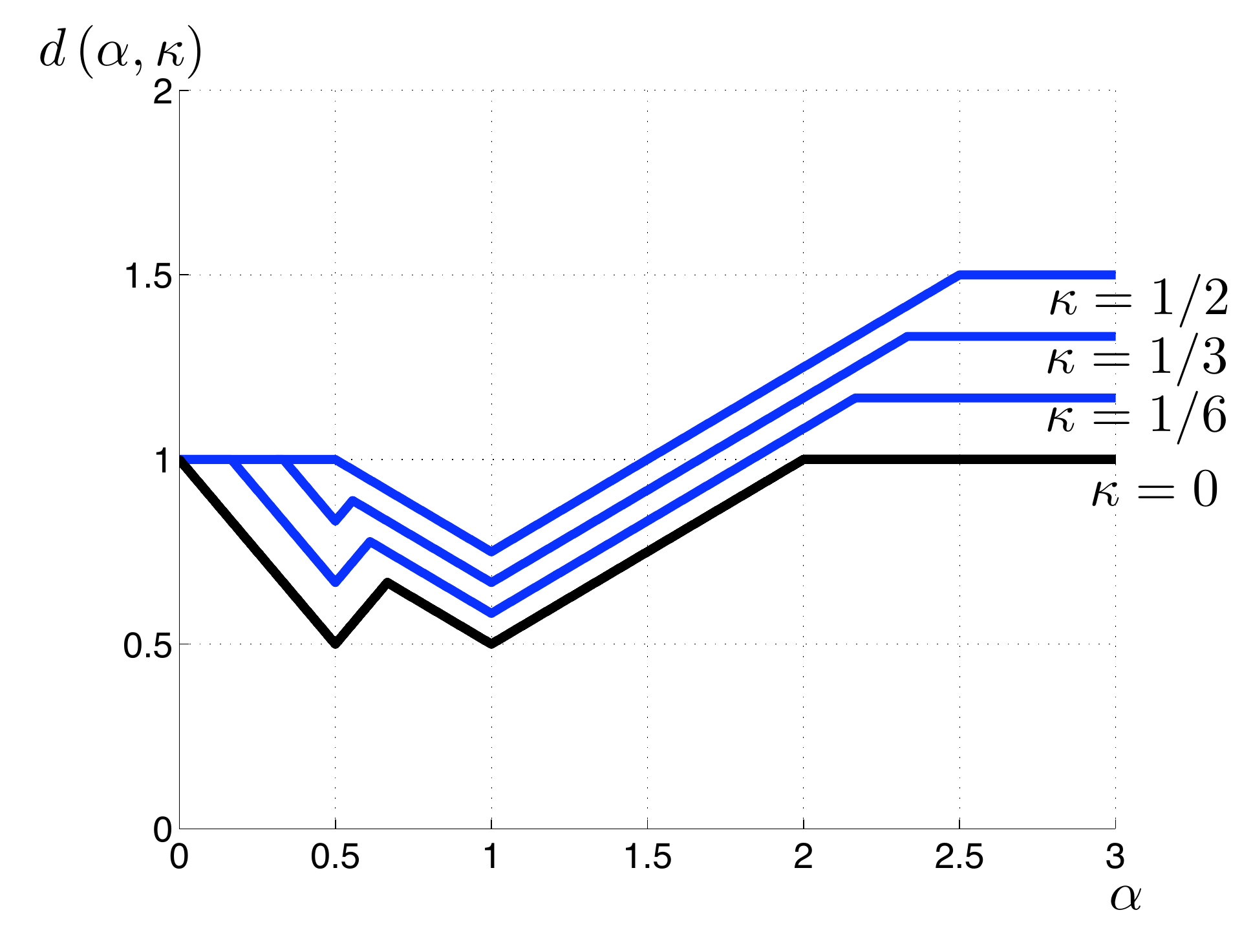}
\caption{Generalized Degrees of Freedom}
\label{fig_Dof}
}
\end{figure}

\subsection{Gain from Limited Receiver Cooperation}
The fundamental behavior of the gain from receiver cooperation is explained in the rest of this section, by looking at two particular points: $\alpha=\frac{1}{2}$ and $\alpha=\frac{2}{3}$. Furthermore, we use the linear deterministic channel (LDC) for illustration.
\begin{figure}[htbp]
{\center
\subfigure[]{\includegraphics[width=3in]{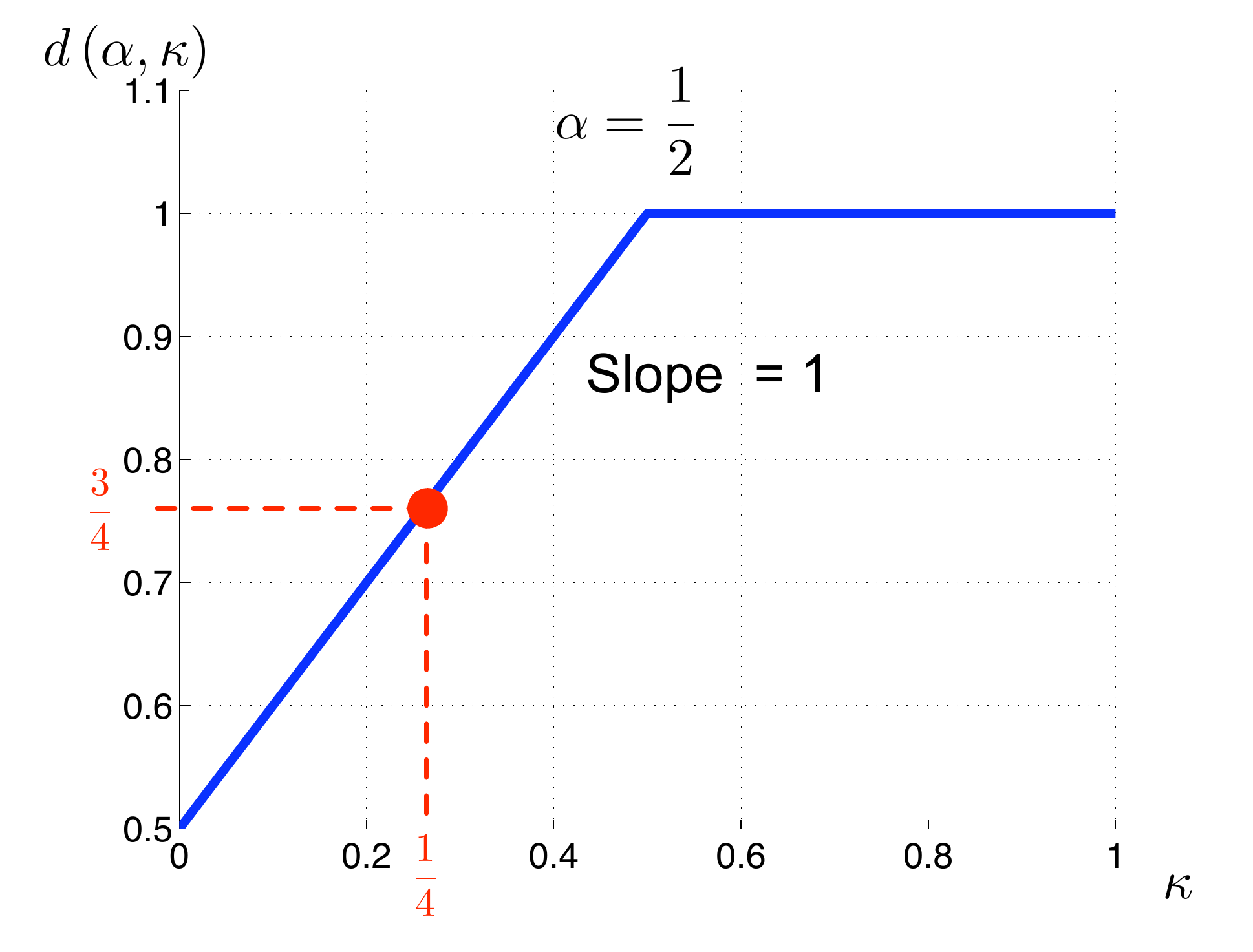}}
\subfigure[]{\includegraphics[width=3in]{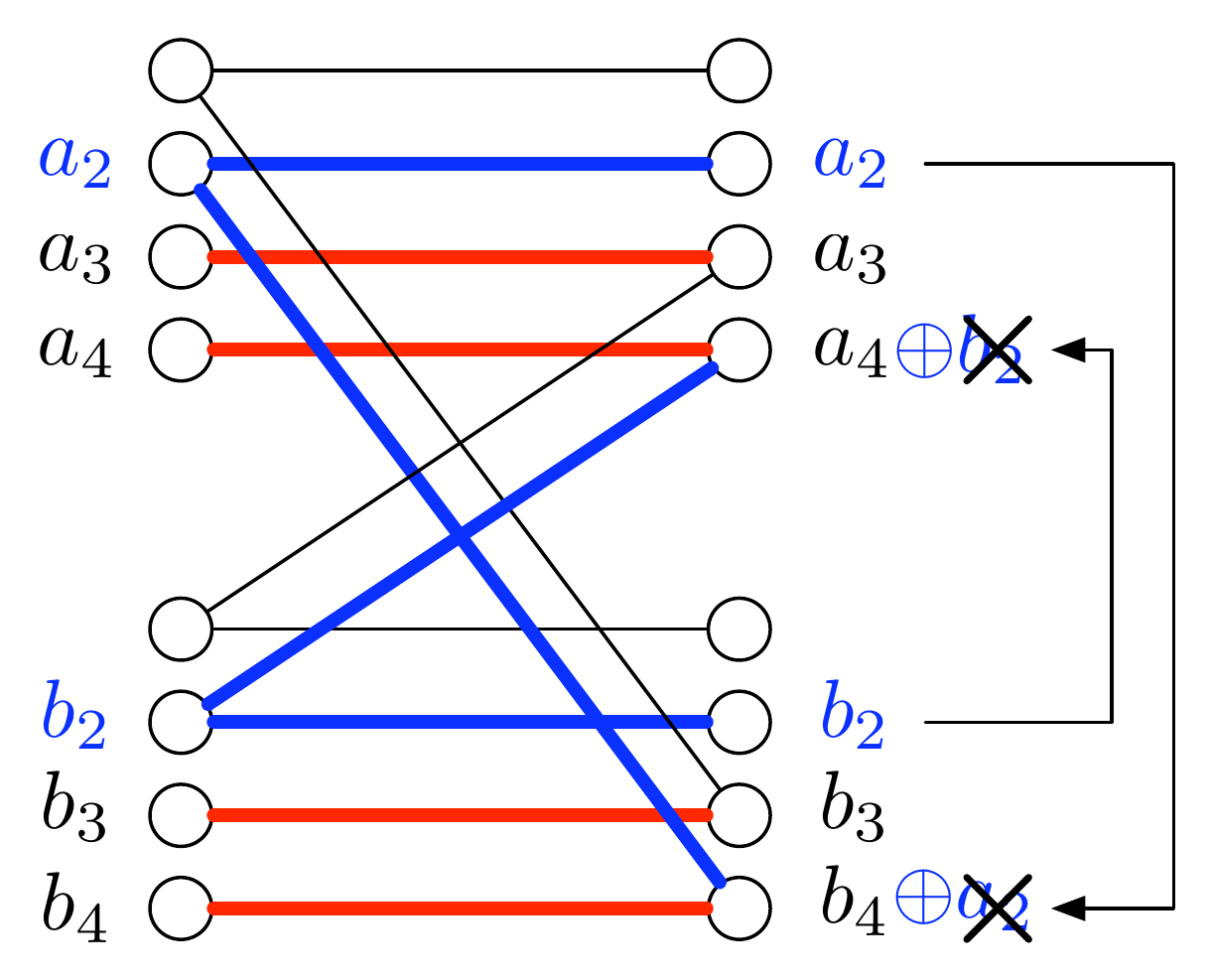}}
\subfigure[]{\includegraphics[width=3in]{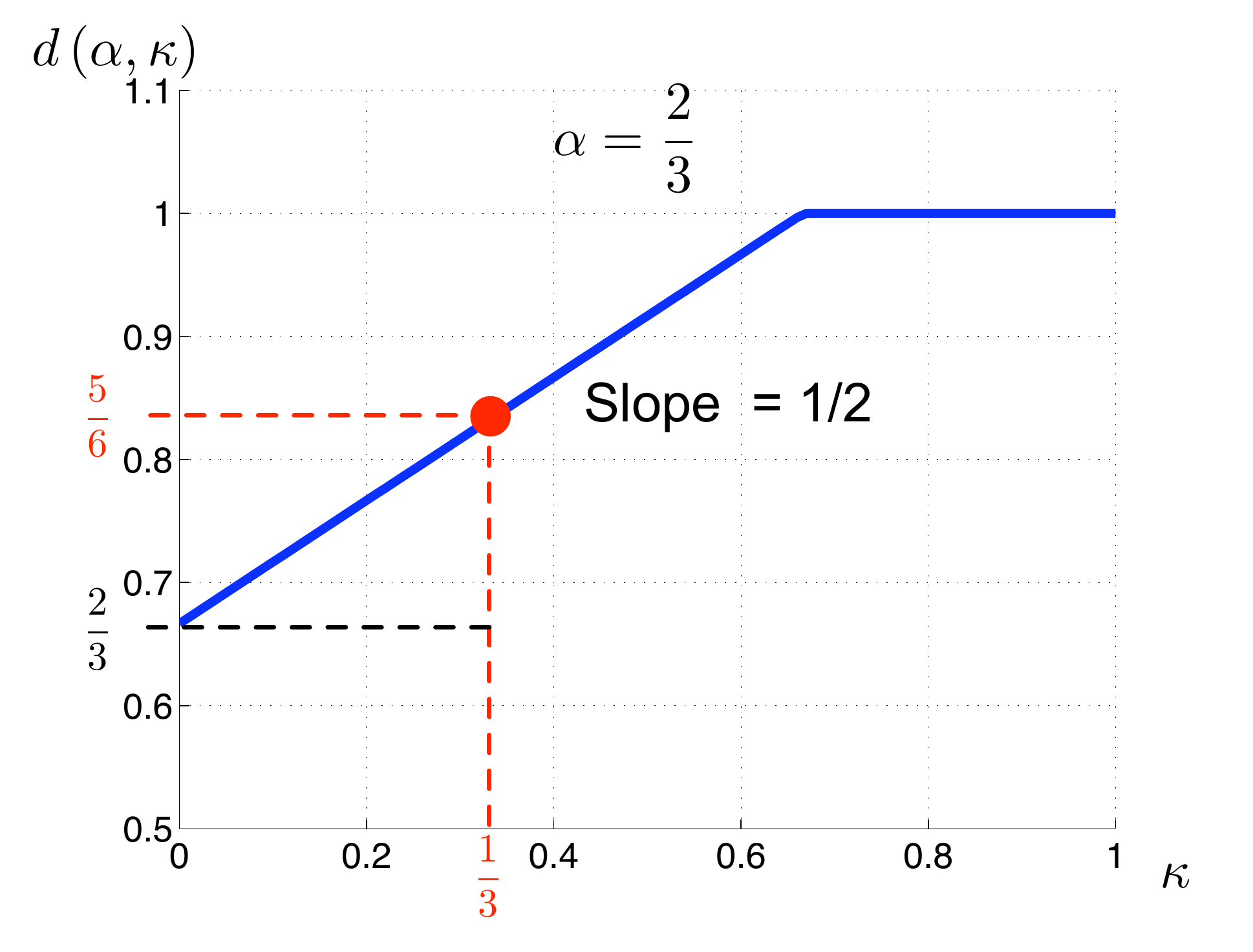}}
\subfigure[]{\includegraphics[width=3in]{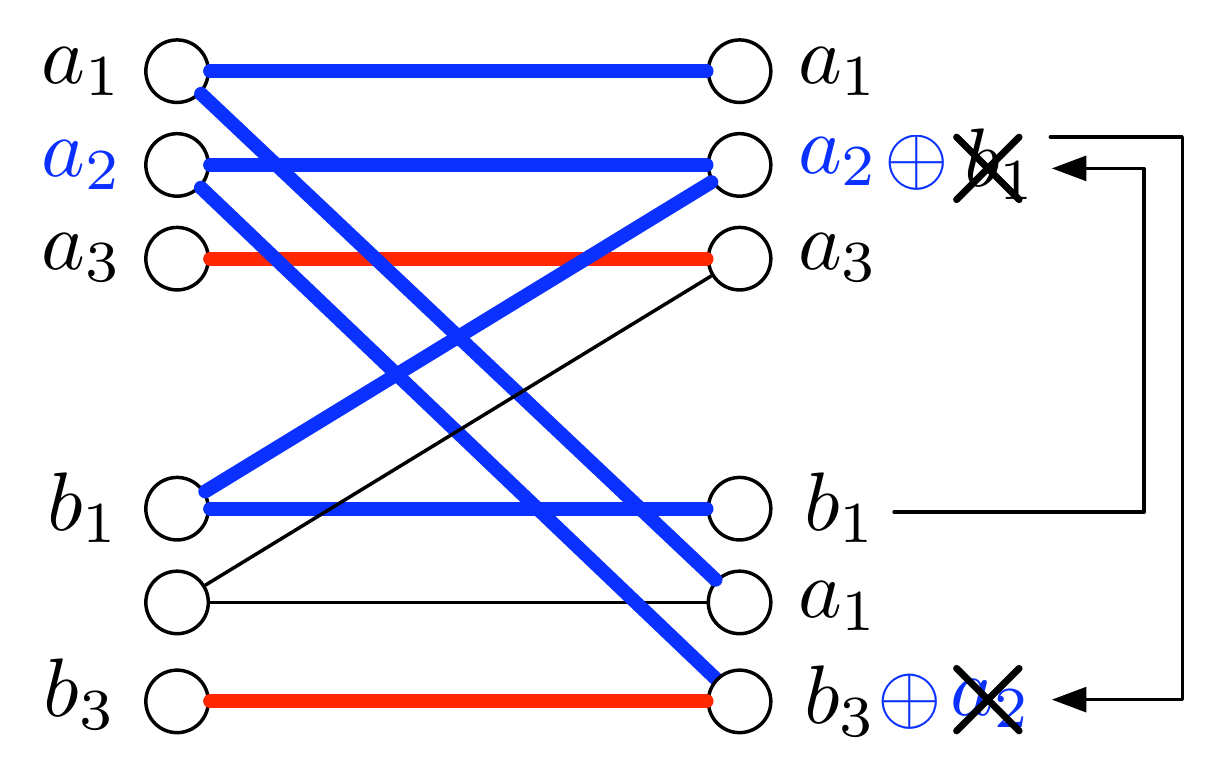}}
\caption{Gain from Cooperation}
\label{fig_Regimes}
}
\end{figure}

At $\alpha = \frac{1}{2}$, the plot of $d$ versus $\kappa$ is given in Fig. \ref{fig_Regimes}.(a). The slope is 1 until full receiver cooperation performance is reached, implying that one-bit cooperation buys one more bit per user. We look at a particular point $\kappa=\frac{1}{4}$ and use its corresponding LDC (Fig. \ref{fig_Regimes}.(b)) to provide insights. Note that 1 bit in the LDC corresponds to $\frac{1}{4}\log\SNR$ in the Gaussian channel, and since $\C \approx \frac{1}{4}\log\SNR$, in the corresponding LDC each receiver is able to sent one-bit information to the other. Without cooperation, the optimal way is to turn on bits not causing interference, that is, the \emph{private} bits $a_3,a_4,b_3,b_4$. We cannot turn on more bits without cooperation since it causes collisions, for example, at the fourth level of receiver 2 if we turn on $a_2$ bit. Now with receiver cooperation, we want to support two more bits $a_2,b_2$. Note that prior to turning on $a_2,b_2$, there are ``holes" left in receiver signal spaces, and turning on each of these bits only causes one collision at one receiver. Therefore, we need 1 bit in each direction to resolve the collision at each receiver. We can achieve 3 bits per user in the corresponding LDC and $d = \frac{3}{4}$ in the Gaussian channel. We cannot turn on more bits in the LDC since it causes collisions while no cooperation capability is left.

At $\alpha = \frac{2}{3}$, the plot of $d$ versus $\kappa$ is given in Fig. \ref{fig_Regimes}.(c). The slope is $\frac{1}{2}$ until full receiver cooperation performance is reached, implying that two-bit cooperation buys one more bit per user. We look at a particular point $\kappa=\frac{1}{3}$ and use its corresponding LDC (Fig. \ref{fig_Regimes}.(d)) to provide insights. Note that now 1 bit in the LDC corresponds to $\frac{1}{3}\log\SNR$ in the Gaussian channel, and since $\C \approx \frac{1}{3}\log\SNR$, in the corresponding LDC each receiver is able to sent one-bit information to the other. Without cooperation, the optimal way is to turn on bits $a_1,a_3,b_1,b_3$. We cannot turn on more bits without cooperation since it causes collisions, for example, at the second level of receiver 2 if we turn on $a_2$ bit. Now with receiver cooperation, we want to support one more bit $a_2$. Note that prior to turning on $a_2$, there are no ``holes" left in receiver signal spaces, and turning on $a_2$ causes collisions at \emph{both} receivers. Therefore, we need 2 bits in total to resolve collisions at both receivers. We can achieve 5 bits in total in the corresponding LDC and $d = \frac{5}{6}$ in the Gaussian channel. We cannot turn on more bits in the LDC since it causes collision while no cooperation capability is left.

From above examples and illustrations, we see that whether \emph{one cooperation bit buys one more bit} or \emph{two cooperation bits buy one more bit} depends on whether there are ``holes" in receiver signal spaces before increasing data rates. The ``holes" play a central role not only in why conventional compress-forward is suboptimal in certain regimes, as mentioned in the previous section, but also in the fundamental behavior of the gain from receiver cooperation. We notice that in \cite{PrabhakaranViswanath_09}, there is a similar behavior about the gain from cooperation as discussed in Section 3.2. of \cite{PrabhakaranViswanath_09}. We conjecture that the behavior can be explained via the concept of ``holes" as well.

\begin{figure}[hbtp]
{\center
\subfigure[$\kappa=0$]{\includegraphics[width=3in]{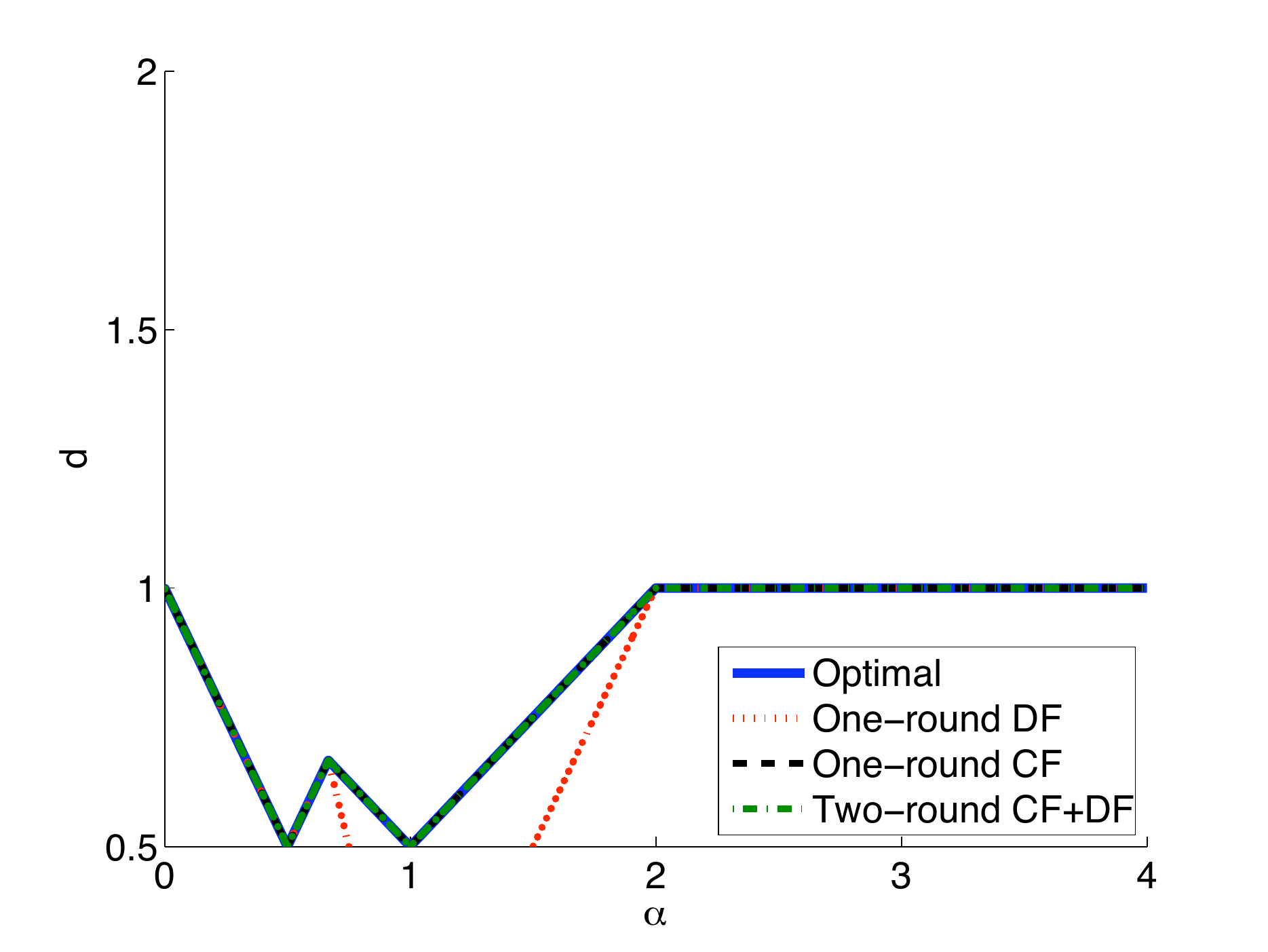}}
\subfigure[$\kappa=0.2$]{\includegraphics[width=3in]{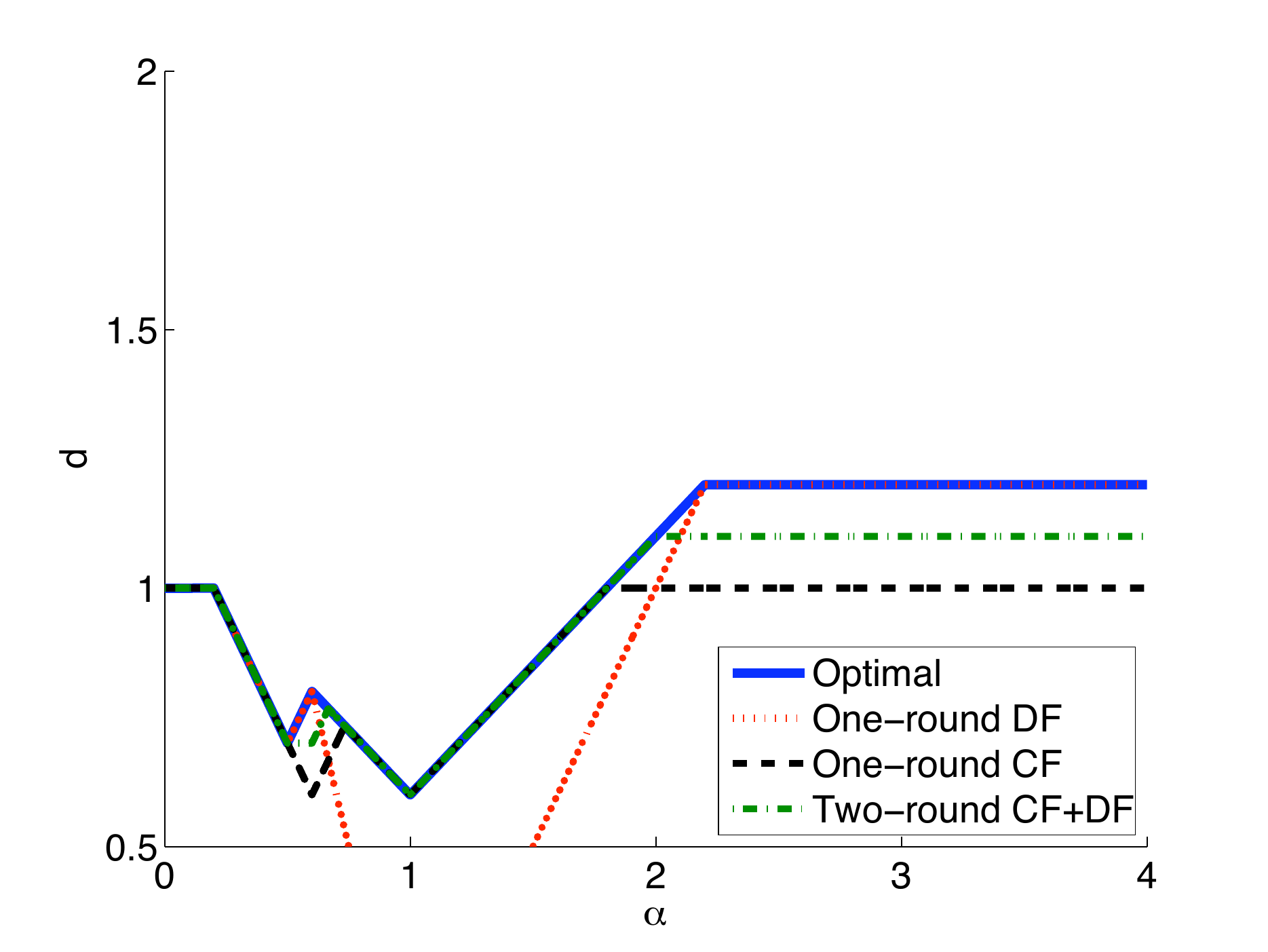}}
\subfigure[$\kappa=0.5$]{\includegraphics[width=3in]{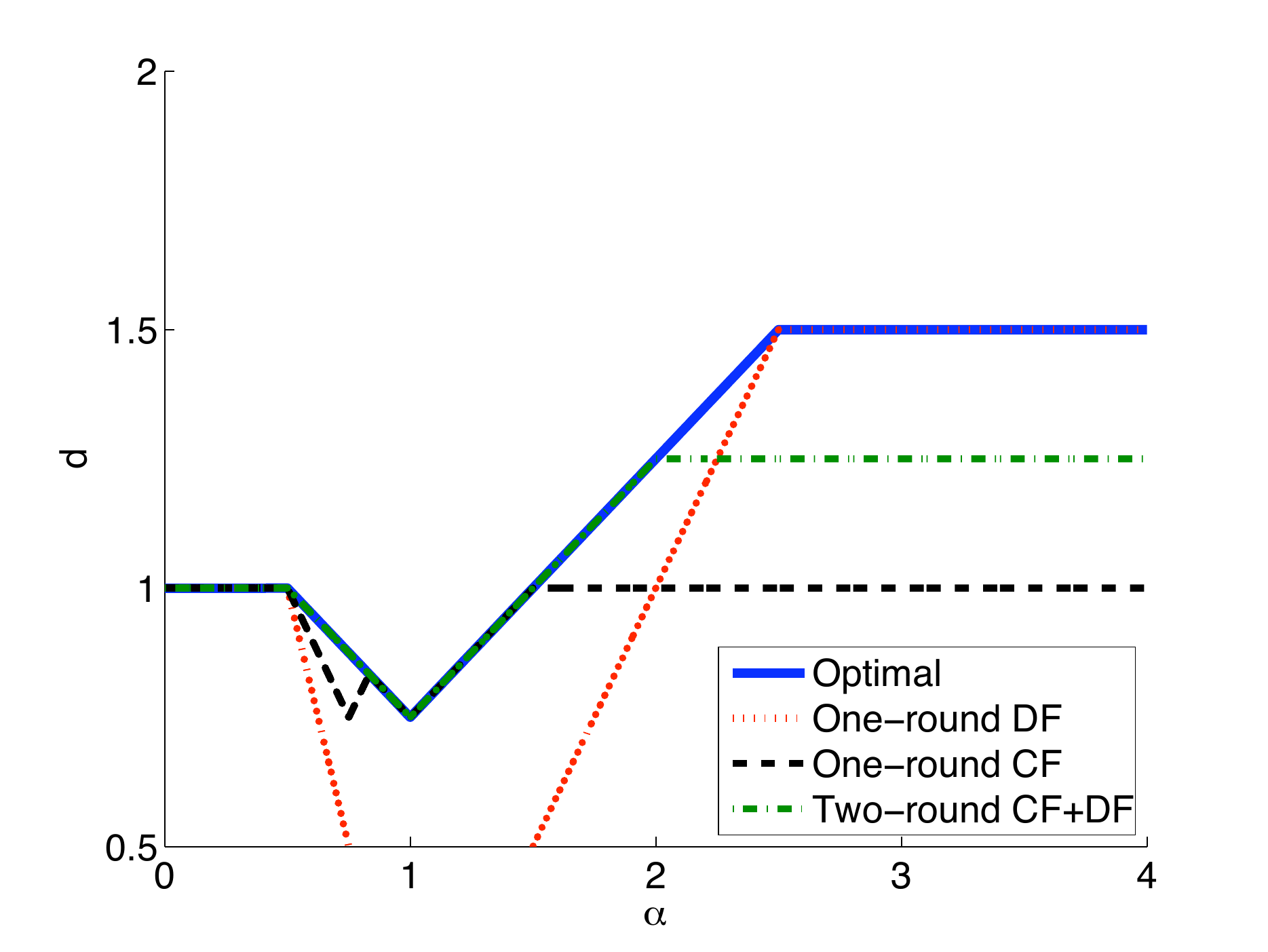}}
\subfigure[$\kappa=0.8$]{\includegraphics[width=3in]{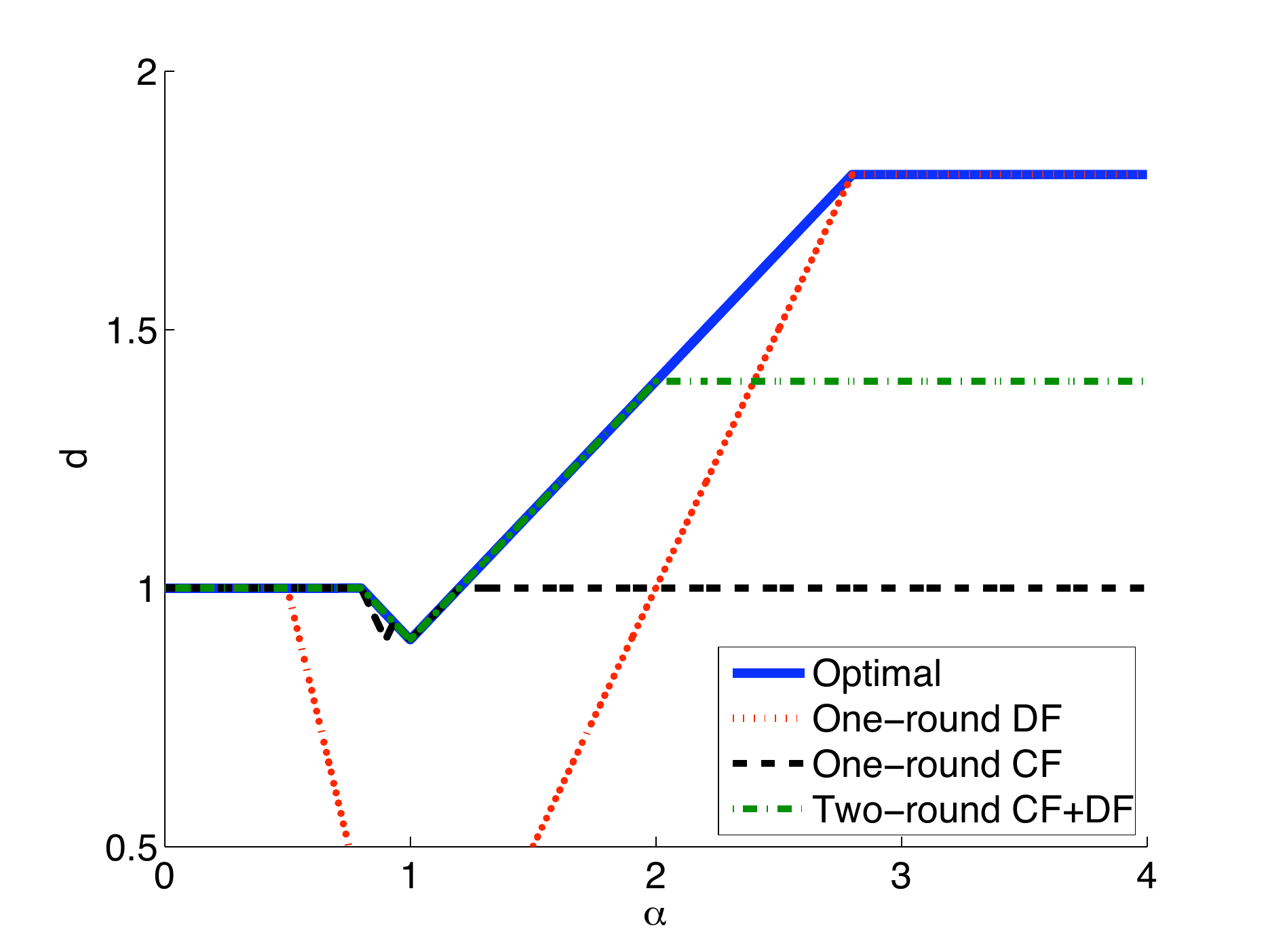}}
\caption{Number of Generalized Degrees of Freedom}
\label{fig_Dof2}
}
\end{figure}

\subsection{Comparison with Suboptimal Strategies}
Pointed out by the motivating example in Section \ref{sec_Example}, conventional compress-forward and decode-forward are not good for receiver cooperation to mitigate interference in certain regimes, which are used in \cite{Simeone_08} and \cite{YuZhou_08}. These suboptimal schemes include:
\begin{itemize}
\item[(1)] One-round compress-forward (CF) strategy: the conventional compress-forward is used for the two receivers to first exchange information and then decode.
\item[(2)] One-round decode-forward (DF) strategy: at the first stage both receivers decode one of the common messages with stronger signal strength without help from the receiver-cooperative links, by treating other signals as noise. Both then bin-and-forward the decoded information to each other. At the second stage, both receivers make use of the bin index send over receiver-cooperative links to decode and enhance the rate.
\item[(3)] Two-round CF+DF strategy: at the first stage one of the receivers, say, receiver 1, compresses its received signal and forwards it to the other receiver. At the second stage, receiver 2 decodes with the side information received at the first round, and then bin-and-forwards the decoded information to receiver 1. Then at the third stage receiver 1 decodes with the help from receiver-cooperative links.
\end{itemize}

Comparisons of these strategies in terms of the number of generalized degrees of freedom for different scaling exponents $\alpha$ of $\log\INR$ and $\kappa$ of $\C$ are depicted in Fig. \ref{fig_Dof2}. None of them achieves the optimal g.d.o.f. universally.
Note that although the two-round CF+DF strategy outperforms one-round CF/DF strategies, it cannot achieve the optimal number of g.d.o.f. for all $\alpha$'s and $\kappa$'s. One-round strategy based on our cooperative protocol, on the other hand, is sufficient to achieve the symmetric capacity to within 3 bits universally and hence achieves the optimal number of g.d.o.f. for all $\alpha$'s and $\kappa$'s.

\bibliographystyle{ieeetr}

\appendices

\section{Proof of Theorem \ref{thm_CodingThm}}\label{app_ErrorAnalysis}
We will first describe the strategy in detail and analyze the error probability rigorously.
\subsection{Description of the Strategy}\label{subsec_AchDes}
In the following, consider all $i,j\in\{1,2\}$ and $i\ne j$. 
{\flushleft \it Codebook generation}:\par
Transmitter $i$ splits its message $m_i \rightarrow (m_{ic} , m_{ip})$. Consider block length-$N$ encoding. First we generate $2^{NR_{ic}}$ common codewords $\{x^N_{ic}(m_{ic})$, $1\le m_{ic}\le 2^{NR_{ic}}\}$, according to distribution $p\lp x_{ic}^N\rp = \prod_{n=1}^N p\big( x_{ic}[n] \big)$ with $x_{ic}[n]\sim\mcal{CN}(0,Q_{ic})$ for all $n$. Then for each common codeword $x^N_{ic}(m_{ic})$ serving as a {\it cloud center}, we generate $2^{NR_{ip}}$ codewords $\{ x^N_i(m_{ic},m_{ip})$, $1\le m_{ip}\le 2^{NR_{ip}}\}$, according to conditional distribution $p\lp x_i^N|x_{ic}^N\rp = \prod_{n=1}^N p\big( x_i[n]| x_{ic}[n] \big)$ such that for all $n$, $x_i[n] = x_{ic}[n] + x_{ip}[n]$, where $x_{ip}[n] \sim \mcal{CN}(0, Q_{ip})$ and independent of everything else. The power split configuration is such that $Q_{ip}+Q_{ic}=1$, $\INR_{jp}:= Q_{ip}|h_{ji}|^2 \le 1$ if $\SNR_i > \INR_j$, and no such split if $\SNR_i \le \INR_j$. Hence, $Q_{ip} = \min\lbp 1,\frac{1}{\INR_j}\rbp$ if $\SNR_i > \INR_j$, and $Q_{ip}=0$ otherwise.

For receiver $2$ serving as relay, it generates a quantization codebook $\widehat{\mathscr{Y}}_2$, of size $\big|\widehat{\mathscr{Y}}_2\big|=2^{N\widehat{R}_2}$, randomly according to marginal distribution $p(\what{y}_2^N) = \int p(y_2^N)p(\widehat{y}_2^N | y_2^N) dy_2^N$, where $p(\widehat{y}_2^N | y_2^N) = \prod_{n=1}^N p\Big(\widehat{y}_2[n] \big| y_2[n] \Big)$. The conditional distribution is such that for all $n$, $\what{y}_2[n] = y_2[n] + \what{z}_2[n]$, where $\what{z}_2[n] \sim \mcal{CN}(0,\Delta_2)$, independent of everything else. Parameters $\widehat{R}_2$ and $\Delta_2$ are to be specified later. For each element in codebook $\widehat{\mathscr{Y}}_2$, map it into $\{1,\ldots,2^{N\C_{21}}\}$ through a uniformly generated random mapping $b_2: \what{\mscr{Y}}_2\rightarrow \{1,\ldots,2^{N\C_{21}}\}, \what{y}_2^N \mapsto l_{21}$ ({\it binning}).

For receiver $1$ serving as relay, it generates two binning functions $b_{1}^{(1c)}$ and $b_{1}^{(2c)}$ independently according to uniform distributions, such that the message set $\{1\le m_{ic} \le 2^{NR_{ic}}\}$ is partitioned into $2^{\lambda_1^{(ic)}N\C_{12}}$ bins, for $i=1,2$, where $0\le \lambda_1^{(ic)} \le 1$, $\lambda_1^{(1c)} + \lambda_1^{(2c)} = 1$, and  
\begin{align}
b_1^{(ic)}: \{1,\ldots,2^{NR_{ic}}\} \rightarrow \{1,\ldots,2^{\lambda_1^{(ic)}N\C_{12}}\},\ m_{ic} \mapsto l_{12}^{(ic)} \in \{1,\ldots,2^{\lambda_1^{(ic)}N\C_{12}}\}.
\end{align} 
The superscript notation ``$(ic)$" denotes which message set is partitioned into bins, while the subscript ``$1$" denotes the binning procedure is at receiver $1$.

{\flushleft \it Encoding}:\par
Transmitter $i$ sends out signals according to its message and the codebook. Receiver $2$, serving as relay, chooses the quantization codeword which is jointly typical with $y_2^N$ (if there is more than one, it chooses the one with the smallest index), and then sends out the bin index $l_{21}$ for the quantization codeword. After decoding $\lp m_{1c}, m_{1p}, m_{2c}\rp$ (to be specified below), receiver 1 sends out bin indices $\lp l_{12}^{(1c)},l_{12}^{(2c)}\rp$ according to binning functions $\lp b_1^{(1c)},b_1^{(2c)}\rp$.

{\flushleft \it Decoding at receiver 1}:\par
To draw comparison with the decoding procedure in the conventional compress-forward, the above decoding can be interpreted as a two-stage procedure as follows. It first constructs a {\it list} of message triples (both users' common messages and its own private message), each element of which indices a codeword triple that is jointly typical with its received signal from the transmitter-receiver link. Then, for each message triple in this list, it constructs an {\it ambiguity set} of quantization codewords, each element of which is jointly typical with the codeword triple and the received signal. Finally, it searches through all ambiguity sets and finds one that contains a quantization codeword with the same bin index it received. If there is no such unique ambiguity set, it declares an error. The two-stage interpretation is illustrated in Fig. \ref{fig_Decode}.

To be specific, upon receiving signal $y_{1}$ and receiver-cooperative side information $l_{21}$, receiver $i$ constructs a list of candidates
\begin{align}
L(y^N_1) := \left\{ \ul{m}:=(m_{1c},m_{1p},m_{2c}) \big\lvert \lp x_{1c}^N(m_{1c}), x_{1}^N(m_{1c},m_{1p}), x_{2c}^N(m_{2c}), y_1^N\rp \in A_{\epsilon}^{(N)}\right\},
\end{align}
where $A_{\epsilon}^{(N)}$ denotes the set of jointly $\epsilon$-typical $N$-sequences, correspondingly \cite{CoverThomas_06}.

For each element $\ul{m} \in L(y^N_1)$, construct an ambiguity set of quantization codewords
\begin{align}
B(\ul{m}) := \left\{ \what{y}^N_2\in\what{\mscr{Y}}_2\ \Big\lvert \lp \what{y}^N_2, x_{1c}^N(m_{1c}), x_{1}^N(m_{1c},m_{1p}), x_{2c}^N(m_{2c}), y_i^N \rp\in A_{\epsilon}^{(N)} \right\}.
\end{align}
Declare the transmitted message is $\what{\ul{m}}$ if there exists an unique $\what{\ul{m}}$ such that $\exists\ \what{y}^N_2\in B(\what{\ul{m}})$ with $b_2(\what{y}^N_2) = l_{21}$. Otherwise, declare an error.

\begin{figure}[htbp]
{\center
\subfigure[Error Event (1)]{\includegraphics[width=3in]{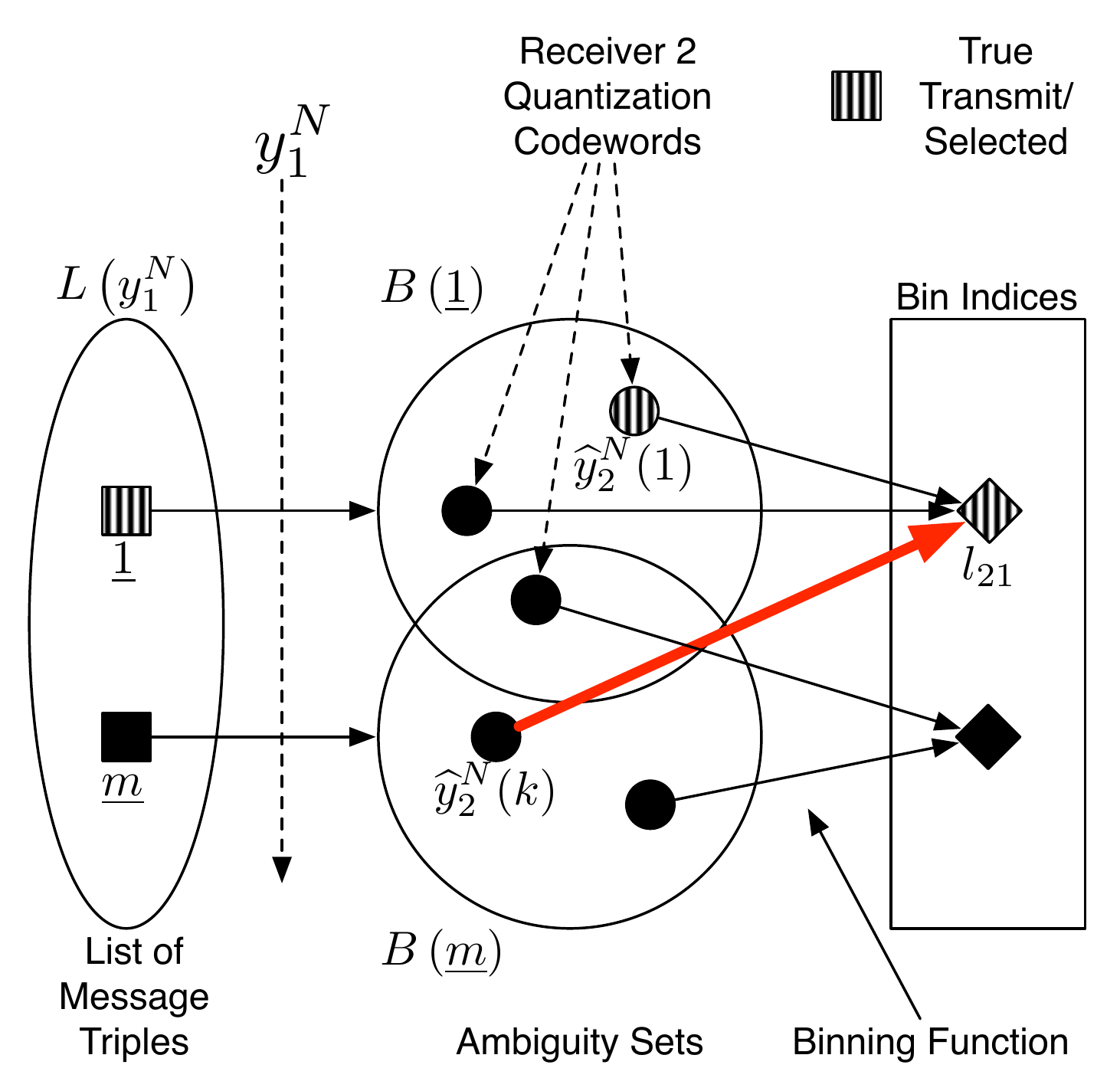}}
\subfigure[Error Event (2)]{\includegraphics[width=3in]{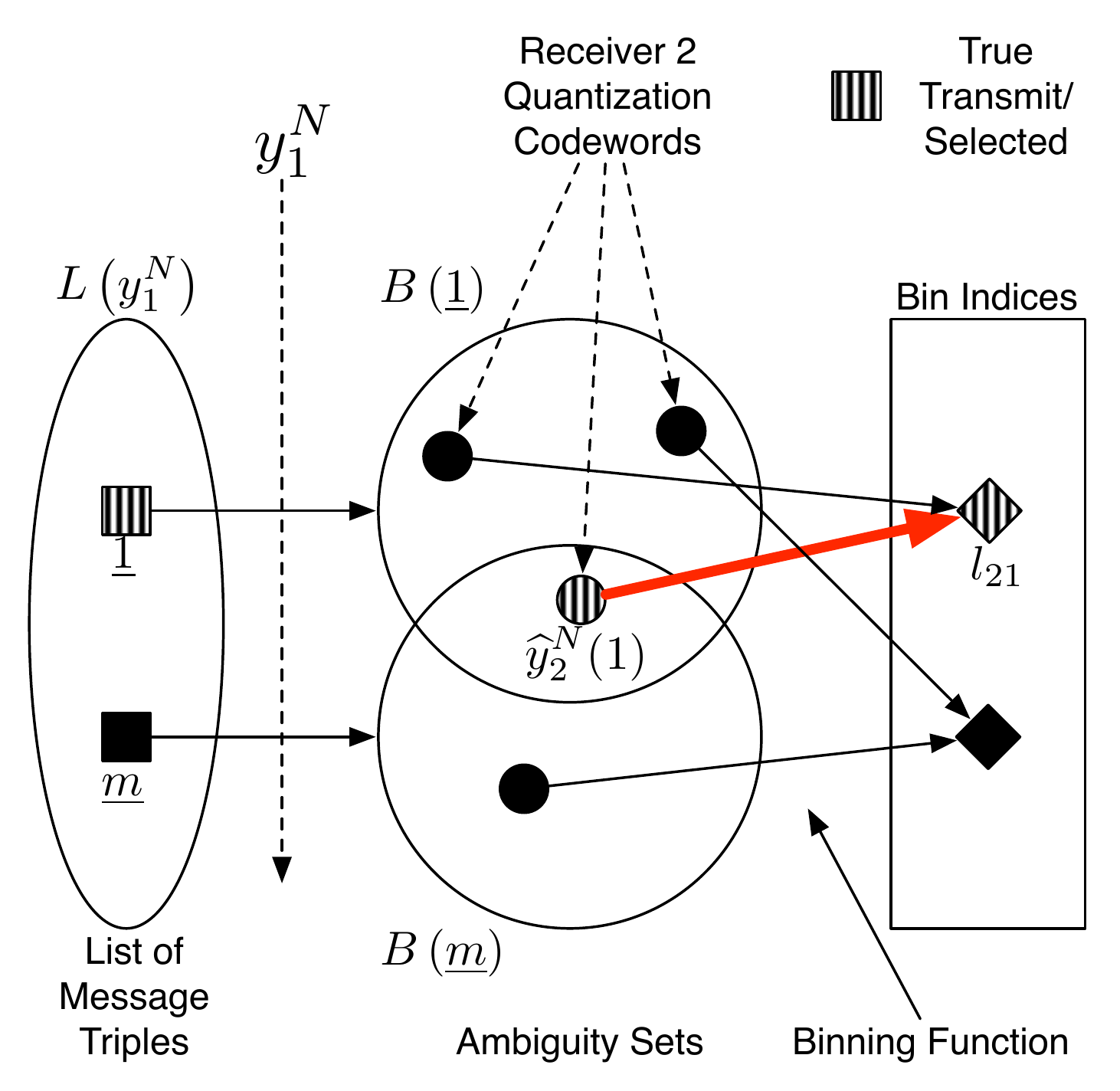}}
\caption{Decoding at Receiver 1 and Error Events}
\label{fig_Decode}
}
\end{figure}

{\flushleft \it Decoding at receiver 2}:\par
After receiving bin indices $\lp l_{12}^{(1c)},l_{12}^{(2c)}\rp$, receiver 2 searches for an unique message triple $( m_{2c},m_{2p},m_{1c} )$ such that $\lp x_{2c}^N(m_{2c}), x_{2}^N(m_{2c},m_{2p}), x_{1c}^N(m_{1c}), y_2^N\rp \in A_{\epsilon}^{(N)}$, and $b_1^{(ic)}\lp m_{ic}\rp = l_{12}^{(ic)}$, for $i=1,2$. If there is no such unique triple, it declares an error.

\subsection{Analysis}\label{subsec_Analysis}
{\flushleft \it Error probability analysis at receiver 1}:\par
Without loss of generality, assume that all transmitted messages are 1's. For simplicity, we first focus on the case where receiver 1 aims to decode while receiver 2 serves as a relay to help it decode.

At receiver 1, due to law of large numbers, the probability that the truly transmitted $\ul{1}:=(m_{1c}=1,m_{2c}=1,m_{1p}=1) \notin L(y^N_1)$ goes to zero as $N\rightarrow \infty$. Besides, the probability that $B(\ul{1})$ does not contain the truly selected $\what{y}^N_2$ is also negligible when $N$ is sufficiently large. Consider the following error events: 

First, there is no quantization codeword jointly typical with received signals. This probability goes to zero as $N\rightarrow \infty$ if $\what{R}_2 \ge I(\what{y}_2;y_2)$, which is a known result in the source coding literature. 

Second, there exists $\ul{m} \ne \ul{1}$ such that both of them are in the candidate list $L(y^N_1)$, and the ambiguity set $B(\ul{m})$ contains some quantization codeword $\what{y}^N_2$ with bin index $b_2(\what{y}^N_2) = l_{21}$. This event can further be distinguished into two cases: First, this $\what{y}^N_2\in B(\ul{m})$ is not the actual selected quantization codeword (illustrated in Fig. \ref{fig_Decode}.(a)); second,  this $\what{y}^N_2\in B(\ul{m})$ is indeed the selected quantization codeword (illustrated in Fig. \ref{fig_Decode}.(b)). In the following we analyze the error probability of these two typical error events.


Again, refer to Fig. \ref{fig_Decode}. for illustration. Define error events as follows: for any nonempty $S\subseteq\{1c,1p,2c\}$,
{\flushleft
$E^{(1)}_{S}:=$ the event that there exists some $\ul{m}\ne \ul{1}$, (where $m_{s}\ne1, \forall s\in S$ and $m_{s}= 1, \forall s\notin S$), such that $\ul{m} \in L(y_1^N)$ and $B(\ul{m})$ contains some $\what{y}^N_2(k)$, $k\in\{1,2,\ldots,2^{N\what{R}_2}\}$ with $b_2(\what{y}^N_2(k)) = l_{21}$. Note: this $\what{y}_2^N(k)$ is not the truly selected quantization codeword $\what{y}_2^N(1)$.
}
{\flushleft
$E^{(2)}_{S}:=$ the event that there exists some $\ul{m}\ne \ul{1}$, (where $m_{s}\ne1, \forall s\in S$ and $m_{s}= 1, \forall s\notin S$), such that $\ul{m} \in L(y_1^N)$ and $B(\ul{m})$ contains $\what{y}^N_2(1)$.
}

\subsubsection{Probability of $E^{(1)}_{S}$}
Consider the probability of the error event $E^{(1)}_{S}$:
\begin{align}
&\Pr\lbp E^{(1)}_{S} \rbp\\
&\le \sum_{\begin{subarray}{c}\ul{m}: m_s\ne1,\\ \forall s\in S \end{subarray}} \sum_{k\ne 1}
\Pr\lbp \ul{m} \in L(y_1^N), \what{y}_2^N(k)\in B(\ul{m}), b_2(\what{y}^N_2(k)) = l_{21} \rbp\\
&= \sum_{\begin{subarray}{c}\ul{m}: m_s\ne1,\\ \forall s\in S \end{subarray}} \sum_{k\ne 1}
\Pr\lbp \lp\what{y}_2^N(k), \ul{x}^N(\ul{m}),y_1^N\rp \in A_{\epsilon}^{(N)}, b_2(\what{y}^N_2(k)) = l_{21} \rbp\\
&\overset{\aaaa}{=} 2^{-N\C_{21}}\sum_{\begin{subarray}{c}\ul{m}: m_s\ne1,\\ \forall s\in S \end{subarray}} \sum_{k\ne 1}
\Pr\lbp \lp\what{y}_2^N(k), \ul{x}^N(\ul{m}),y_1^N\rp \in A_{\epsilon}^{(N)}\rbp\\
&\le 2^{N\lp\sum_{s\in S}R_s\rp}2^{-N\C_{21}} \sum_{k\ne 1}
\Pr\lbp \lp\what{y}_2^N(k), \ul{x}^N(\ul{m}),y_1^N\rp \in A_{\epsilon}^{(N)}\rbp,
\end{align}
where (a) is due to the independent uniform binning.

For notational convenience we use $\ul{x}^N(\ul{m})$ to denote the vector of codewords corresponding to message $\ul{m}$, that is, $\big( x_{1c}^N(m_{1c}), x_{1}^N(m_{1c},m_{1p}), x_{2c}^N(m_{2c})\big)$.

Note that for $k\ne 1$, $\what{y}_2^N(k)$ is independent of $\lp\ul{x}^N(\ul{m}),y_1^N\rp$. We then make use of Theorem 15.2.2 in \cite{CoverThomas_06}, which upper bounds the volume of conditional joint $\epsilon$-typical set $A_{\epsilon}^{(N)}\big(\what{y}_2 \big| \ul{x}^N,y_1^N\big)$ given that $\lp\ul{x}^N,y_1^N\rp \in A_{\epsilon}^{(N)}$:
\begin{align}
&\sum_{k\ne 1}\Pr\lbp \lp\what{y}_2^N(k), \ul{x}^N(\ul{m}),y_1^N\rp \in A_{\epsilon}^{(N)} \rbp\\
&\le 2^{N\what{R}_2} \int_{\lp\what{y}_2^N, \ul{x}^N,y_1^N\rp \in A_{\epsilon}^{(N)}} p\big(\what{y}_2^N\big)p\big(\ul{x}^N,y_1^N\big) d \what{y}_2^N d\ul{x}^N dy_1^N\\
&\overset{}{\le} 2^{N\what{R}_2} \int_{\lp\ul{x}^N,y_1^N\rp \in A_{\epsilon}^{(N)}} p\big(\ul{x}^N,y_1^N\big) d\ul{x}^N dy_1^N \int_{\what{y}_2^N\in A_{\epsilon}^{(N)}\big(\what{y}_2 \big| \ul{x}^N,y_1^N\big)} p\big(\what{y}_2^N\big) dy_2^N\\
&\le 2^{N\what{R}_2} \int_{\lp\ul{x}^N,y_1^N\rp \in A_{\epsilon}^{(N)}} p\big(\ul{x}^N,y_1^N\big) d\ul{x}^N dy_1^N\int_{\what{y}_2^N\in A_{\epsilon}^{(N)}\big(\what{y}_2 \big| \ul{x}^N,y_1^N\big)}2^{-N\lp h(\what{y}_2) -\epsilon \rp} dy_2^N \\
&\overset{\bbbb}{\le} 2^{N\lp h(\what{y}_2|x_{1c},x_1,x_{2c},y_1)+2\epsilon\rp}\cdot2^{-N\lp h(\what{y}_2) - \epsilon\rp}\cdot 2^{N\what{R}_2} \int_{\lp\ul{x}^N,y_1^N\rp \in A_{\epsilon}^{(N)}} p\big(\ul{x}^N,y_1^N\big) d\ul{x}^N dy_1^N\\
&= 2^{N\what{R}_2}2^{-N\lp I(\what{y}_2;x_{1c},x_1,x_{2c},y_1)-3\epsilon\rp} \int_{\lp\ul{x}^N,y_1^N\rp \in A_{\epsilon}^{(N)}} p\big(\ul{x}^N,y_1^N\big) d\ul{x}^N dy_1^N\\
&= \Pr\lbp \ul{m} \in L(y_1^N)\rbp \cdot 2^{N\what{R}_2}2^{-N\lp I(\what{y}_2;x_{1c},x_1,x_{2c},y_1)-3\epsilon\rp},
\end{align}
where (b) is due to Theorem 15.2.2 in \cite{CoverThomas_06}. Besides, according to the results in \cite{ChongMotani_08},
\begin{align}
&\Pr\lbp \ul{m} \in L(y_1^N)\rbp \le \left\{
\begin{array}{ll}
2^{-N\lp I\lp x_{1};y_1|x_{1c},x_{2c}\rp - \epsilon' \rp} & S = \{1p\}\\
2^{-N\lp I\lp x_{1};y_1|x_{2c}\rp - \epsilon' \rp} & S = \{1c\}\\
2^{-N\lp I\lp x_{2c};y_1|x_{1}\rp - \epsilon' \rp} & S = \{2c\}\\
2^{-N\lp I\lp x_{2c},x_{1};y_1| x_{1c}\rp - \epsilon' \rp} & S = \{1p,2c\}\\
2^{-N\lp I\lp x_{1};y_1 | x_{2c}\rp - \epsilon' \rp} & S = \{1p,1c\}\\
2^{-N\lp I\lp x_{1},x_{2c};y_1 \rp - \epsilon' \rp} & S = \{2c,1c\}\\
2^{-N\lp I\lp x_{1},x_{2c};y_1 \rp - \epsilon' \rp} & S = \{1p,2c,1c\}
\end{array}
\right. , 
\end{align}
where $\epsilon' = 4\epsilon$. 
Note that unlike in the interference channel without cooperation as in \cite{ChongMotani_08}, here we require receiver 1 to decode $m_{2c}$ correctly. Hence, the event when $S=\{2c\}$ does cause an error. Therefore, the probability of the first kind of error event vanishes as $N\rightarrow\infty$ if 
\begin{align}
R_{1p} &\le I\lp x_{1};y_1|x_{1c},x_{2c}\rp + \C_{21} - \what{R}_2 + I(\what{y}_2;x_{1c},x_1,x_{2c},y_1)\\
R_{2c} &\le I\lp x_{2c};y_1|x_{1}\rp + \C_{21} - \what{R}_2 + I(\what{y}_2;x_{1c},x_1,x_{2c},y_1)\\
R_{2c}+R_{1p} &\le I\lp x_{2c},x_{1};y_1|x_{1c}\rp + \C_{21} - \what{R}_2 + I(\what{y}_2;x_{1c},x_1,x_{2c},y_1) \\
R_{1c}+R_{1p} &\le I\lp x_{1};y_1|x_{2c}\rp + \C_{21} - \what{R}_2 + I(\what{y}_2;x_{1c},x_1,x_{2c},y_1) \\
R_{1c}+R_{2c}+R_{1p} &\le I\lp x_{1},x_{2c};y_1 \rp + \C_{21} - \what{R}_2 + I(\what{y}_2;x_{1c},x_1,x_{2c},y_1).
\end{align}

On the other hand, since we can rewrite
\begin{align}
&\Pr\lbp E^{(1)}_{S} \rbp \le \\
&\sum_{\begin{subarray}\ul{m}: m_s\ne1,\\ \forall s\in S\end{subarray}} \Pr\lbp \ul{m} \in L(y_1^N)\rbp \cdot
\Pr\lbp \exists\ k\ne 1,\what{y}_2^N(k)\in B(\ul{m}), b_2(\what{y}^N_2(k)) = l_{21} \Big\lvert \ul{m} \in L(y_1^N)\rbp\\
&\le 2^{N\lp\sum_{s\in S}R_s\rp} \Pr\lbp \ul{m} \in L(y_1^N)\rbp.
\end{align}

Hence, the probability of the first kind of error event vanishes as $N\rightarrow\infty$ if 
\begin{align}
R_{1p} &\le I\lp x_{1};y_1|x_{1c},x_{2c}\rp + \lp\C_{21} - \what{R}_2 + I(\what{y}_2;x_{1c},x_1,x_{2c},y_1)\rp^+\\
R_{2c} &\le I\lp x_{2c};y_1|x_{1}\rp + \lp\C_{21} - \what{R}_2 + I(\what{y}_2;x_{1c},x_1,x_{2c},y_1)\rp^+\\
R_{2c}+R_{1p} &\le I\lp x_{2c},x_{1};y_1|x_{1c}\rp + \lp\C_{21} - \what{R}_2 + I(\what{y}_2;x_{1c},x_1,x_{2c},y_1)\rp^+ \\
R_{1c}+R_{1p} &\le I\lp x_{1};y_1|x_{2c}\rp + \lp\C_{21} - \what{R}_2 + I(\what{y}_2;x_{1c},x_1,x_{2c},y_1)\rp^+ \\
R_{1c}+R_{2c}+R_{1p} &\le I\lp x_{1},x_{2c};y_1 \rp + \lp\C_{21} - \what{R}_2 + I(\what{y}_2;x_{1c},x_1,x_{2c},y_1)\rp^+.
\end{align}

Finally, plug in $\what{R}_2 = I(\what{y}_2;y_2)$ and by Markov relation: $(x_{1c},x_1,x_{2c},y_1) - y_2 - \what{y}_2$, we get the rate loss term
\begin{align}
\xi_1 &:= \what{R}_2 - I(\what{y}_2;x_{1c},x_1,x_{2c},y_1) = I(\what{y}_2;y_2) - I(\what{y}_2;x_{1c},x_1,x_{2c},y_1)\\
&= I(\what{y}_2;y_2|x_{1c},x_1,x_{2c},y_1).
\end{align}

\subsubsection{Probability of $E^{(2)}_{S}$}
Consider the probability of the error event $E^{(2)}_{S}$:
\begin{align}
\Pr\lbp E^{(2)}_{S} \rbp
&\le \sum_{\ul{m}: m_s\ne1, \forall s\in S} \Pr\lbp \what{y}_2^N(1)\in B(\ul{m}) , \ul{m} \in L(y_1^N)\rbp \\
&=  \sum_{\ul{m}: m_s\ne1, \forall s\in S} \Pr\lbp \lp\what{y}_2^N(1), \ul{x}^N(\ul{m}),y_1^N\rp \in A_{\epsilon}^{(N)} \rbp \\
&\le \left\{
\begin{array}{ll}
2^{N\lp\sum_{s\in S}R_s\rp}\cdot2^{-N\lp I\lp x_{1};y_1,\what{y}_2|x_{1c},x_{2c}\rp - \epsilon' \rp} & S = \{1p\}\\
2^{N\lp\sum_{s\in S}R_s\rp}\cdot2^{-N\lp I\lp x_{1};y_1,\what{y}_2|x_{2c}\rp - \epsilon' \rp} & S = \{1c\}\\
2^{N\lp\sum_{s\in S}R_s\rp}\cdot2^{-N\lp I\lp x_{2c};y_1,\what{y}_2|x_{1}\rp - \epsilon' \rp} & S = \{2c\}\\
2^{N\lp\sum_{s\in S}R_s\rp}\cdot2^{-N\lp I\lp x_{2c},x_{1};y_1,\what{y}_2| x_{1c}\rp - \epsilon' \rp} & S = \{1p,2c\}\\
2^{N\lp\sum_{s\in S}R_s\rp}\cdot2^{-N\lp I\lp x_{1};y_1,\what{y}_2 | x_{2c}\rp - \epsilon' \rp} & S = \{1p,1c\}\\
2^{N\lp\sum_{s\in S}R_s\rp}\cdot2^{-N\lp I\lp x_{1},x_{2c};y_1,\what{y}_2 \rp - \epsilon' \rp} & S = \{2c,1c\}\\
2^{N\lp\sum_{s\in S}R_s\rp}\cdot2^{-N\lp I\lp x_{1},x_{2c};y_1,\what{y}_2 \rp - \epsilon' \rp} & S = \{1p,2c,1c\}
\end{array}
\right. ,
\end{align}
where $\epsilon' = 4\epsilon$. Note that the event when $S=\{2c\}$ does cause an error. Hence, the probability of the second kind of error event vanishes as $N\rightarrow\infty$ if 
\begin{align}
R_{1p} &\le I\lp x_{1};y_1,\what{y}_2 | x_{1c},x_{2c}\rp\\
R_{2c} &\le I\lp x_{2c};y_1,\what{y}_2 | x_{1}\rp\\
R_{2c}+R_{1p} &\le I\lp x_{2c},x_{1};y_1,\what{y}_2 | x_{1c}\rp\\
R_{1c}+R_{1p} &\le I\lp x_{1};y_1,\what{y}_2 | x_{2c}\rp\\
R_{1c}+R_{2c}+R_{1p} &\le I\lp x_{1},x_{2c};y_1,\what{y}_2 \rp.
\end{align}

{\flushleft \it Error probability analysis at receiver 2}:\par 
After receiving the two bin indices, receiver 2 can decode $(m_{1c},m_{2c},m_{2p})$, with effectively smaller candidate message sets, (namely, the bins,) for $m_{1c}$ and $m_{2c}$. Following the same line as \cite{ChongMotani_08}, it can be shown that (we omit the detailed analysis here), for all $0\le \lambda_1^{(ic)} \le 1$ and $\lambda_1^{(1c)} + \lambda_1^{(2c)} = 1$, the following region is achievable:
\begin{align}
R_{2p} &\le I\lp x_{2};y_2|x_{2c},x_{1c}\rp\\
R_{1c}+R_{2p} &\le I\lp x_{1c},x_{2};y_2|x_{2c}\rp + \lambda_1^{(1c)}\C_{12}\\
R_{2c}+R_{2p} &\le I\lp x_{2};y_2|x_{1c}\rp + \lambda_1^{(2c)}\C_{12}\\
R_{2c}+R_{1c}+R_{2p} &\le I\lp x_{2},x_{1c};y_2 \rp + \C_{12}.
\end{align}
Note that the performance of decoding the private message $m_{2p}$ does not gain from cooperation, since receiver 1 does not decode the private message $m_{2p}$.

Taking convex hull over all possible $\lambda_1^{(1c)}\in[0,1]$. Note that the bounds for $R_{2p}$ and $R_{2c}+R_{1c}+R_{2p}$ remain unchanged. Project the three-dimensional rate region to a two-dimensional space for any fixed $R_{2p} = r_{2p}$, we see that the convexifying procedure results in the following region:
\begin{align}
R_{1c}+r_{2p} &\le I\lp x_{1c},x_{2};y_2|x_{2c}\rp + \C_{12}\\
R_{2c}+r_{2p} &\le I\lp x_{2};y_2|x_{1c}\rp + \C_{12}\\
R_{2c}+R_{1c}+r_{2p} &\le I\lp x_{2},x_{1c};y_2 \rp + \C_{12}.
\end{align}

Hence the following rate region is achievable for receiver 2 to decode successfully:
\begin{align}
R_{2p} &\le I\lp x_{2};y_2|x_{2c},x_{1c}\rp\\
R_{1c}+R_{2p} &\le I\lp x_{1c},x_{2};y_2|x_{2c}\rp + \C_{12}\\
R_{2c}+R_{2p} &\le I\lp x_{2};y_2|x_{1c}\rp + \C_{12}\\
R_{2c}+R_{1c}+R_{2p} &\le I\lp x_{2},x_{1c};y_2 \rp + \C_{12}.
\end{align}


\section{Proof of Lemma \ref{lem_OutBd}}\label{app_PfOutBd}
{\flushleft(1) \emph{Bounds (\ref{eq_CutSetBd}) on $R_1, R_2$}}
\begin{proof}
One can directly use cut-set bounds. As an alternative, we give the following proof in which the decomposition of mutual informations is made clear.

We have the following bounds by Fano's inequality, data-processing inequality, and chain rule: if $R_1$ is achievable,
\begin{align}
&N(R_1-\epsilon_N)\\
&\overset{\aaaa}{\le} I(x_1^N;y_1^N,u_{21}^N) \overset{\bbbb}{\le} I(x_1^N;y_1^N,u_{21}^N,x_2^N) \overset{\cccc}{=}  I(x_1^N;y_1^N,u_{21}^N|x_2^N)\\
&\overset{\dddd}{=} I(x_1^N;y_1^N|x_2^N) +  I(x_1^N;u_{21}^N|y_1^N,x_2^N) = h(h_{11}x_1^N+z_1^N) - h(z_1^N) + I(x_1^N;u_{21}^N|y_1^N,x_2^N)\\
&\overset{\eeee}{\le} N\log(1+\SNR_1) + I(x_1^N;u_{21}^N|y_1^N,x_2^N),
\end{align}
where $\epsilon_N\rightarrow 0$ as $N\rightarrow \infty$. (a) is due to Fano's inequality and data processing inequality. (b) is due to the genie giving side information $x_1^N$ to receiver 1, ie., {\it conditioning reduces entropy}. (c) is due to the fact that $x_1^N$ and $x_2^N$ are independent. (d) is due to chain rule. (e) is due to the fact that i.i.d. Gaussian distribution maximizes differential entropy under covariance constraints.  

To upper bound $I(x_1^N;u_{21}^N|y_1^N,x_2^N)$, which corresponds to the enhancement from cooperation, we make use of the fact that $u_{21}^N$ is a function of $(y_1^N,y_2^N)$:
\begin{align}
&I(x_1^N;u_{21}^N|y_1^N,x_2^N)\\
& = h(x_1^N|y_1^N,x_2^N) - h(x_1^N|u_{21}^N,y_1^N,x_2^N) \overset{\aaaa}{\le} h(x_1^N|y_1^N,x_2^N) - h(x_1^N|u_{21}^N,y_1^N,x_2^N,y_2^N)\\
& \overset{\bbbb}{=} h(x_1^N|y_1^N,x_2^N) - h(x_1^N|y_1^N,x_2^N,y_2^N) = I(x_1^N;y_2^N|y_1^N,x_2^N)\\
& = h(y_2^N|y_1^N,x_2^N) - h(y_2^N|y_1^N,x_2^N,x_1^N) = h(h_{21}x_1^N+z_2^N|h_{11}x_1^N+z_1^N) - h(z_2^N)\\
&\le N\log\left(1+\frac{\INR_2}{1+\SNR_1}\right).
\end{align} 
(a) is due to the fact that conditioning reduces entropy. (b) is due to the fact that $u_{21}^N$ is a function of $(y_1^N,y_2^N)$.

Besides, it is trivial to see that $I(x_1^N;u_{21}^N|y_1^N,x_2^N)\le H(u_{21}^N) \le N\C_{21}$. Hence, (and similarly for $R_2$),
\begin{align}
R_1 &\le \log(1+\SNR_1)+ \min\left\{ \C_{21},\log\left(1+\frac{\INR_2}{1+\SNR_1}\right) \right\}\\
R_2 &\le \log(1+\SNR_2)+ \min\left\{ \C_{12},\log\left(1+\frac{\INR_1}{1+\SNR_2}\right) \right\}
\end{align}
\end{proof}

{\flushleft(2) \emph{Bounds (\ref{eq_ETWBound}) on $R_1+R_2$}}
\begin{figure}[htbp]
{\center
\includegraphics[width=2.5in]{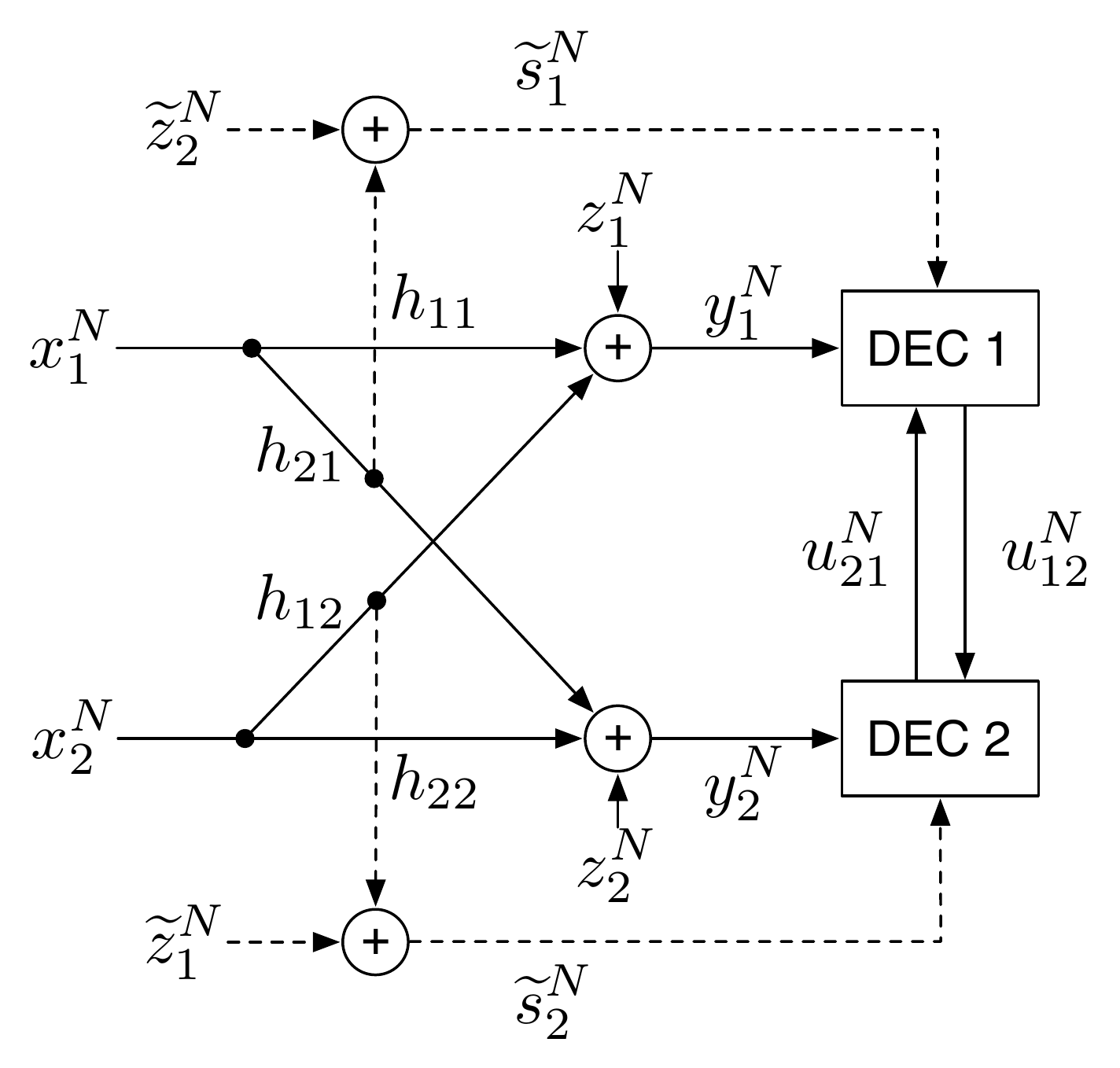}
\caption{Side Information Structure for Bound \eqref{eq_ETWBound}}
\label{fig_GenieOutBd1}
}
\end{figure}

\begin{proof}
Define
\begin{align}
s_1 &:= h_{21}x_1+z_2,\ s_2:= h_{12}x_2+z_1,\\
\wtild{s}_1 &:= h_{21}x_1+\wtild{z}_2,\ \wtild{s}_2:= h_{12}x_2+\wtild{z}_1,
\end{align}
where $\wtild{z}_1,\wtild{z}_2$ are i.i.d. $\mcal{CN}(0,1)$'s, independent of everything else. Note that $s_i$ and $\wtild{s}_i$ have the same marginal distribution, for $i=1,2$.

A genie gives side information $\wtild{s}_i^N$ to receiver $i$ (refer to Fig. \ref{fig_GenieOutBd1}.) Making use of Fano's inequality, data processing inequality, and the fact that Gaussian distribution maximizes conditional entropy subject to conditional variance constraints, we have: if $(R_1,R_2)$ is achievable,
\begin{align}
&N(R_1+R_2-\epsilon_N)\\
&\overset{\aaaa}{\le} I(x_1^N;y_1^N,u_{21}^N) +  I(x_2^N;y_2^N,u_{12}^N)\\
&\overset{\bbbb}{=} I(x_1^N;y_1^N) +  I(x_2^N;y_2^N) + I(x_1^N;u_{21}^N|y_1^N) +  I(x_2^N;u_{12}^N|y_2^N)\\
&\overset{\cccc}{\le} I(x_1^N;y_1^N,\wtild{s}_1^N) +  I(x_2^N;y_2^N,\wtild{s}_2^N) + H(u_{21}^N) + H(u_{12}^N)\\
&\overset{\dddd}{\le} h(y_1^N,\wtild{s}_1^N) - h(s_2^N,\wtild{z}_2^N) +  h(y_2^N,\wtild{s}_2^N) - h(s_1^N,\wtild{z}_1^N) + N\C_{21} + N\C_{12}\\
&\overset{\eeee}{=} h(y_1^N|\wtild{s}_1^N) + h(\wtild{s}_1^N) - h(s_2^N) - h(\wtild{z}_2^N) +  h(y_2^N|\wtild{s}_2^N) + h(\wtild{s}_1^N) - h(s_1^N) - h(\wtild{z}_1^N)\\&\quad + N\C_{21} + N\C_{12}\\
&\overset{}{=} h(y_1^N|\wtild{s}_1^N) - h(\wtild{z}_2^N) +  h(y_2^N|\wtild{s}_2^N) - h(\wtild{z}_1^N) + N\C_{21} + N\C_{12}\\
&\overset{\ffff}{\le} N\left\{\log\left(1+\INR_1+\frac{\SNR_1}{1+\INR_2}\right) + \log\left(1+\INR_2+\frac{\SNR_2}{1+\INR_1}\right) + \C_{21}+\C_{12}\right\},
\end{align}
where $\epsilon_N\rightarrow 0$ as $N\rightarrow \infty$. (a) follows from Fano's inequality and data processing inequality. (b) is due to chain rule. (c) is due to the genie giving side information $\wtild{s}_i^N$ to receiver $i$, $i=1,2$, and $I(x_i^N;u_{ji}^N|y_i^N) \le H(u_{ji}^N)$. (d) is due to the fact that $H(u_{ji}^N) \le N\C_{ji}$. (e) is due to chain rule. (f) is due to the fact that i.i.d. Gaussian distribution maximizes conditional entropy subject to conditional variance constraints. Note that alternatively the genie can give side informations $s_i^N$ to receiver $i$, as in \cite{EtkinTse_07}.

Hence,
\begin{align}
R_1+R_2 &\le \log\left(1+\INR_1+\frac{\SNR_1}{1+\INR_2}\right) + \log\left(1+\INR_2+\frac{\SNR_2}{1+\INR_1}\right)\\&\quad + \C_{21}+\C_{12}.
\end{align}
\end{proof}

{\flushleft(3) \emph{Bounds (\ref{eq_ZBound}) on $R_1+R_2$}}
\begin{figure}[htbp]
{\center
\includegraphics[width=3in]{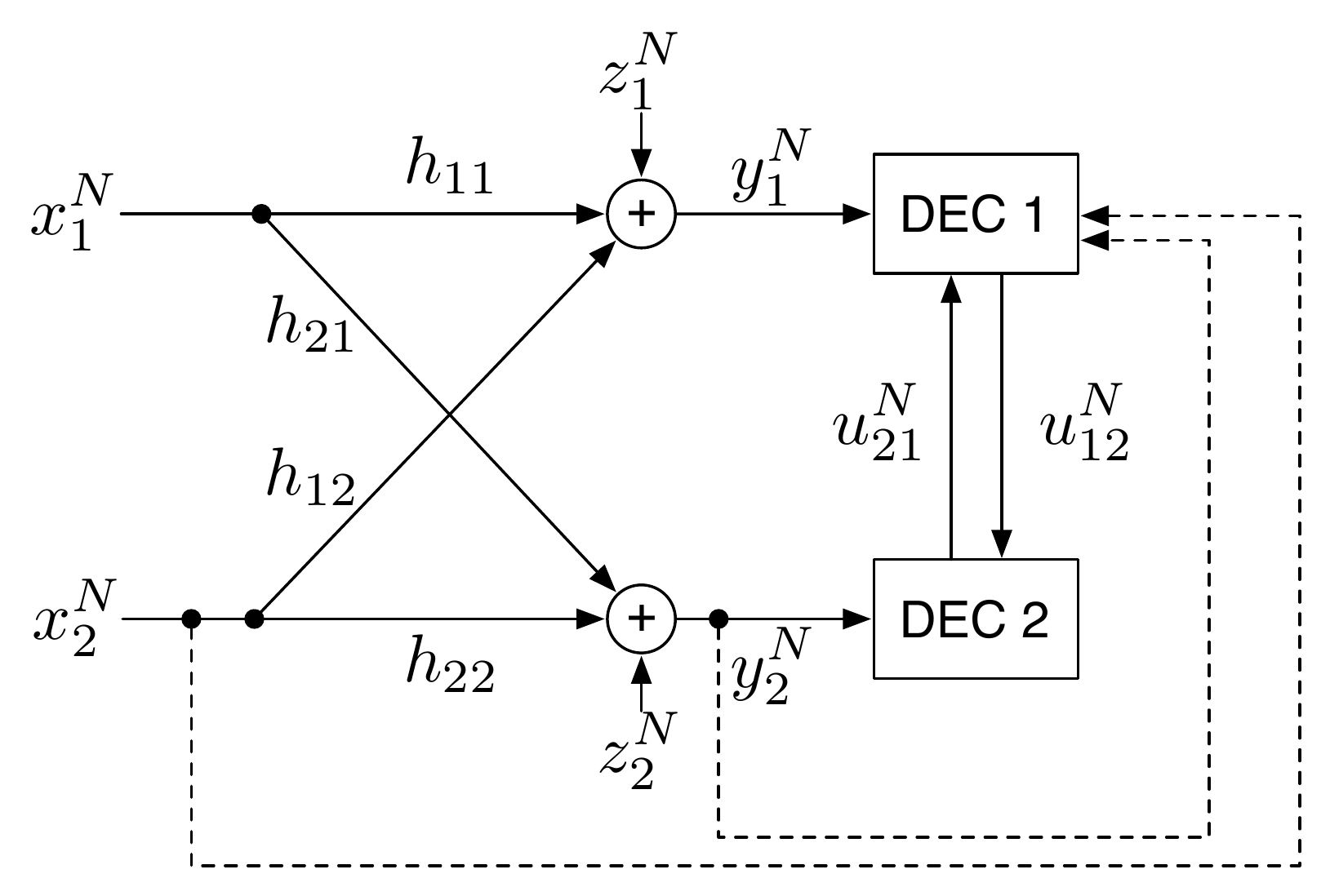}
\caption{Side Information Structure for Bound \eqref{eq_ZBound}}
\label{fig_GenieOutBd2}
}
\end{figure}

\begin{proof}
A genie gives side information $x_2^N$ and $y_2^N$ to receiver $1$ (refer to Fig. \ref{fig_GenieOutBd2}.) Making use of Fano's inequality, data processing inequality, the fact that $u_{21}^N$ is a function of $(y_1^N,y_2^N)$, and the fact that Gaussian distribution maximizes conditional entropy subject to conditional variance constraints, we have: if $(R_1,R_2)$ is achievable,
\begin{align}
&N(R_1+R_2-\epsilon_N)\\
&\le I(x_1^N;y_1^N,u_{21}^N) +  I(x_2^N;y_2^N,u_{12}^N)\\
&\overset{\aaaa}{\le} I(x_1^N;y_1^N,u_{21}^N,y_2^N,x_2^N) +  I(x_2^N;y_2^N) + I(x_2^N;u_{12}^N|y_2^N)\\
&\overset{\bbbb}{\le} I\lp x_1^N;y_1^N,u_{21}^N,y_2^N|x_2^N\rp + h\lp y_2^N\rp - h\lp s_1^N\rp + H\lp u_{12}^N\rp\\
&\overset{\cccc}{=} I\lp x_1^N;y_1^N,y_2^N|x_2^N\rp + h\lp y_2^N\rp - h\lp s_1^N\rp + H\lp u_{12}^N\rp\\
&= h\lp h_{11}x_1^N+z_1^N , s_1^N\rp - h\lp z_1^N,z_2^N\rp + h\lp y_2^N\rp - h\lp s_1^N\rp + H\lp u_{12}^N\rp\\
&= h\lp h_{11}x_1^N+z_1^N | s_1^N\rp - h\lp z_1^N,z_2^N\rp + h\lp y_2^N\rp + H\lp u_{12}^N\rp\\
&\overset{}{\le} N\log\left(1+\frac{\SNR_1}{1+\INR_2}\right) + N\log\left(1+\SNR_2+\INR_2\right) + N\C_{12},
\end{align}
where $\epsilon_N\rightarrow 0$ as $N\rightarrow \infty$. (a) is due to chain rule and the genie giving side information $x_2^N$ and $y_2^N$ to receiver 1. (b) is due to the fact that $x_1^N$ and $x_2^N$ are independent, and $I(x_2^N;u_{12}^N|y_2^N) \le H(u_{12}^N)$. (c) is due to the fact that $u_{21}^N$ is a function of $(y_1^N,y_2^N)$.

Hence, (and similarly if we gives side information $x_1^N$ to receiver $2$), we have
\begin{align}
R_1+R_2 &\le \log\left(1+\SNR_2+\INR_2\right) + \log\left(1+\frac{\SNR_1}{1+\INR_2}\right) + \C_{12}\\
R_1+R_2 &\le \log\left(1+\SNR_1+\INR_1\right) + \log\left(1+\frac{\SNR_2}{1+\INR_1}\right) + \C_{21}
\end{align}
\end{proof}

{\flushleft(4) \emph{Bounds (\ref{eq_SIMOBd}) on $R_1+R_2$}}
\begin{proof}
This is straightforward cut-set upper bound: if $(R_1,R_2)$ is achievable,
\begin{align}
&N(R_1+R_2-\epsilon_N)\\
&\le I\lp x_1^N, x_2^N ; y_1^N, y_2^N \rp = h\lp y_1^N,y_2^N\rp - h\lp z_1^N,z_2^N\rp\\
&\le N\log\lp 1+\SNR_1+\SNR_2+\INR_1+\INR_2 + |h_{11}h_{22} - h_{12}h_{21}|^2 \rp,
\end{align}
where $\epsilon_N\rightarrow 0$ as $N\rightarrow \infty$.

Hence,
\begin{align}
&R_1+R_2 \le \log\lp 1+\SNR_1+\SNR_2+\INR_1+\INR_2 + |h_{11}h_{22} - h_{12}h_{21}|^2 \rp.
\end{align}
\end{proof}

{\flushleft(5) \emph{Bounds (\ref{eq_Slope2Bd}) on $2R_1+R_2$ and $R_1+2R_2$}}
\begin{figure}[htbp]
{\center
\includegraphics[width=2.5in]{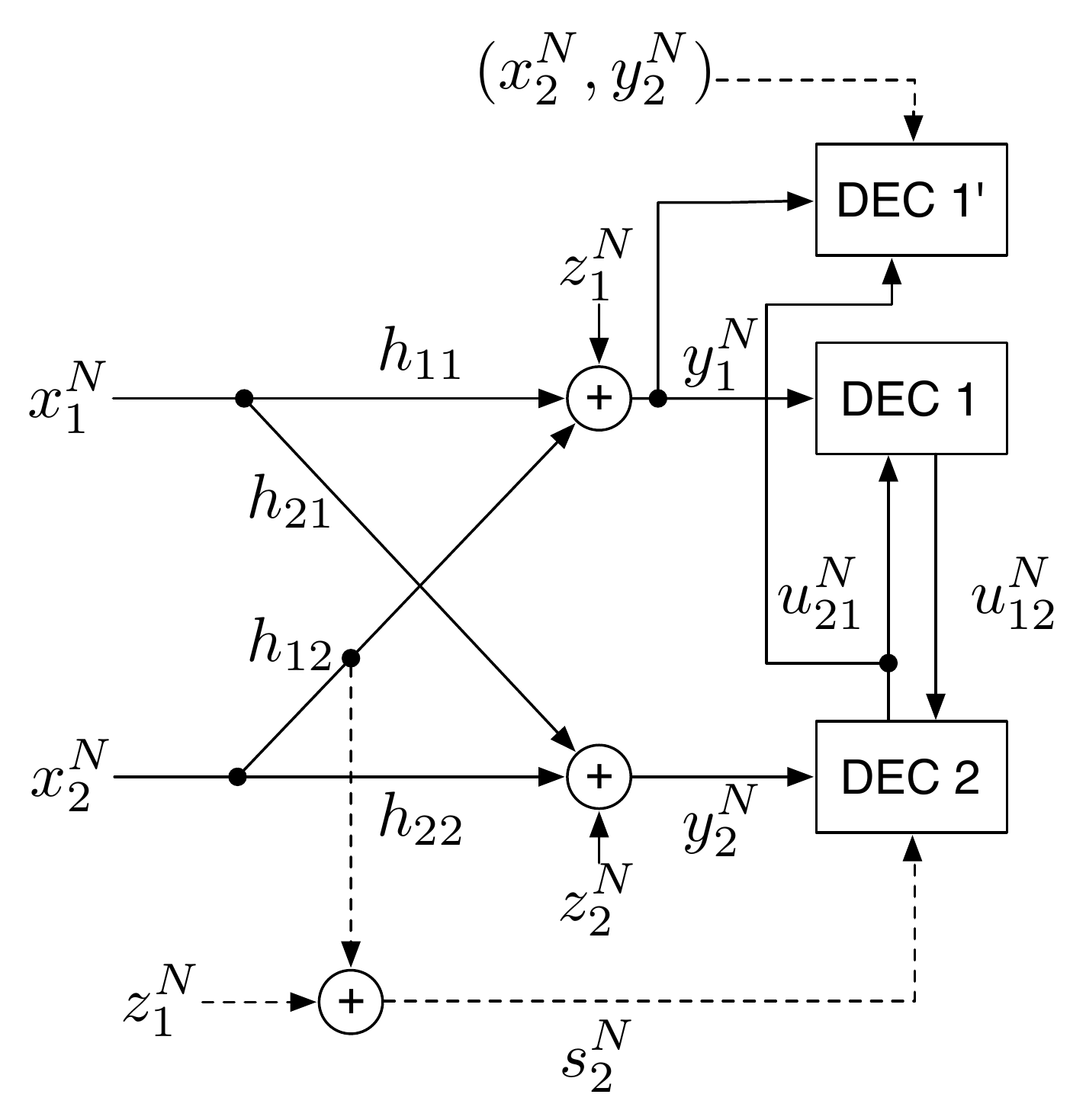}
\caption{Side Information Structure for Bound \eqref{eq_Slope2Bd}}
\label{fig_GenieOutBd3}
}
\end{figure}

\begin{proof}
A genie gives side information $x_2^N$ and $y_2^N$ to one of the two receiver $1$'s, and side information $s_2^N$ to receiver $2$ (refer to Fig. \ref{fig_GenieOutBd3}.) Making use of Fano's inequality, data processing inequality, the fact that $u_{21}^N$ is a function of $(y_1^N,y_2^N)$, and the fact that Gaussian distribution maximizes conditional entropy subject to conditional variance constraints, we have: if $(R_1,R_2)$ is achievable,
\begin{align}
&N(2R_1+R_2-\epsilon_N)\\
&\overset{\aaaa}{\le} I(x_1^N;y_1^N,u_{21}^N) +  I(x_1^N;y_1^N,u_{21}^N)+ I(x_2^N;y_2^N,u_{12}^N)\\
&\overset{\bbbb}{\le} I(x_1^N;y_1^N,u_{21}^N,y_2^N,x_2^N) + I(x_1^N;y_1^N) + I(x_2^N;y_2^N,s_2^N)+ I(x_1^N;u_{21}^N|y_1^N) +I(x_2^N;u_{12}^N|y_2^N)\\
&\overset{\cccc}{\le} I(x_1^N;y_1^N,u_{21}^N,y_2^N|x_2^N) + I(x_1^N;y_1^N) + I(x_2^N;y_2^N,s_2^N)+ H\lp u_{21}^N\rp + H\lp u_{12}^N\rp\\
&\overset{\dddd}{=} I(x_1^N;y_1^N,y_2^N|x_2^N) + I(x_1^N;y_1^N) + I(x_2^N;y_2^N,s_2^N)+ H\lp u_{21}^N\rp + H\lp u_{12}^N\rp\\
&= h\lp h_{11}x_1^N+z_1^N , s_1^N\rp - h\lp z_1^N,z_2^N\rp + h(y_1^N) - h(s_2^N) + h(y_2^N,s_2^N) - h(s_1^N,z_1^N)\\&\quad + H\lp u_{21}^N\rp + H\lp u_{12}^N\rp\\
&= h\lp h_{11}x_1^N+z_1^N | s_1^N\rp - h\lp z_1^N,z_2^N\rp + h(y_1^N) + h(y_2^N|s_2^N) - h(z_1^N)+ H\lp u_{21}^N\rp + H\lp u_{12}^N\rp\\
&\le N \Bigg\{ \log\left(1+\frac{\SNR_1}{1+\INR_2}\right) + \log\lp 1+ \SNR_1+\INR_1\rp + \log\left(1+\INR_2+\frac{\SNR_2}{1+\INR_1}\right)\\&\quad\quad\quad + \C_{21}+\C_{12} \Bigg\},
\end{align}
where $\epsilon_N\rightarrow 0$ as $N\rightarrow \infty$. (a) follows from Fano's inequality and data processing inequality. (b) is due to chain rule and the genie giving side information $x_2^N$ and $y_2^N$ to one of the receiver 1's and side information $s_2^N$ to receiver 2. (c) is due to the fact that $x_1^N,x_2^N$ are independent and $I(x_i^N;u_{ji}^N|y_i^N) \le H(u_{ji}^N)$. (d) is due to the fact that $u_{21}^N$ is a function of $(y_1^N,y_2^N)$.

Hence, 
\begin{align}
2R_1 + R_2 & \le \log\left(1+\INR_2+\frac{\SNR_2}{1+\INR_1}\right) + \log\left(1+\frac{\SNR_1}{1+\INR_2}\right)\\&\quad + \log\lp 1+ \SNR_1+\INR_1\rp + \C_{21}+\C_{12}
\end{align}

Similarly, 
\begin{align}
R_1 + 2R_2 & \le \log\left(1+\INR_1+\frac{\SNR_1}{1+\INR_2}\right) + \log\left(1+\frac{\SNR_2}{1+\INR_1}\right)\\&\quad + \log\lp 1+ \SNR_2+\INR_2\rp + \C_{12}+\C_{21}
\end{align}
\end{proof} 

{\flushleft(6) \emph{Bounds (\ref{eq_Slope2CutSetBd}) on $2R_1+R_2$ and $R_1+2R_2$}}
\begin{figure}[htbp]
{\center
\includegraphics[width=3in]{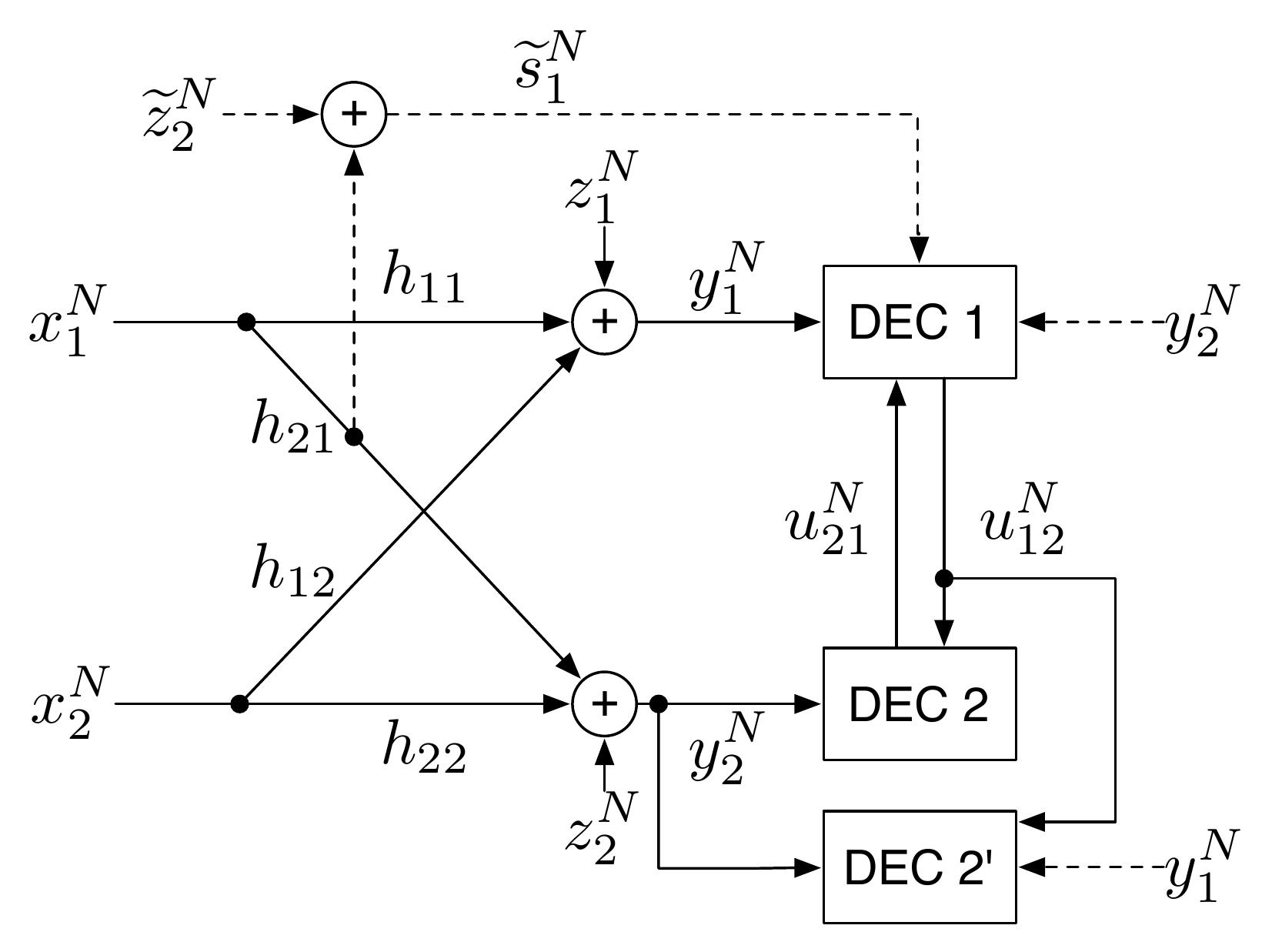}
\caption{Side Information Structure for Bound \eqref{eq_Slope2CutSetBd}}
\label{fig_GenieOutBd4}
}
\end{figure}

\begin{proof}
A genie gives side information $\wtild{s}_1^N,y_2^N$ to receiver $1$, and side information $y_1^N$ to one of the receiver 2's (refer to Fig. \ref{fig_GenieOutBd4}.) Making use of Fano's inequality, data processing inequality, the fact that $u_{12}^N, u_{21}^N$ are functions of $(y_1^N,y_2^N)$, and the fact that Gaussian distribution maximizes conditional entropy subject to conditional variance constraints, we have: if $(R_1,R_2)$ is achievable,
\begin{align}
&N(R_1+2R_2-\epsilon_N)\\
&\le I(x_1^N;y_1^N,u_{21}^N) +  I(x_2^N;y_2^N,u_{12}^N)+ I(x_2^N;y_2^N,u_{12}^N)\\
&\overset{\aaaa}{\le} I(x_1^N;y_1^N,u_{21}^N,y_2^N,\wtild{s}_1^N) + I(x_2^N;y_2^N,u_{12}^N,y_1^N)+ I(x_2^N;y_2^N) + I(x_2^N;u_{12}^N|y_2^N)\\
&\overset{\bbbb}{\le} I(x_1^N;y_1^N,u_{21}^N,y_2^N|\wtild{s}_1^N) + I(x_1^n; \wtild{s}_1^N)+ I(x_2^N;y_2^N,u_{12}^N,y_1^N)+ I(x_2^N;y_2^N) + H(u_{12}^N)\\
&\overset{\cccc}{\le} I(x_1^N;y_1^N,y_{2}^N|\wtild{s}_1^N) + I(x_2^N;y_1^N,y_2^N) + h(\wtild{s}_1^N) - h(z_2^N) + h(y_2^N) - h(s_1^N) + N\C_{12}\\
&\overset{\dddd}{\le} I(x_1^N;y_1^N,y_{2}^N|\wtild{s}_1^N) + I(x_2^N;y_1^N,y_2^N | x_1^N,\wtild{s}_1^N) + h(y_2^N) - h(z_2^N) + N\C_{12}\\
&= I(x_1^N,x_2^N;y_1^N,y_{2}^N|\wtild{s}_1^N) + h(y_2^N) - h(z_2^N) + N\C_{12}\\
&= h(y_1^N,y_{2}^N|\wtild{s}_1^N) + h(y_2^N) - h(z_1^N,z_2^N) - h(z_2^N) + N\C_{12}\\
&\overset{\eeee}{\le} N\log\lp 1+ \frac{\SNR_1}{1+\INR_2} + \INR_1 + \SNR_2 + \frac{\INR_2}{1+\INR_2} + \frac{|h_{11}h_{22} - h_{12}h_{21}|^2}{1+\INR_2}\rp\\
&\quad + N\log\lp1+\SNR_2+\INR_2\rp + N\C_{12},
\end{align}
where $\epsilon_N\rightarrow 0$ as $N\rightarrow \infty$. (a) is due to the genie giving side information $\wtild{s}_1^N,y_2^N$ to receiver $1$, and side information $y_1^N$ to one of the receiver 2's. (b) is due to chain rule and the fact that $I(x_2^N;u_{12}^N|y_2^N) \le H(u_{12}^N)$. (c) is due to the fact that $u_{21}^N$ and $u_{12}^N$ are both functions of $(y_1^N,y_2^N)$, and that $H(u_{12}^N)\le N\C_{12}$. (d) is due to the fact that conditioning reduces entropy and that $x_2^N$ and $(x_1^N,\wtild{s}_1^N)$ are independent. (e) is due to the fact that Gaussian distribution maximizes conditional entropy subject to conditional variance constraints.

Hence,
\begin{align}
R_1+2R_2 &\le \log\lp 1+ \frac{\SNR_1}{1+\INR_2} + \INR_1 + \SNR_2 + \frac{\INR_2}{1+\INR_2} + \frac{|h_{11}h_{22} - h_{12}h_{21}|^2}{1+\INR_2}\rp\\
&\quad + \log\lp1+\SNR_2+\INR_2\rp + \C_{12}
\end{align}

Similarly,
\begin{align}
2R_1+R_2 &\le \log\lp 1+ \frac{\SNR_2}{1+\INR_1} + \INR_2 + \SNR_1 + \frac{\INR_1}{1+\INR_1} + \frac{|h_{11}h_{22} - h_{12}h_{21}|^2}{1+\INR_1}\rp\\
&\quad + \log\lp1+\SNR_1+\INR_1\rp + \C_{21}
\end{align}
\end{proof}


\section{Proof of Claim \ref{claim_Corner}, Claim \ref{claim_WeakGap}, Claim \ref{claim_MixedGap}, and Claim \ref{claim_StrongGap}}\label{app_PfClaims}
\subsection{Proof of Claim \ref{claim_Corner}}
\begin{proof}
To show (a), since we have four possible $R_1+2R_2$ bounds, we distinguish into 4 cases: 
{\flushleft (1)} If the bound
\begin{align}
R_1+2R_2 &\le I\lp x_1,x_{2c}; y_1| x_{1c}\rp + I\lp x_{1c},x_2; y_2\rp + I\lp x_2; y_2| x_{1c},x_{2c}\rp + \C_{12} + \lp\C_{21}-\xi_1\rp^+
\end{align}
is active, note that the point $(R^*_1,R^*_2)$ where the $R_1+2R_2$ bound and the $2R_1+R_2$ bound \eqref{eq_Trouble} intersect, satisfies
\begin{align}
&3R_1^*+3R_2^*\\
&=  I\lp x_1,x_{2c}; y_1,\what{y}_2\rp + I\lp x_1; y_1| x_{1c},x_{2c}\rp + I\lp x_{1c},x_2; y_2| x_{2c}\rp + \C_{12}\\
&\quad + I\lp x_1,x_{2c}; y_1| x_{1c}\rp + I\lp x_{1c},x_2; y_2\rp + I\lp x_2; y_2| x_{1c},x_{2c}\rp + \C_{12} + \lp\C_{21}-\xi_1\rp^+\\
&= I\lp x_1,x_{2c}; y_1,\what{y}_2\rp + I\lp x_2; y_2| x_{1c},x_{2c}\rp\\
&\quad + I\lp x_1; y_1| x_{1c},x_{2c}\rp + I\lp x_{1c},x_2; y_2\rp + \C_{12}\\
&\quad + I\lp x_1,x_{2c}; y_1| x_{1c}\rp + I\lp x_{1c},x_2; y_2| x_{2c}\rp + \C_{12} + \lp\C_{21}-\xi_1\rp^+\\
&= \eqref{eq_SumBd2} + \eqref{eq_SumBd5} + \eqref{eq_SumBd3},
\end{align}
which is greater than three times the active sum rate bound.

{\flushleft (2)} If the bound
\begin{align}
R_1+2R_2 &\le I\lp x_1,x_{2c}; y_1| x_{1c}\rp + I\lp x_{2c};y_1|x_1\rp + I\lp x_{1c},x_2; y_2| x_{2c}\rp + I\lp x_2; y_2| x_{1c},x_{2c}\rp\\&\quad + \C_{12} + \lp\C_{21}-\xi_1\rp^+
\end{align}
is active, note that the point $(R^*_1,R^*_2)$ where the $R_1+2R_2$ bound and the $2R_1+R_2$ bound \eqref{eq_Trouble} intersect, satisfies
\begin{align}
&3R_1^*+3R_2^*\\
&=  I\lp x_1,x_{2c}; y_1,\what{y}_2\rp + I\lp x_1; y_1| x_{1c},x_{2c}\rp + I\lp x_{1c},x_2; y_2| x_{2c}\rp + \C_{12}\\
&\quad + I\lp x_1,x_{2c}; y_1| x_{1c}\rp + I\lp x_{2c};y_1|x_1\rp + I\lp x_{1c},x_2; y_2| x_{2c}\rp + I\lp x_2; y_2| x_{1c},x_{2c}\rp\\&\quad + \C_{12} + \lp\C_{21}-\xi_1\rp^+\\
&= I\lp x_1,x_{2c}; y_1,\what{y}_2\rp + I\lp x_2; y_2| x_{1c},x_{2c}\rp\\
&\quad + I\lp x_1; y_1| x_{1c},x_{2c}\rp + I\lp x_{2c};y_1|x_1\rp + I\lp x_{1c},x_2; y_2| x_{2c}\rp + \C_{12}\\
&\quad + I\lp x_1,x_{2c}; y_1| x_{1c}\rp + I\lp x_{1c},x_2; y_2| x_{2c}\rp + \C_{12} + \lp\C_{21}-\xi_1\rp^+\\
&= \eqref{eq_SumBd2} + \eqref{eq_SumBd6} + \eqref{eq_SumBd3},
\end{align}
which is greater than three times the active sum rate bound.

{\flushleft (3)} If the bound
\begin{align}
R_1+2R_2 &\le I\lp x_1,x_{2c}; y_1,\what{y}_2| x_{1c}\rp + I\lp x_{1c},x_2; y_2\rp + I\lp x_2; y_2| x_{1c},x_{2c}\rp + \C_{12}
\end{align}
is active, note that the point $(R^*_1,R^*_2)$ where the $R_1+2R_2$ bound and the $2R_1+R_2$ bound \eqref{eq_Trouble} intersect, satisfies
\begin{align}
&3R_1^*+3R_2^*\\
&=  I\lp x_1,x_{2c}; y_1,\what{y}_2\rp + I\lp x_1; y_1| x_{1c},x_{2c}\rp + I\lp x_{1c},x_2; y_2| x_{2c}\rp + \C_{12}\\
&\quad + I\lp x_1,x_{2c}; y_1,\what{y}_2| x_{1c}\rp + I\lp x_{1c},x_2; y_2\rp + I\lp x_2; y_2| x_{1c},x_{2c}\rp + \C_{12}\\
&= I\lp x_1,x_{2c}; y_1,\what{y}_2\rp + I\lp x_2; y_2| x_{1c},x_{2c}\rp\\
&\quad + I\lp x_1; y_1| x_{1c},x_{2c}\rp + I\lp x_{1c},x_2; y_2\rp + \C_{12}\\
&\quad +  I\lp x_1,x_{2c}; y_1,\what{y}_2| x_{1c}\rp + I\lp x_{1c},x_2; y_2| x_{2c}\rp + \C_{12}\\
&= \eqref{eq_SumBd2} + \eqref{eq_SumBd5} + \eqref{eq_SumBd4},
\end{align}
which is greater than three times the active sum rate bound.

{\flushleft (4)} If the bound
\begin{align}
&R_1+2R_2 \le\\ &I\lp x_1,x_{2c}; y_1,\what{y}_2| x_{1c}\rp + I\lp x_{2c};y_1|x_1\rp + I\lp x_{1c},x_2; y_2| x_{2c}\rp + I\lp x_2; y_2| x_{1c},x_{2c}\rp + \C_{12}
\end{align}
is active, note that the point $(R^*_1,R^*_2)$ where the $R_1+2R_2$ bound and the $2R_1+R_2$ bound \eqref{eq_Trouble} intersect, satisfies
\begin{align}
&3R_1^*+3R_2^*\\
&=  I\lp x_1,x_{2c}; y_1,\what{y}_2\rp + I\lp x_1; y_1| x_{1c},x_{2c}\rp + I\lp x_{1c},x_2; y_2| x_{2c}\rp + \C_{12}\\
&\quad + I\lp x_1,x_{2c}; y_1,\what{y}_2| x_{1c}\rp + I\lp x_{2c};y_1|x_1\rp + I\lp x_{1c},x_2; y_2| x_{2c}\rp + I\lp x_2; y_2| x_{1c},x_{2c}\rp + \C_{12}\\
&= I\lp x_1,x_{2c}; y_1,\what{y}_2\rp + I\lp x_2; y_2| x_{1c},x_{2c}\rp\\
&\quad + I\lp x_1; y_1| x_{1c},x_{2c}\rp + I\lp x_{2c};y_1|x_1\rp + I\lp x_{1c},x_2; y_2| x_{2c}\rp + \C_{12}\\
&\quad +  I\lp x_1,x_{2c}; y_1,\what{y}_2| x_{1c}\rp + I\lp x_{1c},x_2; y_2| x_{2c}\rp + \C_{12}\\
&= \eqref{eq_SumBd2} + \eqref{eq_SumBd6} + \eqref{eq_SumBd4},
\end{align}
which is greater than three times the active sum rate bound.

Hence, we conclude that in case (a), the corner point where $R_1+R_2$ bound and $R_1+2R_2$ bound intersect can be achieved.



To show (b), since we have two possible $R_2$ bounds, we distinguish into 2 cases:
{\flushleft (1)} If the bound
\begin{align}
R_2 &\le I\lp x_2; y_2| x_{1c}\rp + \C_{12}
\end{align}
is active, note that the point $(R^*_1,R^*_2)$ where the $R_2$ bound and the $2R_1+R_2$ bound \eqref{eq_Trouble} intersect, satisfies
\begin{align}
&2R_1^*+2R_2^*\\
&=  I\lp x_1,x_{2c}; y_1,\what{y}_2\rp + I\lp x_1; y_1| x_{1c},x_{2c}\rp + I\lp x_{1c},x_2; y_2| x_{2c}\rp + \C_{12}\\
&\quad + I\lp x_2; y_2| x_{1c}\rp + \C_{12}\\
&= I\lp x_1,x_{2c}; y_1,\what{y}_2\rp + I\lp x_2; y_2| x_{1c},x_{2c}\rp\\
&\quad + I\lp x_1; y_1| x_{1c},x_{2c}\rp + I\lp x_{1c},x_2; y_2\rp + \C_{12}\\
&\quad +  I\lp x_2; y_2| x_{1c}\rp + I\lp x_{1c}; y_2|x_{2c}\rp - I\lp x_{1c},x_2;y_2\rp + \C_{12}\\
&= \eqref{eq_SumBd2} + \eqref{eq_SumBd5} + \Big[ I\lp x_{1c}; y_2|x_{2c}\rp - I\lp x_{1c}; y_2\rp + \C_{12} \Big]\\
&\overset{(**)}{\ge } \eqref{eq_SumBd2} + \eqref{eq_SumBd5},
\end{align}
which is greater than two times the active sum rate bound. $(**)$ is due to
\begin{align}
I\lp x_{1c}; y_2|x_{2c}\rp &= I\lp x_{1c}; y_2,x_{2c}\rp - I\lp x_{1c};x_{2c}\rp = I\lp x_{1c}; y_2,x_{2c}\rp\\
&\ge I\lp x_{1c}; y_2\rp,
\end{align}
since $x_{1c}$ and $x_{2c}$ are independent.

{\flushleft (2)} If the bound
\begin{align}
R_2 &\le I\lp x_{2c}; y_1| x_1\rp + I\lp x_2; y_2| x_{1c}, x_{2c}\rp
\end{align}
is active, note that the point $(R^*_1,R^*_2)$ where the $R_2$ bound and the $2R_1+R_2$ bound \eqref{eq_Trouble} intersect, satisfies
\begin{align}
&2R_1^*+2R_2^*\\
&=  I\lp x_1,x_{2c}; y_1,\what{y}_2\rp + I\lp x_1; y_1| x_{1c},x_{2c}\rp + I\lp x_{1c},x_2; y_2| x_{2c}\rp + \C_{12}\\
&\quad + I\lp x_{2c}; y_1| x_1\rp + I\lp x_2; y_2| x_{1c}, x_{2c}\rp\\
&= I\lp x_1,x_{2c}; y_1,\what{y}_2\rp + I\lp x_2; y_2| x_{1c},x_{2c}\rp\\
&\quad +I\lp x_1; y_1| x_{1c},x_{2c}\rp + I\lp x_{2c};y_1|x_1\rp + I\lp x_{1c},x_2; y_2| x_{2c}\rp + \C_{12}\\
&= \eqref{eq_SumBd2} + \eqref{eq_SumBd6},
\end{align}
which is greater than two times the active sum rate bound.

Hence, we conclude that in case (b), the corner point where $R_1+R_2$ bound and $R_2$ bound intersect can be achieved.
\end{proof}

\subsection{Proof of Claim \ref{claim_WeakGap}}
\begin{proof}
(Keep in mind $\Delta_2 = \max\{\SNR_{2p},1\}$)
{\flushleft (1)} $R_1$ bound:
We have two bounds. First, $I\lp x_1; y_1| x_{2c}\rp = \log\lp1+\frac{\SNR_1}{1+\INR_{1p}}\rp$, which is within 2 bits to the upper bound $\log\lp1+\SNR_1+\INR_2\rp$. Second, 
\begin{align}
&I\lp x_1;y_1| x_{1c},x_{2c}\rp + I\lp x_{1c},x_2; y_2| x_{2c}\rp +\C_{12}\\
&= \log\lp1+\frac{\SNR_{1p}}{1+\INR_{1p}}\rp + \log\lp \frac{1+\SNR_{2p}+\INR_2}{1+\INR_{2p}}\rp + \C_{12}\\
&\ge \log\lp\frac{1+\SNR_1+\INR_2}{1+\INR_{1p}}\rp - 1.
\end{align}
Hence, if the second bound is active, it is within 2 bits to the upper bound\\ $\log\lp1+\SNR_1+\INR_2\rp$.

{\flushleft (2)} $R_2$ bound:
We have two bounds. First, $I\lp x_2; y_2| x_{1c}\rp + \C_{12} = \log\lp1+\frac{\SNR_2}{1+\INR_{2p}}\rp +\C_{12}$. If the first bound is active, it is within 2 bits to the upper bound $\log\lp1+\SNR_2+\INR_1\rp$. Second,
\begin{align}
&I\lp x_{2c}; y_1| x_1\rp + I\lp x_2; y_2| x_{1c}, x_{2c}\rp\\
&= \log\lp\frac{1+\INR_1}{1+\INR_{1p}}\rp + \log\lp\frac{1+\SNR_{2p}+\INR_{2p}}{1+\INR_{2p}}\rp\\
&\ge \log\lp\frac{1+\SNR_2+\INR_1}{1+\INR_{2p}}\rp - 1.
\end{align}
Hence, the second bound is within 2 bits to the upper bound $\log\lp1+\SNR_2+\INR_1\rp$.

{\flushleft (3)} $R_1+R_2$ bound:
We have six bounds for $R_1+R_2$, investigated as follows:
\begin{itemize}
\item First,
\begin{align}
&I\lp x_1,x_{2c}; y_1\rp + I\lp x_2; y_2| x_{1c},x_{2c}\rp + \lp\C_{21}-\xi_1\rp^+\\
&= \log\lp\frac{1+\SNR_1+\INR_1}{1+\INR_{1p}}\rp + \log\lp 1+\frac{\SNR_{2p}}{1+\INR_{2p}}\rp + \lp\C_{21}-\xi_1\rp^+,
\end{align}
which is within $2+\log3 = \log12$ bits to the upper bound
\begin{align}
R_1+R_2 &\le \log\left(1+\SNR_1+\INR_1\right) + \log\left(1+\frac{\SNR_2}{1+\INR_1}\right) + \C_{21}.
\end{align}
\item Second, 
\begin{align}
&I\lp x_1,x_{2c}; y_1,\what{y}_2\rp + I\lp x_2; y_2| x_{1c},x_{2c}\rp \\
&= \log\lp\frac{(1+\Delta_2)(1+\SNR_1+\INR_1)+\SNR_2+\INR_2+|h_{11}h_{22}-h_{12}h_{21}|^2}{(1+\Delta_2)(1+\INR_{1p})+\SNR_{2p}}\rp\\
&\quad + \log\lp1+\frac{\SNR_{2p}}{1+\INR_{2p}}\rp\\
&\overset{\aaaa}{\ge} \log\lp\frac{1+\SNR_1+\INR_1+\SNR_2+\INR_2+|h_{11}h_{22}-h_{12}h_{21}|^2}{5\Delta_2}\rp\\&\quad + \log\lp1+\SNR_{2p}\rp -1\\
&= \log\lp 1+\SNR_1+\INR_1+\SNR_2+\INR_2+|h_{11}h_{22}-h_{12}h_{21}|^2\rp\\&\quad + \log\lp \frac{1+\SNR_{2p}}{\Delta_2}\rp - \log10\\
&\overset{\bbbb}{\ge} \log\lp 1+\SNR_1+\INR_1+\SNR_2+\INR_2+|h_{11}h_{22}-h_{12}h_{21}|^2\rp - \log10,
\end{align}
where (a) is due to $(1+\Delta_2)(1+\INR_{1p})+\SNR_{2p} \le (1+\Delta_2)2 + \Delta_2 \le 5\Delta_2$ since $\Delta_2 = \max\{\SNR_{2p},1\}$ and $\INR_{1p}\le 1$. (b) is due to $\Delta_2 = \max\{\SNR_{2p},1\} \le 1+\SNR_{2p}$.

This lower bound is within $\log10$ bits to the upper bound
\begin{align}
R_1+R_2 &\le \log\lp1+\SNR_1+\INR_1+\SNR_2+\INR_2+|h_{11}h_{22}-h_{12}h_{21}|^2\rp.
\end{align}


\item Third,
\begin{align}
&I\lp x_1,x_{2c}; y_1| x_{1c}\rp + I\lp x_{1c},x_2; y_2| x_{2c}\rp + \C_{12} + \lp\C_{21}-\xi_1\rp^+\\
&= \log\lp\frac{1+\SNR_{1p}+\INR_1}{1+\INR_{1p}}\rp + \log\lp\frac{1+\SNR_{2p}+\INR_2}{1+\INR_{2p}}\rp + \C_{12} + \lp\C_{21}-\xi_1\rp^+,
\end{align}
which is within $2+\log3 = \log12$ bits to the upper bound
\begin{align}
R_1+R_2 &\le \log\left(1+\INR_1+\frac{\SNR_1}{1+\INR_2}\right) + \log\left(1+\INR_2+\frac{\SNR_2}{1+\INR_1}\right)\\&\quad + \C_{21}+\C_{12}.
\end{align}

\item Fourth,
\begin{align}
&I\lp x_1,x_{2c}; y_1,\what{y}_2| x_{1c}\rp + I\lp x_{1c},x_2; y_2| x_{2c}\rp + \C_{12}\\
&= I\lp x_{2c}; y_1,\what{y}_2| x_{1c}\rp + I\lp x_1; y_1,\what{y}_2| x_{1c},x_{2c}\rp + I\lp x_{1c},x_2; y_2| x_{2c}\rp + \C_{12}\\
&\ge I\lp x_{2c}; \what{y}_2| x_{1c}\rp + I\lp x_1; y_1| x_{1c},x_{2c}\rp + I\lp x_{1c},x_2; y_2| x_{2c}\rp + \C_{12}\\
&\overset{\aaaa}{\ge} I\lp x_{2c}; y_2| x_{1c}\rp -1 + I\lp x_1; y_1| x_{1c},x_{2c}\rp + I\lp x_{1c},x_2; y_2| x_{2c}\rp + \C_{12}\\
&\overset{\bbbb}{\ge} I\lp x_1; y_1| x_{1c},x_{2c}\rp + I\lp x_{1c},x_2; y_2\rp + \C_{12} - 1\\
&= \log\lp1+\frac{\SNR_{1p}}{1+\INR_{1p}}\rp + \log\lp\frac{1+\SNR_2+\INR_2}{1+\INR_{2p}}\rp + \C_{12}-1,
\end{align}
which is within $3$ bits to the upper bound
\begin{align}
R_1+R_2 &\le \log\left(1+\SNR_2+\INR_2\right) + \log\left(1+\frac{\SNR_1}{1+\INR_2}\right) + \C_{12}.
\end{align}
Note that (a) is due to
\begin{align}
&I\lp x_{2c}; \what{y}_2| x_{1c}\rp = \log\lp\frac{1+\Delta_2+\INR_{2p}+\SNR_2}{1+\Delta_2+\INR_{2p}+\SNR_{2p}}\rp\\
&\ge \log\lp\frac{1+\INR_{2p}+\SNR_2}{1+(1+\SNR_{2p})+\INR_{2p}+\SNR_{2p}}\rp\\
&\ge \log\lp\frac{1+\INR_{2p}+\SNR_2}{1+\INR_{2p}+\SNR_{2p}}\rp -1 = I\lp x_{2c}; y_2| x_{1c}\rp - 1.
\end{align}
(b) is due to
\begin{align}
&I\lp x_{2c}; y_2| x_{1c}\rp + I\lp x_{1c},x_2; y_2| x_{2c}\rp =  I\lp x_{2c}; y_2,x_{1c}\rp + I\lp x_{1c},x_2; y_2| x_{2c}\rp\\
&\ge  I\lp x_{2c}; y_2\rp + I\lp x_{1c},x_2; y_2| x_{2c}\rp =  I\lp x_{1c},x_2,x_{2c}; y_2\rp\\
&= I\lp x_{1c},x_2; y_2\rp.
\end{align}


\item Fifth,
\begin{align}
&I\lp x_1; y_1| x_{1c},x_{2c}\rp + I\lp x_{1c},x_2; y_2\rp + \C_{12}\\
&= \log\lp1+\frac{\SNR_{1p}}{1+\INR_{1p}}\rp + \log\lp\frac{1+\SNR_2+\INR_2}{1+\INR_{2p}}\rp + \C_{12},
\end{align}
which is within $2$ bits to the upper bound
\begin{align}
R_1+R_2 &\le \log\left(1+\SNR_2+\INR_2\right) + \log\left(1+\frac{\SNR_1}{1+\INR_2}\right) + \C_{12}.
\end{align}

\item Sixth, 
\begin{align}
&I\lp x_1; y_1| x_{1c},x_{2c}\rp + I\lp x_{2c};y_1|x_1\rp + I\lp x_{1c},x_2; y_2| x_{2c}\rp + \C_{12}\\
&= \log\lp1+\frac{\SNR_{1p}}{1+\INR_{1p}}\rp + \log\lp\frac{1+\INR_1}{1+\INR_{1p}}\rp + \log\lp\frac{1+\SNR_{2p}+\INR_2}{1+\INR_{2p}}\rp + \C_{12}\\
&\ge \log\lp1+\frac{\SNR_{1p}}{1+\INR_{1p}}\rp + \log\lp\frac{1+\SNR_2+\INR_2}{(1+\INR_{1p})(1+\INR_{2p})}\rp + \C_{12},
\end{align}
which is within $3$ bits to the upper bound
\begin{align}
R_1+R_2 &\le \log\left(1+\SNR_2+\INR_2\right) + \log\left(1+\frac{\SNR_1}{1+\INR_2}\right) + \C_{12}.
\end{align}
\end{itemize}

{\flushleft (4)} $2R_1+R_2$ bound: The bound
\begin{align}
&I\lp x_1,x_{2c}; y_1\rp + I\lp x_1; y_1| x_{1c},x_{2c}\rp + I\lp x_{1c},x_2; y_2| x_{2c}\rp + \C_{12} + \lp\C_{21}-\xi_1\rp^+\\
&= \log\lp\frac{1+\SNR_1+\INR_1}{1+\INR_{1p}}\rp + \log\lp1+\frac{\SNR_{1p}}{1+\INR_{1p}}\rp + \log\lp\frac{1+\SNR_{2p}+\INR_2}{1+\INR_{2p}}\rp\\
&\quad + \C_{12} + \lp\C_{21}-\xi_1\rp^+,
\end{align}
which is within $3+\log3 = \log24$ bits to the upper bound
\begin{align}
2R_1 + R_2 & \le \log\left(1+\INR_2+\frac{\SNR_2}{1+\INR_1}\right) + \log\left(1+\frac{\SNR_1}{1+\INR_2}\right)\\&\quad + \log\lp 1+ \SNR_1+\INR_1\rp + \C_{21}+\C_{12}.
\end{align}

{\flushleft (5)} $R_1+2R_2$ bound:
We have six bounds for $R_1+2R_2$, investigated as follows:
\begin{itemize}
\item First,
\begin{align}
&I\lp x_1,x_{2c}; y_1| x_{1c}\rp + I\lp x_{1c},x_2; y_2\rp + I\lp x_2; y_2| x_{1c},x_{2c}\rp + \C_{12} + \lp\C_{21}-\xi_1\rp^+\\
&= \log\lp\frac{1+\SNR_{1p}+\INR_1}{1+\INR_{1p}}\rp + \log\lp\frac{1+\SNR_2+\INR_2}{1+\INR_{2p}}\rp + \log\lp1+\frac{\SNR_{2p}}{1+\INR_{2p}}\rp\\
&\quad + \C_{12} + \lp\C_{21}-\xi_1\rp^+,
\end{align}
which is within $3+\log3 = \log24$ bits to the upper bound
\begin{align}
R_1 + 2R_2 & \le \log\left(1+\INR_1+\frac{\SNR_1}{1+\INR_2}\right) + \log\left(1+\frac{\SNR_2}{1+\INR_1}\right)\\&\quad + \log\lp 1+ \SNR_2+\INR_2\rp + \C_{12}+\C_{21}.
\end{align}

\item Second,
\begin{align}
&I\lp x_1,x_{2c}; y_1| x_{1c}\rp + I\lp x_{2c};y_1|x_1\rp + I\lp x_{1c},x_2; y_2| x_{2c}\rp + I\lp x_2; y_2| x_{1c},x_{2c}\rp\\&\quad + \C_{12} + \lp\C_{21}-\xi_1\rp^+\\
&= \log\lp\frac{1+\SNR_{1p}+\INR_1}{1+\INR_{1p}}\rp + \log\lp\frac{1+\INR_1}{1+\INR_{1p}}\rp + \log\lp\frac{1+\SNR_{2p}+\INR_2}{1+\INR_{2p}}\rp\\&\quad + \log\lp1+\frac{\SNR_{2p}}{1+\INR_{2p}}\rp + \C_{12} + \lp\C_{21}-\xi_1\rp^+\\
&\ge \log\lp\frac{1+\SNR_{1p}+\INR_1}{1+\INR_{1p}}\rp + \log\lp\frac{1+\SNR_2+\INR_2}{(1+\INR_{1p})(1+\INR_{2p})}\rp\\&\quad + \log\lp1+\frac{\SNR_{2p}}{1+\INR_{2p}}\rp + \C_{12} + \lp\C_{21}-\xi_1\rp^+,
\end{align}
which is within $4+\log3 = \log48$ bits to the upper bound
\begin{align}
R_1 + 2R_2 & \le \log\left(1+\INR_1+\frac{\SNR_1}{1+\INR_2}\right) + \log\left(1+\frac{\SNR_2}{1+\INR_1}\right)\\&\quad + \log\lp 1+ \SNR_2+\INR_2\rp + \C_{12}+\C_{21}.
\end{align}

\item Third, 
\begin{align}
&I\lp x_1,x_{2c}; y_1,\what{y}_2| x_{1c}\rp + I\lp x_{1c},x_2; y_2\rp + I\lp x_2; y_2| x_{1c},x_{2c}\rp + \C_{12}\\
&= \log\lp\frac{(1+\Delta_2)(1+\SNR_{1p}+\INR_1) + \SNR_2 + \INR_{2p} + |h_{11}h_{22}-h_{12}h_{21}|^2Q_{1p}}{(1+\Delta_2)(1+\INR_{1p})+\SNR_{2p}}\rp\\
&\quad + \log\lp\frac{1+\SNR_2+\INR_2}{1+\INR_{2p}}\rp + \log\lp1+\frac{\SNR_{2p}}{1+\INR_{2p}}\rp + \C_{12}\\
&\ge \log\lp\frac{1+\SNR_{1p}+\INR_1 + \SNR_2 + \INR_{2p} + |h_{11}h_{22}-h_{12}h_{21}|^2Q_{1p}}{5\Delta_2}\rp\\
&\quad + \log\lp\frac{1+\SNR_2+\INR_2}{1+\INR_{2p}}\rp + \log\lp1+\SNR_{2p}\rp + \C_{12} - 1\\
&\ge \log\lp 1+\SNR_{1p}+\INR_1 + \SNR_2 + \INR_{2p} + |h_{11}h_{22}-h_{12}h_{21}|^2Q_{1p}\rp\\
&\quad + \log\lp\frac{1+\SNR_2+\INR_2}{1+\INR_{2p}}\rp + \C_{12} - 1 - \log5,
\end{align}
which is within $2+\log5 = \log20$ bits to the upper bound
\begin{align}
&\log\lp 1+ \frac{\SNR_1}{1+\INR_2} + \INR_1 + \SNR_2 + \frac{\INR_2}{1+\INR_2} + \frac{|h_{11}h_{22} - h_{12}h_{21}|^2}{1+\INR_2}\rp\\
&\quad + \log\lp1+\SNR_2+\INR_2\rp + \C_{12}.
\end{align}


\item Fourth, 
\begin{align}
&I\lp x_1,x_{2c}; y_1,\what{y}_2| x_{1c}\rp + I\lp x_{2c};y_1|x_1\rp + I\lp x_{1c},x_2; y_2| x_{2c}\rp + I\lp x_2; y_2| x_{1c},x_{2c}\rp + \C_{12}\\
&= \log\lp\frac{(1+\Delta_2)(1+\SNR_{1p}+\INR_1) + \SNR_2 + \INR_{2p} + |h_{11}h_{22}-h_{12}h_{21}|^2Q_{1p}}{(1+\Delta_2)(1+\INR_{1p})+\SNR_{2p}}\rp\\
&\quad +\log\lp\frac{1+\INR_1}{1+\INR_{1p}}\rp + \log\lp\frac{1+\SNR_{2p}+\INR_2}{1+\INR_{2p}}\rp + \log\lp1+\frac{\SNR_{2p}}{1+\INR_{2p}}\rp + \C_{12}\\
&\ge \log\lp\frac{1+\SNR_{1p}+\INR_1 + \SNR_2 + \INR_{2p} + |h_{11}h_{22}-h_{12}h_{21}|^2Q_{1p}}{5\Delta_2}\rp\\
&\quad + \log\lp\frac{1+\SNR_2+\INR_2}{(1+\INR_{1p})(1+\INR_{2p})}\rp + \log\lp1+\SNR_{2p}\rp + \C_{12} - 1\\
&\ge \log\lp 1+\SNR_{1p}+\INR_1 + \SNR_2 + \INR_{2p} + |h_{11}h_{22}-h_{12}h_{21}|^2Q_{1p}\rp\\
&\quad + \log\lp\frac{1+\SNR_2+\INR_2}{(1+\INR_{1p})(1+\INR_{2p})}\rp + \C_{12} - 1 - \log5,
\end{align}
which is within $3+\log5 = \log40$ bits to the upper bound
\begin{align}
&\log\lp 1+ \frac{\SNR_1}{1+\INR_2} + \INR_1 + \SNR_2 + \frac{\INR_2}{1+\INR_2} + \frac{|h_{11}h_{22} - h_{12}h_{21}|^2}{1+\INR_2}\rp\\
&\quad + \log\lp1+\SNR_2+\INR_2\rp + \C_{12}.
\end{align}


\end{itemize}

Therefore, we see that the bounds in $\mscr{R}_{2\ra1\ra2}$ except \eqref{eq_Trouble} satisfies:
\begin{itemize}
\item $R_1$ bound is within $2$ bits to outer bounds;
\item $R_2$ bound is within $2$ bits to outer bounds;
\item $R_1+R_2$ bound is within $\log 12$ bits to outer bounds;
\item $2R_1+R_2$ bound is within $\log 24$ bits to outer bounds;
\item $R_1+2R_2$ bound is within $\log 48$ bits to outer bounds.
\end{itemize}
\end{proof}

\subsection{Proof of Claim \ref{claim_MixedGap}}
\begin{proof}
(Keep in mind $\Delta_2 = 1$)
{\flushleft (1)} $R_1$ bound:
We have two bounds. First, $I\lp x_1; y_1| x_2\rp = \log\lp 1+\SNR_1\rp$, which is within 1 bit to the upper bound $R_1\le \log\lp 1+\SNR_1+\INR_2\rp$. Second,
\begin{align}
&I\lp x_1; y_1| x_{1c},x_2\rp + I\lp x_{1c}; y_2| x_2\rp + \C_{12}\\
&= \log\lp 1+\SNR_{1p}\rp + \log\lp \frac{1+\INR_2}{1+\INR_{2p}}\rp +\C_{12}\\
&\ge \log\lp \frac{1+\SNR_1+\INR_2}{1+\INR_{2p}}\rp + \C_{12}.
\end{align}
Hence, if the second bound is active, it is within 1 bit to the upper bound $\log\lp 1+\SNR_1+\INR_2\rp$.

{\flushleft (2)} $R_2$ bound:
We have two bounds. First, $I\lp x_1; y_1| x_1\rp = \log\lp1+\INR_1\rp$, which is within 1 bit to the upper bound $R_2 \le \log\lp1+\SNR_2+\INR_1\rp$. Second, $I\lp x_2; y_2| x_{1c}\rp + \C_{12} = \log\lp \frac{1+\SNR_2+\INR_{2p}}{1+\INR_{2p}}\rp + \C_{12}$, which is within 1 bit to the upper bound $R_2 \le \log\lp1+\SNR_2\rp + \C_{12}$.

{\flushleft (3)} $R_1+R_2$ bound: We have five bounds, investigated as follows:
\begin{itemize}
\item First, 
\begin{align}
I\lp x_1,x_2; y_1\rp + \lp\C_{21}-\xi_1\rp^+ = \log\lp 1+\SNR_1+\INR_1\rp + \lp\C_{21}-\xi_1\rp^+,
\end{align}
which is within $1+\xi_1= 2$ bits to the upper bound
\begin{align}
R_1+R_2 &\le \log\left(1+\SNR_1+\INR_1\right) + \log\left(1+\frac{\SNR_2}{1+\INR_1}\right) + \C_{21}.
\end{align}

\item Second,
\begin{align}
&I\lp x_1,x_2; y_1,\what{y}_2\rp\\
&= \log\lp \frac{2(1+\SNR_1+\INR_1) + \SNR_2 + \INR_2 + |h_{11}h_{22}-h_{12}h_{21}|^2}{2}\rp
\end{align}
which is within $1$ bit to the upper bound
\begin{align}
R_1+R_2 &\le \log\lp 1+\SNR_1+\SNR_2+\INR_1+\INR_2 + |h_{11}h_{22} - h_{12}h_{21}|^2 \rp.
\end{align}

\item Third, 
\begin{align}
&I\lp x_1; y_1|x_{1c},x_2\rp + I\lp x_{1c},x_2; y_2\rp + \C_{12}\\
&= \log\lp1+\SNR_{1p}\rp + \log\lp\frac{1+\SNR_2+\INR_2}{1+\INR_{2p}}\rp + \C_{12}
\end{align}
which is within $1$ bit to the upper bound
\begin{align}
R_1+R_2 &\le \log\left(1+\SNR_2+\INR_2\right) + \log\left(1+\frac{\SNR_1}{1+\INR_2}\right) + \C_{12}.
\end{align}

\item Fourth, 
\begin{align}
&I\lp x_1,x_2; y_1| x_{1c}\rp + I\lp x_{1c}; y_2| x_2\rp + \C_{12} + \lp\C_{21}-\xi_1\rp^+\\
&= \log\lp 1+\SNR_{1p}+\INR_1\rp + \log\lp\frac{1+\INR_2}{1+\INR_{2p}}\rp + \C_{12} + \lp\C_{21}-\xi_1\rp^+,
\end{align}
which is within $2+\xi_1= 3$ bits to the upper bound
\begin{align}
R_1+R_2 &\le \log\left(1+\INR_1+\frac{\SNR_1}{1+\INR_2}\right) + \log\left(1+\INR_2+\frac{\SNR_2}{1+\INR_1}\right)\\&\quad + \C_{21}+\C_{12}.
\end{align}

\item Fifth, 
\begin{align}
&I\lp x_1,x_2; y_1,\what{y}_2| x_{1c}\rp + I\lp x_{1c}; y_2| x_2\rp + \C_{12}\\
&=I\lp x_{2};y_1,\what{y}_2 | x_{1c}\rp + I\lp x_1;y_1,\what{y}_2 | x_{1c},x_{2}\rp + I\lp x_{1c}; y_2| x_2\rp + \C_{12}\\
&\ge I\lp x_{2};y_1,\what{y}_2 | x_{1c}\rp + I\lp x_1;y_1 | x_{1c},x_{2}\rp + I\lp x_{1c}; y_2| x_2\rp + \C_{12}\\
&= \log\lp \frac{2(1+\SNR_{1p}+\INR_1) + \SNR_2+\INR_{2p} + |h_{11}h_{22} - h_{12}h_{21}|^2Q_{1p}}{2(1+\SNR_{1p})+\INR_{2p}} \rp\\
&\quad + \log\lp1+\SNR_{1p}\rp + \log\lp\frac{1+\INR_2}{1+\INR_{2p}}\rp + \C_{12}\\
&\ge \log\lp \frac{1+\SNR_{1p}+\INR_1 + \SNR_2+\INR_{2p} + |h_{11}h_{22} - h_{12}h_{21}|^2Q_{1p}}{3(1+\SNR_{1p})} \rp\\
&\quad + \log\lp1+\SNR_{1p}\rp + \log\lp\frac{1+\INR_2}{1+\INR_{2p}}\rp + \C_{12}\\
&\ge \log\lp 1+\SNR_{1}+\INR_1 + \SNR_2+\INR_{2} + |h_{11}h_{22} - h_{12}h_{21}|^2 \rp\\
&\quad + \log\lp\frac{1}{1+\INR_{2p}}\rp + \C_{12} - \log3.
\end{align}
Hence, if this bound is active, it is within $1+\log3=\log6$ bits to the upper bound
\begin{align}
R_1+R_2 &\le \log\lp 1+\SNR_1+\SNR_2+\INR_1+\INR_2 + |h_{11}h_{22} - h_{12}h_{21}|^2 \rp.
\end{align}
\end{itemize}

{\flushleft (4)} $R_1+2R_2$ bound: We have two bounds. First,
\begin{align}
&I\lp x_1,x_2; y_1| x_{1c}\rp + I\lp x_{1c},x_2; y_2\rp + \C_{12} + \lp\C_{21}-\xi_1\rp^+\\
&= \log\lp1+\SNR_{1p}+\INR_1\rp + \log\lp\frac{1+\SNR_2+\INR_2}{1+\INR_{2p}}\rp + \C_{12} + \lp\C_{21}-\xi_1\rp^+,
\end{align}
which is within $2+\xi_1 = 3$ bits to the upper bound
\begin{align}
R_1 + 2R_2 & \le \log\left(1+\INR_1+\frac{\SNR_1}{1+\INR_2}\right) + \log\left(1+\frac{\SNR_2}{1+\INR_1}\right)\\&\quad + \log\lp 1+ \SNR_2+\INR_2\rp + \C_{12}+\C_{21}.
\end{align}

Second, 
\begin{align}
&I\lp x_1,x_2; y_1,\what{y}_2| x_{1c}\rp + I\lp x_{1c},x_2; y_2\rp + \C_{12}\\
&=\log\lp \frac{2(1+\SNR_{1p}+\INR_1) + \SNR_2+\INR_{2p} + |h_{11}h_{22} - h_{12}h_{21}|^2Q_{1p}}{2} \rp\\
&\quad + \log\lp\frac{1+\SNR_2+\INR_2}{1+\INR_{2p}}\rp + \C_{12},
\end{align}
which is within $2$ bits to the upper bound
\begin{align}
R_1+2R_2 &\le \log\lp 1+ \frac{\SNR_1}{1+\INR_2} + \INR_1 + \SNR_2 + \frac{\INR_2}{1+\INR_2} + \frac{|h_{11}h_{22} - h_{12}h_{21}|^2}{1+\INR_2}\rp\\
&\quad + \log\lp1+\SNR_2+\INR_2\rp + \C_{12}.
\end{align}

Therefore, we see that the bounds in $\mscr{R}_{2\ra1\ra2}$ satisfies:
\begin{itemize}
\item $R_1$ bound is within $1$ bit to outer bounds;
\item $R_2$ bound is within $1$ bit to outer bounds;
\item $R_1+R_2$ bound is within $3$ bits to outer bounds;
\item $R_1+2R_2$ bound is within $3$ bits to outer bounds.
\end{itemize}
\end{proof}

\subsection{Proof of Claim \ref{claim_StrongGap}}
\begin{proof}
(Keep in mind that $\Delta_1 = \Delta_2 = 1$)
{\flushleft (1)} $R_1$ bound: We have four bounds. First,
\begin{align}
I\lp x_{1};y_1|x_{2}\rp + (\C_{21} - \xi_1)^+ = \log\lp1+\SNR_1\rp + (\C_{21} - \xi_1)^+,
\end{align}
which is within $\xi_1 = 1$ bit to the upper bound $\log\lp1+\SNR_1\rp + \C_{21}$. Second, 
\begin{align}
&I\lp x_{1};y_2|x_{2}\rp + (\C_{12} - \xi_2)^+\\
&= \log\lp1+\INR_2\rp + (\C_{12} - \xi_2)^+ \ge \log\lp1+\SNR_1+\INR_2\rp -1.
\end{align}
Hence if this bound is active, it is within $1$ bit to the upper bound $\log\lp1+\SNR_1+\INR_2\rp$. Finally,
\begin{align}
&I\lp x_{1};y_1,\what{y}_2 | x_{2}\rp = \log\lp\frac{2+2\SNR_1+\INR_2}{2}\rp\\
&I\lp x_{1};y_2,\what{y}_1 | x_{2}\rp = \log\lp\frac{2+\SNR_1+2\INR_2}{2}\rp,
\end{align} 
which are both within $1$ bit to the upper bound $\log\lp1+\SNR_1+\INR_2\rp$.

{\flushleft (2)} $R_2$ bound: By symmetry we have the same gap result as (1).

{\flushleft (3)} $R_1+R_2$ bound:
We have four bounds. First,
\begin{align}
I\lp x_{1},x_{2};y_1 \rp + (\C_{21} - \xi_1)^+ = \log\lp1+\SNR_1+\INR_1\rp + (\C_{21} - \xi_1)^+,
\end{align}
which is within $1+\xi_1=2$ bits to the upper bound 
\begin{align}R_1+R_2\le\log\lp1+\SNR_1+\INR_1\rp + \log\lp1+\frac{\SNR_2}{1+\INR_1}\rp+ \C_{21}.\end{align}
Second, 
\begin{align}
I\lp x_{2},x_{1};y_2 \rp + (\C_{12} - \xi_2)^+ = \log\lp1+\SNR_2+\INR_2\rp + (\C_{12} - \xi_2)^+,
\end{align}
which is within $1+\xi_2=2$ bits to the upper bound 
\begin{align}R_1+R_2\le\log\lp1+\SNR_2+\INR_2\rp + \log\lp1+\frac{\SNR_1}{1+\INR_2}\rp+ \C_{12}.\end{align}
Finally,
\begin{align}
&I\lp x_{1},x_2;y_1,\what{y}_2\rp = \log\lp\frac{2(1+\SNR_1+\INR_1)+\SNR_2+\INR_2+|h_{11}h_{22}-h_{12}h_{21}|^2}{2}\rp\\
&I\lp x_2,x_{1};y_2,\what{y}_1\rp = \log\lp\frac{2(1+\SNR_2+\INR_2)+\SNR_1+\INR_1+|h_{11}h_{22}-h_{12}h_{21}|^2}{2}\rp,
\end{align} 
which are both within $1$ bit to the upper bound
\begin{align}
R_1+R_2 \le \log\lp 1+\SNR_1+\INR_1+\SNR_2+\INR_2+|h_{11}h_{22}-h_{12}h_{21}|^2 \rp.
\end{align}

Therefore, we see that the bounds in $\mscr{R}_{\OR}$ satisfies:
\begin{itemize}
\item $R_1$ bound is within $1$ bit to outer bounds;
\item $R_2$ bound is within $1$ bit to outer bounds;
\item $R_1+R_2$ bound is within $2$ bits to outer bounds.
\end{itemize}

\end{proof}


\section{Proof of Theorem \ref{thm_ORSymm}}\label{app_PfORSymm}
From Section \ref{subsec_StrongAchieve}, we have shown that when $\SNR\le \INR$, 
\begin{align}
R_{\sym,\OR} \le C_{\sym} \le \overline{C}_{\sym} \le R_{\sym,\OR}+1.
\end{align}

Hence we focus on the case $\SNR > \INR$ in the rest of the proof. 

By symmetry and by Theorem \ref{thm_ErrProb}, if $R_{\sym,\OR}\ge 0$ satisfies the following, it is achievable:
\begin{align}
R_{\sym,\OR} &\le \min \lbp I\lp x_{2c},x_{1};y_1|x_{1c}\rp + (\C_{21} - \xi_1)^+, I\lp x_{2c},x_{1};y_1,\what{y}_2 | x_{1c}\rp \rbp\\
R_{\sym,\OR} &\le \min \lbp I\lp x_{1};y_1|x_{2c}\rp + (\C_{21} - \xi_1)^+, I\lp x_{1};y_1,\what{y}_2 | x_{2c}\rp \rbp\\ 
2R_{\sym,\OR} &\le \min \lbp I\lp x_{1},x_{2c};y_1 \rp + (\C_{21} - \xi_1)^+, I\lp x_{1},x_{2c};y_1,\what{y}_2 \rp \rbp\\
&\quad + \min \lbp I\lp x_{1};y_1|x_{1c},x_{2c}\rp + (\C_{21} - \xi_1)^+, I\lp x_{1};y_1,\what{y}_2 | x_{1c},x_{2c}\rp \rbp
\end{align}

Note that since
\begin{align}
&I\lp x_{1};y_1|x_{1c},x_{2c}\rp \le I\lp x_{1};y_1,\what{y}_2 | x_{1c},x_{2c}\rp \le I\lp x_{1};y_1|x_{1c},x_{2c}\rp + \rm{constant},\\
&I\lp x_{1};y_1|x_{2c}\rp \le I\lp x_{1};y_1,\what{y}_2 | x_{2c}\rp \le I\lp x_{1};y_1|x_{2c}\rp + \rm{constant},
\end{align}
a sufficient condition for achievable $R_{\sym,\OR}$ is
\begin{align}
R_{\sym,\OR} &\le \min \lbp I\lp x_{2c},x_{1};y_1|x_{1c}\rp + (\C_{21} - \xi_1)^+, I\lp x_{2c},x_{1};y_1,\what{y}_2 | x_{1c}\rp \rbp\\
R_{\sym,\OR} &\le I\lp x_{1};y_1|x_{2c}\rp\\ 
R_{\sym,\OR} &\le \frac{1}{2}\min \lbp I\lp x_{1},x_{2c};y_1 \rp + (\C_{21} - \xi_1)^+, I\lp x_{1},x_{2c};y_1,\what{y}_2 \rp \rbp\\&\quad + \frac{1}{2}I\lp x_{1};y_1|x_{1c},x_{2c}\rp
\end{align}

{\flushleft (1) $I\lp x_{2c},x_{1};y_1|x_{1c}\rp + (\C_{21} - \xi_1)^+$:}
\begin{align}
I\lp x_{2c},x_{1};y_1|x_{1c}\rp + (\C_{21} - \xi_1)^+ = \log\lp \frac{1+\SNR_p+\INR}{1+\INR_p}\rp +(\C - \xi)^+
\end{align}

The gap to the outer bound $\log\left(1+\INR+\frac{\SNR}{1+\INR}\right) + \C$:
\begin{align}
\rm{gap} &\le \log\lp1+\INR_p\rp +\xi \le 1+\log3. 
\end{align}

{\flushleft (2) $I\lp x_{2c},x_{1};y_1,\what{y}_2 | x_{1c}\rp$:}
\begin{align}
I\lp x_{2c},x_{1};y_1,\what{y}_2 | x_{1c}\rp &= I\lp x_{2c};y_1,\what{y}_2 | x_{1c}\rp + I\lp x_{1};y_1,\what{y}_2 | x_{2c},x_{1c}\rp\\
&\ge I\lp x_{2c};\what{y}_2 | x_{1c}\rp + I\lp x_{1};y_1 | x_{2c},x_{1c}\rp\\
&= \log\lp \frac{1+\Delta+\SNR+\INR_p}{1+\Delta+\SNR_p+\INR_p}\rp + \log\lp \frac{1+\SNR_p+\INR_p}{1+\INR_p}\rp\\
&\overset{\aaaa}{\ge} \log\lp \frac{1+\SNR}{2+2\SNR_p}\rp + \log\lp\frac{1+\SNR_p}{1+\INR_p}\rp\\
&= \log\lp \frac{1+\SNR}{1+\INR_p}\rp - 1,
\end{align}
where (a) is due to $\Delta = \max\lbp 1,\SNR_p\rbp$, $\SNR_p >\INR_p$, and hence
\begin{align}
1+\Delta+\SNR_p+\INR_p = \lbp \begin{array}{ll}
1+2\SNR_p+\INR_p \le 2+2\SNR_p &\textrm{if }\Delta = \SNR_p \ge 1\\
2 + \SNR_p+\INR_p \le 2+2\SNR_p &\textrm{if }\Delta = 1 > \SNR_p
\end{array}\right.
\end{align}

Therefore, the gap to the outer bound $\log\lp1+\SNR+\INR\rp$:
\begin{align}
\rm{gap} &\le 1+\log\lp \frac{1+\SNR+\INR}{1+\SNR}\rp + \log\lp 1+\INR_p\rp\\
&\le 1+\log\lp \frac{2+2\SNR}{1+\SNR}\rp + \log\lp 1+1\rp =3.
\end{align}

{\flushleft (3) $I\lp x_{1};y_1|x_{2c}\rp$:}
\begin{align}
I\lp x_{1};y_1|x_{2c}\rp = \log\lp 1+\SNR+\INR_p\rp - \log\lp 1+\INR_p\rp.
\end{align}

The gap to the outer bound $\log\lp1+\SNR+\INR\rp$:
\begin{align}
\rm{gap} &\le \log\lp \frac{1+\SNR+\INR}{1+\SNR+\INR_p}\rp + \log\lp 1+\INR_p\rp\\
&\le \log\lp \frac{2+2\SNR}{1+\SNR}\rp + \log\lp 1+1\rp = 2.
\end{align}

{\flushleft (4) $\frac{1}{2}I\lp x_{1},x_{2c};y_1 \rp + \frac{1}{2}(\C_{21} - \xi_1)^+ + \frac{1}{2}I\lp x_{1};y_1|x_{1c},x_{2c}\rp$:}
\begin{align}
&\frac{1}{2}I\lp x_{1},x_{2c};y_1 \rp + \frac{1}{2}(\C_{21} - \xi_1)^+ +  \frac{1}{2}I\lp x_{1};y_1|x_{1c},x_{2c}\rp\\
&= \frac{1}{2}\log\lp1+\SNR+\INR\rp + \frac{1}{2}(\C - \xi)^+ + \frac{1}{2}\log\lp1+\SNR_p+\INR_p\rp - \log\lp 1+\INR_p\rp
\end{align}

The gap to the outer bound $\frac{1}{2}\log\left(1+\SNR+\INR\right) + \frac{1}{2}\log\left(1+\frac{\SNR}{1+\INR}\right) + \frac{1}{2}\C$:
\begin{align}
\rm{gap} \le \frac{1}{2}\xi + \log\lp1+\INR_p\rp \le \frac{1}{2}\log3+1.
\end{align}

{\flushleft (5) $\frac{1}{2}I\lp x_{1},x_{2c};y_1,\what{y}_2 \rp+\frac{1}{2}I\lp x_{1};y_1|x_{1c},x_{2c}\rp$:}
\begin{align}
&\frac{1}{2}I\lp x_{1},x_{2c};y_1,\what{y}_2 \rp+\frac{1}{2}I\lp x_{1};y_1|x_{1c},x_{2c}\rp\\
&= \frac{1}{2}\log\lp \frac{\Delta\lp1+\SNR+\INR\rp+1+2\SNR+2\INR+|h_{11}h_{22}-h_{12}h_{21}|^2}{\Delta\lp1+\INR_p\rp+1+\SNR_p+\INR_p}\rp\\
&\quad + \frac{1}{2}\log\lp \frac{1+\SNR_p+\INR_p}{1+\INR_p}\rp\\
&\ge \frac{1}{2}\log\lp \frac{1+2\SNR+2\INR+|h_{11}h_{22}-h_{12}h_{21}|^2}{5\Delta}\rp + \frac{1}{2}\log\lp \frac{\Delta}{1+\INR_p}\rp\\
&= \frac{1}{2}\log\lp 1+2\SNR+2\INR+|h_{11}h_{22}-h_{12}h_{21}|^2\rp- \frac{1}{2}\log\lp 1+\INR_p\rp - \frac{1}{2}\log5
\end{align}

Therefore, the gap to the outer bound $\frac{1}{2}\log\lp 1+2\SNR+2\INR + |h_{11}h_{22} - h_{12}h_{21}|^2 \rp$:
\begin{align}
\rm{gap}\le \frac{1}{2}\log\lp 1+\INR_p\rp + \frac{1}{2}\log5 \le \frac{1+\log5}{2}.
\end{align}

From (1) - (5), we conclude that when $\SNR > \INR$,
\begin{align}
R_{\sym,\OR} \le C_{\sym} \le \overline{C}_{\sym} \le R_{\sym,\OR}+3.
\end{align}

This completes the proof.

\section{Proof of Lemma \ref{prop_ConvDof}}\label{app_PfConvDof}
\begin{proof}
From Corollary \ref{cor_SymBd} we see that except the term
\begin{align}
V:=\frac{1}{2}\log\lp 1+2\SNR+2\INR + |h_{11}h_{22} - h_{12}h_{21}|^2 \rp,
\end{align}
all terms scaled by $\log\SNR$ converges {\it everywhere} as $\SNR\rightarrow\infty$ with $\alpha,\kappa$ fixed. Note that
\begin{align}
&|h_{11}h_{22} - h_{12}h_{21}|^2 = |g_{11}e^{j\Theta_{11}}g_{22}e^{j\Theta_{22}}-g_{12}e^{j\Theta_{12}}g_{21}e^{j\Theta_{21}}|^2\\
&= \Big[g_{11}g_{22}\cos\lp\Theta_{11}+\Theta_{22}\rp - g_{12}g_{21}\cos\lp\Theta_{12}+\Theta_{21}\rp\Big]^2\\
&\quad + \Big[g_{11}g_{22}\sin\lp\Theta_{11}+\Theta_{22}\rp - g_{12}g_{21}\sin\lp\Theta_{12}+\Theta_{21}\rp\Big]^2\\
&= g_{11}^2g_{22}^2+g_{12}^2g_{21}^2 - 2g_{11}g_{22}g_{12}g_{21}\cos\lp \Theta_{11}+\Theta_{22} - \Theta_{12}-\Theta_{21}\rp\\
&= \SNR^2+\INR^2 - 2(\cos\Theta)\SNR\INR,
\end{align}
where $\Theta=\Theta_{11}+\Theta_{22} - \Theta_{12}-\Theta_{21}\mod 2\pi$. Obviously $\Theta$ is uniformly distributed over $[0,2\pi]$. Now, consider the limit
\begin{align}
L(\alpha,\kappa):=\lim_{\begin{subarray}{c} \mathrm{fix}\ \alpha,\kappa\\ \SNR\rightarrow\infty\end{subarray}}\frac{V}{\log\SNR}.
\end{align}

We have the following upper and lower bounds for $V$ due to the fact that $\big| |h_{11}||h_{22}| - |h_{12}||h_{21}|\big| \le |h_{11}h_{22} - h_{12}h_{21}| \le |h_{11}||h_{22}| + |h_{12}||h_{21}|$:
\begin{align}
V &\ge \frac{1}{2}\log\lp 1+2\SNR+2\INR + (\SNR-\INR)^2 \rp;\\
V &\le \frac{1}{2}\log\lp 1+2\SNR+2\INR + (\SNR+\INR)^2 \rp.
\end{align}
Hence, when $\alpha<1$, taking limits at both sides yields $1\le L(\alpha,\kappa) \le 1$ and implies $L(\alpha,\kappa)=1$. Similarly, when $\alpha>1$, taking limits at both sides yields $\alpha \le L(\alpha,\kappa) \le \alpha$ and implies $L(\alpha,\kappa)=\alpha$. When $\alpha=1$, note that 
\begin{align}
V &= \frac{1}{2}\log\lp 1+2\SNR+2\INR + \SNR^2+\INR^2 - 2(\cos\Theta)\SNR\INR \rp\\
&= \frac{1}{2}\log\lp (1+\SNR+\INR)^2 - 4\cos^2\frac{\Theta}{2}\SNR\INR \rp,
\end{align}
and therefore $L(\alpha,\kappa)=1$ if $\Theta \ne 0, 2\pi$. Since the event $\{\Theta=0,2\pi\}$ is of zero measure, the limit $L(\alpha,\kappa)$ exists almost surely.
\end{proof}

\end{document}